\documentclass[sigconf]{aamas}

\usepackage{tikz}
\usetikzlibrary{calc}
\usetikzlibrary{decorations.pathreplacing}
\usetikzlibrary{angles}
\usetikzlibrary{graphs}
\usetikzlibrary{through}
\usetikzlibrary{intersections}
\usetikzlibrary{math}
\usepackage{complexity}
\usepackage{cleveref}
\usepackage{minibox}
\usepackage[obeyFinal]{todonotes}
\usepackage{appendix}
\usepackage{booktabs}
\usepackage{fontspec}

\usepackage{balance} 

\doi{TNYC5178}

\makeatletter
\gdef\@copyrightpermission{
  \begin{minipage}{0.2\columnwidth}
   \href{https://creativecommons.org/licenses/by/4.0/}{\includegraphics[width=0.90\textwidth]{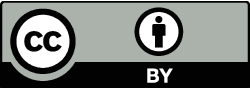}}
  \end{minipage}\hfill
  \begin{minipage}{0.8\columnwidth}
   \href{https://creativecommons.org/licenses/by/4.0/}{This work is licensed under a Creative Commons Attribution International 4.0 License.}
  \end{minipage}
  \vspace{5pt}
}
\makeatother

\setcopyright{ifaamas}
\acmConference[AAMAS '26]{Proc.\@ of the 25th International Conference
on Autonomous Agents and Multiagent Systems (AAMAS 2026)}{May 25 -- 29, 2026}
{Paphos, Cyprus}{C.~Amato, L.~Dennis, V.~Mascardi, J.~Thangarajah (eds.)}
\copyrightyear{2026}
\acmYear{2026}
\acmDOI{}
\acmPrice{}
\acmISBN{}

\title{Practical approach to $2$-Euclidean Preferences}

\author{Michal Dvořák}

\affiliation{
	\institution{Czech Technical University in Prague}
	\city{Prague}
	\country{Czech Republic}
}
\author{Jan Pokorný}

\affiliation{
	\institution{Czech Technical University in Prague}
	\city{Prague}
	\country{Czech Republic}
}
\author{Dušan Knop}

\affiliation{
	\institution{Czech Technical University in Prague}
	\city{Prague}
	\country{Czech Republic}
}
\author{Martin Slávik}

\affiliation{
	\institution{Czech Technical University in Prague}
	\city{Prague}
	\country{Czech Republic}
}

\begin{abstract}
	An election is a pair $(C,V)$ of candidates and voters. Each vote is a ranking (permutation) of the candidates. An election is $d$-Euclidean if there is an embedding of both candidates and voters into $\mathbb{R}^d$ such that voter $v$ prefers candidate $a$ over $b$ if and only if $a$ is closer to $v$ than $b$ is to $v$ in the embedding. For $d\geq 2$ the problem of deciding whether $(C,V)$ is $d$-Euclidean is $\exists \mathbb{R}$-complete.
	In this paper, we propose a practical approach to recognizing and refuting $2$-Euclidean preferences. We design a new class of forbidden substructures that works very well on practical instances. We utilize the framework of integer linear programming (ILP) and quadratically constrained programming (QCP). We also introduce reduction rules that simplify many real-world instances significantly. Our approach beats the previous algorithm of Escoffier, Spanjaard and Tydrichová~[Algorithmic Recognition of 2-Euclidean Preferences, ECAI 2023] both in number of resolved instances and the running time. In particular, we were able to lower the number of unresolved PrefLib instances from $343$ to $60$. Moreover, $98.7\%$ of PrefLib instances are resolved in under $1$ second using our approach.
\end{abstract}


\usepackage{etoolbox}
\usepackage{placeins}
\def\ShortVersion{1}
\ifdefined\ShortVersion
\newcommand{\sv}[1]{#1}
\newcommand{\lv}[1]{}
\newcommand{\appendixText}{}
\newcommand{\toappendix}[1]{\gappto{\appendixText}{{#1}}}
\else
\newcommand{\sv}[1]{}
\newcommand{\lv}[1]{#1}
\newcommand{\appendixText}{}
\newcommand{\toappendix}[1]{#1}
\fi
\newcommand{\dtoappendix}[1]{#1\gappto{\appendixText}{{#1}}} 

\definecolor{darkgreen}{RGB}{0, 150, 0} 

\allowdisplaybreaks

\newtheorem{reductionrule}{Reduction Rule}
\newtheorem{reductionrulealt}{Reduction Rule}[reductionrule]

\newtheorem{observation}{Observation}

\newenvironment{reductionrulep}[1]{
	\renewcommand\thereductionrulealt{#1}
	\reductionrulealt
}{\endreductionrulealt}

\newcommand{\restatableeq}[3]{\label{#3}#2\gdef#1{#2\tag{\ref{#3}}}}

\newcommand{\dswapnoarg}{\ensuremath{d^{\operatorname{swap}}}}
\newcommand{\dswap}[2]{\ensuremath{\dswapnoarg(#1,#2)}}

\newcommand{\arrangement}[1]{\ensuremath{\mathcal{A}^{#1}}}
\newcommand{\primalg}[1]{\ensuremath{\mathcal{P}^{#1}}}
\newcommand{\dualg}[1]{\ensuremath{\mathcal{D}^{#1}}}

\newcommand{\region}[2]{\ensuremath{R^{#1}(#2)}}
\newcommand{\halfplane}[3]{\ensuremath{H^{#1}_{#2,#3}}}

\newcommand{\bisector}[3]{\ensuremath{\beta^{#1}_{#2,#3}}}

\newcommand{\pos}[2]{\ensuremath{\operatorname{pos}_{#1}(#2)}}
\newcommand{\posi}[1]{\ensuremath{\operatorname{pos}^{-1}(#1)}}

\newcommand{\sym}[1]{\ensuremath{\mathcal{S}_{#1}}}

\newcommand{\conv}[1]{\ensuremath{\operatorname{conv}(#1)}}

\newcommand{\controversity}[2]{\ensuremath{\mathcal{C}\mathcal{G}(#1,#2)}}

\newcommand{\openballcr}[2]{\ensuremath{B_{#2}(#1)}}

\newcommand{\openballrc}[2]{\openballcr{#2}{#1}}

\newcommand{\openballcp}[2]{\ensuremath{B(#1,#2)}}

\newcommand{\boundary}[1]{\ensuremath{\partial{#1}}}

\newcommand{\closure}[1]{\ensuremath{\overline{#1}}}

\newcommand{\annuluscpp}[3]{\ensuremath{A(#1,#2,#3)}}

\newcommand{\ray}[2]{\ensuremath{\overrightarrow{{#1}{#2}}}}

\newcommand{\infline}[2]{\ensuremath{\overleftrightarrow{{#1}{#2}}}}

\newcommand{\linesegment}[2]{\ensuremath{\overline{{#1}{#2}}}}

\newcommand{\ub}{\ensuremath{\operatorname{ub}}}

\crefname{lemma}{lemma}{lemmata}
\crefname{reduction rule}{reduction rule}{reduction rules}
\crefname{observation}{observation}{observations}
	{\end{quote}}
\newcounter{clcnt}
\newcommand{\ER}{\ensuremath{\exists\mathbb{R}}}
\usepackage{footnote}
\makesavenoteenv{tabular}
\makesavenoteenv{table}

\begin{document}


	\pagestyle{fancy}
	\fancyhead{}
	
	
	\maketitle 

	
	\section{Introduction}
	
	The complexity of recognition problems of graph classes and geometrically inspired objects is very classical in computer science.
	These problems exhibit a wide range of complexities, from polynomial time to \NP-complete, and even \ER-complete; see, e.g.,~\cite{Simon91,Corneil04,BretscherCHP08,CardinalFMTV18,KangM12,Kratochvil91,MilanicRT14,BertschingerEKMW23,KratochvilM94,matousek2014intersection,MvLvL13,schaefer2003recognizing}.
	While some of these classes arise from one-dimensional objects, many are based on higher-dimensional structures.
	
	In voting theory, preference restrictions studied over the past decades have predominantly focused on ``linearly ordered'' profiles, such as single-peaked or single-crossing; see
	Elkind, Lackner, and Peters~\cite{ElkindLP22survey}. Note that these profiles are inherently one-dimensional in a certain sense.
	Elkind, Lackner, and Peters also highlight that multidimensional domain restrictions present many challenging research questions.
	Geometric representations help visualize both the profiles and the election results~\cite{DONALDG2011897}. Peters~\cite{Peters17} showed that recognizing two-dimensional profiles is \ER-complete under the standard Euclidean metric.
	
	To the best of our knowledge, no direct algorithmic advantage has been shown to arise from higher-dimensional profiles, despite significant efforts within the community.
	For example, the determination of Kemeny winner was recently shown to be \NP-complete even on two-dimensional profiles~\cite{EscoffierST22}.
	However, if the Kemeny consensus is additionally restricted to be embeddable within the given profile, it can be found in polynomial time for any fixed dimension $d$ and is a $2$-approximation of the unrestricted Kemeny consensus~\cite{HammLR21}.
	
	Many mathematical objects (graphs, elections, matrices, \ldots) allow characterization by a certain set of forbidden substructures. A~classical example is the result of Kuratowski~\cite{Kuratowski1930} characterizing planar graphs as those that do not contain subdivisions of $K_{3,3}$ or $K_5$ as subgraphs. Less known examples include the characterization of of interval graphs by infinite families of forbidden induces subgraphs by Lekkerkerker and Boland~\cite{LekkeikerkerB1962} or characterization of the consecutive $1$'s property for binary matrices by Tucker~\cite{Tucker1972}. 
	
	In the area of voting theory, certain properties of elections are similarly characterized by forbidden substructures. An example is the result of Ballester and Haeringer~\cite{Ballester2011} characterizing the single-peaked elections in terms of two forbidden substructures. Another result of this kind is due to Bredereck, Chen and Woeginger~\cite{Bredereck2013} who characterized single-crossing elections by forbidding two particular patterns to occur as a subelection. 
	
	The power of forbidden substructure characterizations extends beyond theoretical classifications. Some of these characterizations can be even used to design an efficient algorithm recognizing objects with desired properties. For example the characterization of interval graphs can be directly applied algorithmically~\cite{Lindzey2013OnFL}. 
	
	However, in some cases, the `characterizations' are incomplete, meaning they do not contain all possible forbidden substructures that would disqualify an object from possessing the desired property. In this case, an algorithm based on these forbidden substructures is only able to detect that the object does not satisfy the given property. If no such substructure is found by the algorithm, the algorithm cannot conclude anything about the object as it may contain another unidentified disqualifying substructure. This limitation occurs in particular when dealing with $d$-Euclidean elections. For $d\geq 2$, the problem of deciding whether an election is $d$-Euclidean is $\exists \mathbb{R}$-complete~\cite{Peters17} and there cannot be finite characterization of $d$-Euclidean elections and unless $\coNP\subseteq\ER$ there isn't even polynomial-time recognizable characterization (see~\cite{Peters17} for more details). Interestingly, Chen, Pruhs, and Woeginger~\cite{ChenPW15} showed that already $1$-Euclidean preferences cannot be characterized by finitely many forbidden substructures. On the other hand, $1$-Euclidean elections are recognizable in polynomial time~\cite{Doignon1994,Knoblauch2010,Elkind2014}.
	
	Theory suggests that it is impossible to give a good characterizing set of forbidden substructures for $2$-Euclidean elections. However, many real-world instances are not $2$-Euclidean (at least $91.5\%$ of PrefLib~\cite{EscoffierST23}). This result was obtained by exhaustively trying all subelections and simply using the fact that $2$-Euclidean elections can contain roughly $O(|C|^4)$ votes (generally, an election can contain up to $|C|!$ votes) or the characterization of $2$-Euclidean profiles on at most $4$ candidates~\cite{KamiyaTT11}.
	
	However, many of the instances also fail to be $2$-Euclidean due to a much simpler and quickly recognizable forbidden configuration. We demonstrate this in our paper. One such example arises from the construction of 
	Bogomolnaïa and Laslier~\cite{BogomolnaiaL07} which we refer to as the \emph{$3$-$8$ pattern} (see \Cref{subsec:3_8_pattern}). Another such class of forbidden configuration arises from the convex hull of the voters (see \Cref{sec:forbidden_substructures}). This phenomenon illustrates the gap between worst-case theoretical limitations and practical application, where even an incomplete characterization may suffice to solve most practical instances.
	
	\vspace{-1em}
	\subsection{Related Work}
	Peters~\cite{Peters17} showed that there exist $d$-Euclidean elections that inherently require exponentially many bits to even represent any of their \mbox{$d$-Euclidean~embeddings}. Bennet and Hays~\cite{HaysW1961,Bennett1960} studied what is the sufficient dimension $d$ for an election to be $d$-Euclidean. They also showed that the maximum number of voters with distinct preference ranking in a $d$-dimensional election with $m$ candidates is equal to $\sum_{k=m-d}^m |s(m,k)|$, where $s(m,k)$ are the Stirling numbers of the first kind. The same result was later obtained by Good and Tideman~\cite{Good1977}. Bogomolnaïa and Laslier~\cite{BogomolnaiaL07} studied lower bounds on $d$ based on the election to guarantee a $d$-Euclidean embedding. Kamiya, Takemura, and Terao~\cite{KamiyaTT11} established the number of maximal $d$-Euclidean profiles if the number of candidates $m$ satisfies $m=d-2$ and they were able to enumerate them for $m=4$. A simpler geometrical proof of this characterization for $m=4$ was later given by Escoffier, Spanjaard and Tydrichová~\cite{EscoffierST22_arx}. Bulteau and Chen~\cite{BulteauC23} showed that any election with at most $2$ voters is $2$-Euclidean and for $3$ voters they show that any election with at most $7$ candidates is $2$-Euclidean.  There are also results considering similar topics  in a different metric than the standard Euclidean, e.g., the $\ell_1$ and $\ell_\infty$ metrics~\cite{EscoffierST22_arx,Chen22_arx}. 
	
	The most important previous work for us is the paper of Escoffier, Spanjaard and Tydrichová~\cite{EscoffierST23}. In their work, they propose an algorithm (partially) deciding whether given election is $2$-Euclidean. Their algorithm consists of two phases. In the first phase, the aim is to find a no-certificate based on the maximal non-$2$-Euclidean elections on $4$ candidates. This is achieved by testing every possible mapping to candidates in the maximal profiles. In the second phase, they aim to find an embedding using a certain randomized procedure. If no embedding is found within a specified time limit, the algorithm reports \texttt{Unknown}. Their experiments show that many real-world elections in the PrefLib dataset~\cite{Preflib} are not 2-Euclidean (i.e., the algorithm finishes within the first phase). We will refer to this particular algorithm as the \texttt{EST} algorithm.
	
	Apart from QCP used in the \texttt{EST} algorithm, there were previously more works using constrained-based techniques for recognizing restricted preference domains. In particular, Magiera and Faliszewski~\cite{MagieraF17} showed that recognizing top-monotonic preference profiles can be reduced to {\sc 2SAT}. Elkind et al.~\cite{elkind2025survey} provide an extensive survey of algorithms and complexity results for recognizing restricted preference domains.
	
	\vspace{-0.6em}
	\subsection{Our contribution}
	We introduce new, efficiently recognizable class of forbidden substructures for $2$-Euclidean elections based on the convex hull of the voters. Our experiments suggest that, in most real-world cases, using the convex hull along with the $3$-$8$ pattern (see \Cref{subsec:3_8_pattern}) is enough to conclude that an election is not $2$-Euclidean.
	
	Next, we utilize the framework of reduction rules that simplify the input instance and remove the `trivial' parts. In particular, this improves the number of classified instances simply by using our reduction rules and then running previously known algorithms (e.g., the \texttt{EST} algorithm).
	
	Finally, we study the graph-theoretical properties of the embedding graph based on a possible embedding of the election. We use integer linear program to verify that a given election cannot be $2$-Euclidean by falsifying these properties. On the other side, we propose an enhancement of the QCP approach used by Escoffier, Spanjaard and Tydrichová~\cite{EscoffierST23} and in fact show that it can be modified to provide the embedding for previously unclassified instances that are (now known to be) $2$-Euclidean.
	Using our methods, we were able to classify more than $82\%$ of previously unclassified instances of PrefLib. Previously~\cite{EscoffierST23}, there were $343$ unclassified instances from PrefLib. Using our approach, we reduce this number to $60$. For most previously classified instances we also improve the running time of finding the appropriate yes- or no-certificate by orders of magnitude. In particular, $98.7\%$ of instances of PrefLib are solved with our approach under $1$ second.

	\subsection{Paper Organization}
	In \Cref{sec:preliminaries} we review basic definitions and notation used throughout the paper. \Cref{sec:forbidden_substructures} focuses on the convex hull. The reduction rules are shown in \Cref{sec:reducing_the_number_of_candidates}. In \Cref{sec:implied_regions} we introduce the graph-theoretical framework and the ILP approach. In \Cref{sec:qcp} we discuss the QCP approach. Lastly, we comment on implementation details and experiments in \Cref{sec:experiments} and we conclude the paper with open questions and future research directions in \Cref{sec:conclusion}.
	\sv
	{
		Statements where proofs or details are omitted due to space constraints are marked with $\star$. The omitted details are available in the long version~\cite{dvorak2025arxiv2d}.
	}
	
	\section{Preliminaries}\label{sec:preliminaries}
	\sv
	{
		For $i,j\in\mathbb{Z}^+_0$ we let $[i,j]=\{x\in \mathbb{Z}^+_0\mid i\leq x \wedge x \leq j\}$ and $[j]=[1,j]$.
		An \emph{election} is a pair $(C,V)$ of \emph{candidates} and \emph{voters}. Each voter casts a \emph{vote} which is a permutation (or ranking) of $C$, that is, a bijection $v\colon C\to [|C|]$, or equivalently a strict linear order $\succ_v$ on $C$ given by $a\succ_v b \Leftrightarrow v(a) < v(b)$. We are not dealing with preference aggregation, hence we do not distinguish between vote and voter. We say that \emph{voter $v$ prefers $a$ over $b$} if $a\succ_v b$. The \emph{position} of candidate $c$ in a vote $v$, denoted by $\pos{v}{c}$ is defined to be $|\{d\in C\mid d \succ_v c\}|+1$. An election $(C',V')$ is a \emph{subelection} of $(C,V)$ (or $(C',V')$ is \emph{contained} in $(C,V)$) if there is an injective function $f\colon C'\cup V'\to C\cup V$ such that $f(C')\subseteq C$, $f(V')\subseteq V$, and $\forall a,b\in C',\forall v\in V':a\succ_v b \Leftrightarrow f(a)\succ_{f(v)}f(b)$. For election $(C,V)$ and $C'\subseteq C$ the \emph{subelection of $(C,V)$ induced by $C'$} is the election $(C',V[C'])$, where $V[C']$ is the set of all votes in $V$ restricted to the candidates in ~$C'$. We say that election $(C,V)$ is \emph{$2$-Euclidean} if there exists an embedding (i.e., an injective function) $\gamma\colon C\cup V \to \mathbb{R}^2$ such that $\forall v\in V, a,b\in C$ we have $a\succ_v b \Rightarrow \ell_2(\gamma(v),\gamma(a))<\ell_2(\gamma(v),\gamma(b))$, where $\ell_2$ is the two-dimensional euclidean distance. We refer to $\gamma$ as \emph{$2$-Euclidean} embedding of $(C,V)$. Let $\gamma$ be a $2$-Euclidean embedding of $(C,V)$ and $a,b\in C$. The perpendicular bisector of the line segment with endpoints $\gamma(a)$ and $\gamma(b)$ is denoted by $\bisector{\gamma}{a}{b}$. The bisector $\bisector{\gamma}{a}{b}$ splits the plane into two open half-planes $\halfplane{\gamma}{a}{b}$, and $\halfplane{\gamma}{b}{a}$, where $\gamma(a)\in \halfplane{\gamma}{a}{b}$ and $\gamma(b)\in \halfplane{\gamma}{b}{a}$. The set of all bisectors $\arrangement{\gamma}=\{\bisector{\gamma}{a}{b}\mid a,b\in C, a\neq b\}$ induces an arrangement of lines whose cells are refered to as \emph{regions}. We let $\region{\gamma}{v}=\bigcap_{a\succ_v b}\halfplane{\gamma}{a}{b}$ be the region corresponding to the vote $v$. Note that $(C,V)$ is $2$-euclidean if and only if there is an embedding $\gamma\colon C \to \mathbb{R}^2$ such that $\region{\gamma}{v}\neq \emptyset$ for $v\in V$. We use the same letter $\gamma$ to refer to an embedding of the candidate set $C$.

		Candidates $a,b\in C$ are \emph{consecutive} in vote $v$ if \mbox{$|v(a)-v(b)|=1$}. A transposition $\tau_{a,b}\colon C \to C$ is a \emph{consecutive swap} with respect to $v$ if $a,b$ are consecutive in $v$. For two votes $u,v$ we denote by $\dswap{u}{v}$ their \emph{swap distance} which is the minimum number of consecutive swaps needed to obtain $v$ from $u$. We denote the set of all possible votes over $C$ (i.e., permutations of $C$) by $\sym{C}$. 
	}
	
	\toappendix
	{
		\sv
		{
			\section{Missing material from section Preliminaries}\label{app:sec:preliminaries}
		}
		\paragraph*{Overview} We start by reviewing basic concepts from geometry (\Cref{subsec:geometry}), graph theory (\Cref{subsec:graph_theory}), permutations (\Cref{subsec:permutations}), and logic (\Cref{subsec:logic}). In \Cref{subsec:elections} we introduce the main terminology from elections and election embeddings. In \Cref{subsec:technical_lemmata} we prove several technical lemmata about embeddings and in \Cref{subsec:3_8_pattern} we comment on the $3$-$8$ pattern.
		
		\subsection{Geometry}\label{subsec:geometry}
		Let $p\in\mathbb{R}^2$ be a point. We denote by $x(p),y(p)$ its $x$- and $y$-coordinates, respectively. We denote by $\ell_2$ the Euclidean distance in $\mathbb{R}^2$, that is $\ell_2(p,q)=\sqrt{(x(p)-x(q))^2+(y(p)-y(q))^2}$.
		For $r\in \mathbb{R}^+$ we denote by $\openballrc{r}{p}=\{x\in \mathbb{R}^2 \mid \ell_2(x,p)<r\}$ the \emph{open ball of radius $r$ centered at $p$}. For a set $P\subseteq \mathbb{R}^2$ we denote by $\boundary{P}$ the \emph{topological boundary} of $P$ which is defined to be
		$\boundary{P} = \{p\in \mathbb{R}^2 \mid \forall \varepsilon>0: B_\varepsilon(p)\cap P \neq \emptyset \wedge B_\varepsilon(p)\cap (\mathbb{R}^2\setminus P)\neq \emptyset\}$ and $\closure{P}$ denotes the closure of $P$ defined as $\closure{P}=P\cup \boundary{P}$. 
		Let $h\in \mathbb{R}^2$ be a point and $k\in\mathbb{R}\setminus \{0\}$. A \emph{homothety with center~$h$ and ratio~$k$} is a mapping $\mathbb{R}^2\to \mathbb{R}^2$ given by $x\mapsto h + k(x-h)$. See \Cref{fig:linear_algebra} for an example of a homothety.

		If $T$ is a homothety with $k > 0$ and $h=(0,0)$ we refer to $T$ as \emph{uniform scaling with factor $k$}. Let $t\in \mathbb{R}^2$, a \emph{translation} given by $t$ is the mapping $\mathbb{R}^2\to\mathbb{R}^2$ given by $x\mapsto x + t$. A \emph{rotation by angle $\theta\in\mathbb{R}$} (around the origin) is a mapping $\mathbb{R}^2\to\mathbb{R}^2$ given by $x\mapsto A\cdot x$ where $A$ is the rotation matrix $\begin{pmatrix} \cos \theta & -\sin \theta \\ \sin \theta & \cos \theta\end{pmatrix}$.
		
		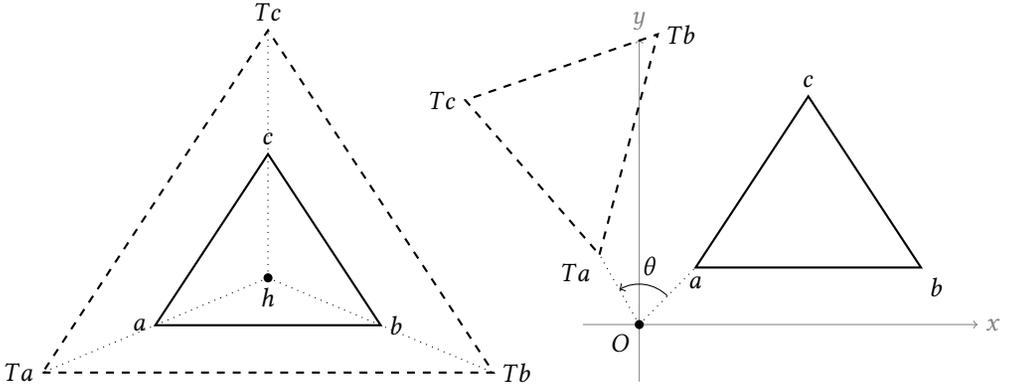
\begin{figure}[ht]
			\begin{minipage}{0.4\textwidth}
				\begin{tikzpicture}[scale=0.75]
					\coordinate (A) at (0,0);
					\coordinate (B) at (4,0);
					\coordinate (C) at (2,3);
					
					\coordinate (O) at (2,0.83);
					
					\def\k{2}
					
					\coordinate (A') at ($(O)!\k!(A)$);
					\coordinate (B') at ($(O)!\k!(B)$);
					\coordinate (C') at ($(O)!\k!(C)$);
					
					\draw[thick] (A) -- (B) -- (C) -- cycle;
					\node[left] at (A) {$a$};
					\node[right] at (B) {$b$};
					\node[above] at (C) {$c$};
					
					\draw[dashed, thick] (A') -- (B') -- (C') -- cycle;
					\node[left] at (A') {$Ta$};
					\node[right] at (B') {$Tb$};
					\node[above] at (C') {$Tc$};
					
					\filldraw[] (O) circle (2pt);
					\node[below] at (O) {$h$};
					
					\draw[dotted] (O) -- (A');
					\draw[dotted] (O) -- (B');
					\draw[dotted] (O) -- (C');
				\end{tikzpicture}
			\end{minipage}
			\begin{minipage}{0.4\textwidth}
				\begin{tikzpicture}[scale=0.75]
					\coordinate (A) at (1,1);
					\coordinate (B) at (5,1);
					\coordinate (C) at (3,4);
					
					\def\angle{75}
					
					\coordinate (A') at ({cos(\angle)*1 - sin(\angle)*1}, {sin(\angle)*1 + cos(\angle)*1});
					\coordinate (B') at ({cos(\angle)*5 - sin(\angle)*1}, {sin(\angle)*5 + cos(\angle)*1});
					\coordinate (C') at ({cos(\angle)*3 - sin(\angle)*4}, {sin(\angle)*3 + cos(\angle)*4});
					
					\draw[help lines,->] (-1,0) -- (6,0) node[right] {$x$};
					\draw[help lines,->] (0,-1) -- (0,5) node[above] {$y$};
					
					\draw[thick] (A) -- (B) -- (C) -- cycle;
					\node[below] at (A) {$a$};
					\node[below right] at (B) {$b$};
					\node[above] at (C) {$c$};
					
					\draw[dashed, thick] (A') -- (B') -- (C') -- cycle;
					\node[below left] at (A') {$Ta$};
					\node[right] at (B') {$Tb$};
					\node[left] at (C') {$Tc$};
					
					\filldraw[] (0,0) circle (2pt);
					\node[below left] at (0,0) {$O$};
					
					\draw[dotted] (0,0) -- (A);
					\draw[dotted] (0,0) -- (A');
					\draw[->] (0.5,0.5) arc[start angle=45, end angle=45+\angle, radius=0.707] node[midway,above,xshift=2pt] {$\theta$};
				\end{tikzpicture}
			\end{minipage}
			
			\caption{Examples of transformations. On the left, the solid triangle $\triangle a b c$ is mapped via a homothety with center $h$ and ratio $k=2$ to the dashed triangle $\triangle Ta Tb Tc$. We note that $h$ is chosen to be the circumcenter of the triangle $\triangle a b c$. Note that $h$ remains the circumcenter of the triangle $\triangle Ta Tb Tc$. This is because the circumcenter is the intersection of perpendicular bisectors of the three sides of a triangle.
				On the right the solid triangle $\triangle a b c$ is mapped via a rotation by angle $\theta=75^{\circ}$ to the dashed triangle $\triangle Ta Tb Tc$. Note that the center of rotation is always the origin. }\label{fig:linear_algebra}
		\end{figure}

		
		For two distinct points $p,q\in\mathbb{R}^2$ we denote by $\openballcp{p}{q}$ the open ball centered at $p$ whose boundary contains $q$, that is, $\openballcp{p}{q}=\openballcr{p}{\ell_2(p,q)}$. For three distinct points $p,q_1,q_2$ with $\ell_2(p,q_1)<\ell_2(p,q_2)$ we denote by $\annuluscpp{p}{q_1}{q_2}$ the two-dimensional open annulus given by the two concentric circles $\boundary{\openballcp{p}{q_1}}$ and $\boundary{\openballcp{p}{q_2}}$. That is, $\annuluscpp{p}{q_1}{q_2}=\openballcp{p}{q_2}\setminus \closure{\openballcp{p}{q_1}}$. A set $P\subseteq \mathbb{R}^2$ is \emph{bounded} if for any point $p$ there is $r>0$ such that $P\subseteq B_r(p)$.
		
		Let $p_1,p_2,p_3\in \mathbb{R}^2$ be three distinct points not lying on a common line. We denote by $\triangle p_1p_2p_3$ the triangle with vertices $p_1,p_2,p_3$. We denote by $\linesegment{p_1}{p_2}=\{ p_1 + t(p_2-p_1) \mid t \in [0,1]\}$ the closed line segment with endpoints $p_1,p_2$ and $\ray{p_1}{p_2}=\{p_1+t (p_2-p_1)\mid t \geq 0\}$ is the infinite ray with origin at $p_1$ and direction $p_2-p_1$. We denote by $\infline{p_1}{p_2}$ the line given by the two distinct points $p_1,p_2$, that is, $\infline{p_1}{p_2}=\{p_1+t (p_2-p_1)\mid t \in\mathbb{R}\}$. A set $P\subseteq \mathbb{R}^2$ is \emph{convex} if for any two distinct $p_1,p_2\in P$ we have $\linesegment{p_1}{p_2}\subseteq P$. For a set $P\subseteq \mathbb{R}^2$ we denote by $\conv{P}$ the convex hull of $P$, which is defined to be the smallest convex set containing~$P$.
		
		\subsection{Graph theory}\label{subsec:graph_theory}
		A (simple undirected) graph $G$ is a pair $(V(G),E(G))$, where $V(G)$ is a set of \emph{vertices} and $E(G)$ is a set of \emph{edges}. An edge is a pair of distinct vertices, i.e., $E(G)\subseteq \{\{u,v\}\mid u,v\in V(G),u\neq v\}$. The \emph{neighborhood} of a vertex $v$ is $N_G(v)=\{u\in V\mid \{u,v\}\in E\}$. The degree of $v$ is $\deg v = |N(v)|$. The maximum degree of a graph $G$ is $\Delta (G)  = \max_{v\in V} \deg v$ and the minimum degree is $\delta(G) = \min_{v\in V}\deg v$. A \emph{path (from $v_1$ to $v_k$)} or \emph{$v_1$-$v_k$ path} in a graph~$G$ is a sequence $v_1,v_2,\ldots,v_k$ of vertices of $G$ such that for every $i\in \{1,\ldots,k-1\}$ we have $\{v_i,v_{i+1}\}\in E(G)$. A graph $G$ is \emph{connected} if for any two vertices $u,v\in V(G)$ there is an $u$-$v$ path. A \emph{cycle} in a graph $G$ is a sequence $v_1,v_2,\ldots,v_k,v_1$ of vertices in $G$ such that $v_1,\ldots,v_k$ is a path and $\{v_k,v_1\}\in E(G)$. A graph is \emph{bipartite} if it contains no cycle of odd length. Graph $H$ is a subgraph of graph $G$ if $V(H)\subseteq V(G)$ and $E(H)\subseteq E(G)$. The \emph{distance} of two vertices $u,v$ in graph $G$ is denoted by $d_G(u,v)$ and is defined as the length of a shortest $u$-$v$ path in $G$. A subgraph $H$ of graph $G$ is said to be \emph{distance-preserving} if for any two vertices $u,v\in V(H)$ we have $d_H(u,v)=d_G(u,v)$.
		
		\paragraph*{Plane graphs.} An \emph{arc} is a subset of $\mathbb{R}^2$ of the form $f([0,1])$ where $f\colon [0,1]\to \mathbb{R}^2$ is an injective continuous function. We refer to the points $f(0)$ and $f(1)$ as \emph{endpoints} of the arc.
		A \emph{plane graph} is a pair $(V(G),E(G))$ of finite sets where $V(G)\subseteq \mathbb{R}^2$ are the \emph{vertices} and $E(G)$ are the \emph{edges}. Every edge $e\in E(G)$ is an arc with endpoints in $V(G)$. Any two arcs $e,f\in E(G)$ have different set of endpoints and the set $f((0,1))$ contains no vertex and no point of any other edge. Plane graph defines an undirected graph $G=(V(G),E(G))$ in a natural way. If no confusion can occur we will use the name $G$ for both the undirected graph and for the plane graph $(V(G),E(G))$. We say that an undirected graph is \emph{planar} if it is isomorphic to a plane graph. A \emph{face} is a region of $\mathbb{R}^2\setminus (V(G)\cup \bigcup E(G))$. Note that the set $V(G)\cup \bigcup E(G)$ is bounded, thus every plane graph has exactly one unbounded face. The unbounded face is referred to as the \emph{outer face} of $G$, while the other faces are the \emph{inner faces} (of $G$). An arc is \emph{incident} to a face if it is contained in its topological boundary. 
		
		A~\emph{dual} of a plane graph $G=(V(G),E(G))$ is a plane graph $H=(V(H),E(H))$ created as follows. First, place a point arbitrarily to each face of $G$. This is the set of vertices $V(H)$. Next, for each edge $e\in E(G)$ of $G$ we connect the two new vertices in the faces incident with $e$ by an arc crossing only $e$ and exactly once. Note that in general plane duals might contain self loops and multiedges. However, all duals we work with in this work induce a simple graph. The \emph{weak dual} for a plane graph is the dual without the vertex corresponding to the outer face.
		
		\subsection{Permutations}\label{subsec:permutations} For $i,j\in \mathbb{Z}^+_0$ we denote by $[i,j]$ the discrete interval $\{x\in \mathbb{Z}^+_0\mid i\leq x \wedge x \leq j\}$ and $[j]=[1,j]$. Let $C$ be a finite nonempty set. A~\emph{permutation} of $C$ is a bijection $\pi\colon C\to [|C|]$. Each permutation $\pi$ of $C$ induces a strict linear order $\succ_\pi$ on $C$ given by $a\succ_\pi b\Leftrightarrow \pi(a)<\pi(b)$. We equivalently denote a permutation~$\pi$ of~$C$ by the $|C|$-tuple $(\pi^{-1}(1),\pi^{-1}(2),\ldots,\pi^{-1}(|C|))$. If no confusion can occur, we omit commas and parentheses, e.g., for the set $C=\{a,b,c,d\}$ an example of a permutation of $C$ is $\pi=dacb$. Note that this corresponds to the linear order $d\succ a \succ c \succ b$. A~\emph{transposition} of $a,b\in C$ is a bijection $\tau_{a,b}\colon C \to C$ where $\tau_{a,b}(a)=b,\tau_{a,b}(b)=a$ and $\tau_{a,b}(x)=x$ for $x\in C\setminus \{a,b\}$. Note that the composition $\pi \circ \tau_{a,b}$ induces the same linear order as $\pi$ except $a,b$ have swapped positions. For a permutation~$\pi$ we denote by $\pi^R$ the reversed permutation, i.e., the linear order given by~$\pi$, reversed, e.g., $(dabc)^R=cbad$. We say that two elements~$a,b$ are consecutive in a permutation~$\pi$ if $|\pi(a)-\pi(b)|=1$. A transposition $\tau_{a,b}$ is \emph{consecutive swap with respect to $\pi$} if $a,b$ are consecutive in~$\pi$. Given two permutations $\pi,\rho$ of~$C$, we denote by $\dswap{\pi}{\rho}$ their \emph{swap distance} which is defined to be the minimum number of consecutive swaps needed to obtain $\rho$ from~$\pi$. Note that the swap distance is equivalent to the so-called Kendall Tau distance between $\pi$ and $\rho$ defined as the number of discordant pairs of $\pi$ and $\rho$. An unordered pair of distinct elements $a,b\in C$ is a \emph{discordant pair} if $\pi(a)<\pi(b) \wedge \rho(a)>\rho(b)$ or $\pi(a)>\pi(b)\wedge \rho(a)<\rho(b)$.
		
		We denote by $\sym{C}$ the set of all permutations of~$C$. There is a naturally associated undirected graph of permutations, which we denote by $G(\sym{C})$ where $V(G(\sym{C}))=\sym{C}$ and for two permutations $\pi,\rho\in \sym{C}$ we have $\{\pi,\rho\}\in E(G(\sym{C}))$ if and only if $\dswap{\pi}{\rho}=1$. See \Cref{fig:graph_of_permutations} for an example.
		
		\begin{figure}[ht]
			\centering
			
				%
				%
				%
			\begin{tikzpicture}[scale=1.5,z={(0,0,0.9)}, line join=round]
				\foreach [var=\x, var=\y, var=\z, count=\n] in {
					2/1/0, 2/0/1, 2/-1/0, 2/0/-1,
					-2/1/0,-2/0/1,-2/-1/0,-2/0/-1,
					1/0/ 2, 0/1/ 2, -1/0/ 2, 0/-1/ 2,
					1/0/-2, 0/1/-2, -1/0/-2, 0/-1/-2,
					1/2/0,0/2/1,-1/2/0,0/2/-1,
					1/-2/0,0/-2/1,-1/-2/0,0/-2/-1
				}
				{
					\coordinate (n\n) at (\x,\y,\z); 
				}
				\node (abcd) at (n1){$abcd$};
				\node (bacd) at (n2){$bacd$};
				\node (acbd) at (n17){$acbd$};
				\node (abdc) at (n4){$abdc$};
				\node (badc) at (n3){$badc$};
				\node (bcad) at (n9){$bcad$};
				\node (cbad) at (n10){$cbad$};
				\node (bcda) at (n12){$bcda$};
				\node (bdca) at (n22){$bdca$};
				\node (cbda) at (n11){$cbda$};
				\node (bdac) at (n21){$bdac$};
				\node (dbac) at (n24){$dbac$};
				\node (dbca) at (n23){$dbca$};
				\node (dcba) at (n7){$dcba$};
				\node (cabd) at (n18){$cabd$};
				\node (acdb) at (n20){$acdb$};
				\node (cadb) at (n19){$cadb$};
				\node (adcb) at (n14){$adcb$};
				\node (adbc) at (n13){$adbc$};
				\node (dabc) at (n16){$dabc$};
				\node (dacb) at (n15){$dacb$};
				\node (dcab) at (n8){$dcab$};
				\node (cdba) at (n6){$cdba$};
				\node (cdab) at (n5){$cdab$};

				\draw[shorten <=10pt, shorten >=10pt,red] (n1) -- (n2);
				\draw[shorten <=10pt, shorten >=10pt] (n2) -- (n3);
				\draw[shorten <=10pt, shorten >=10pt] (n3) -- (n4);
				\draw[shorten <=10pt, shorten >=10pt] (n4) -- (n1);
				\draw[shorten <=10pt, shorten >=10pt] (n5)--(n6);
				\draw[shorten <=10pt, shorten >=10pt] (n6)--(n7);
				\draw[dashed,shorten <=10pt, shorten >=10pt] (n7)--(n8);
				\draw[dashed,shorten <=10pt, shorten >=10pt] (n8)--(n5);
				\draw[shorten <=10pt, shorten >=10pt] (n9)--(n10);
				\draw[shorten <=10pt, shorten >=10pt] (n10)--(n11);
				\draw[shorten <=10pt, shorten >=10pt] (n11)--(n12);
				\draw[shorten <=10pt, shorten >=10pt,red] (n12)--(n9);
				\draw[dashed,shorten <=10pt, shorten >=10pt] (n13)--(n14);
				\draw[dashed,shorten <=10pt, shorten >=10pt] (n14)--(n15);
				\draw[dashed,shorten <=10pt, shorten >=10pt] (n15)--(n16);
				\draw[dashed,shorten <=10pt, shorten >=10pt] (n16)--(n13);
				\draw[shorten <=10pt, shorten >=10pt] (n17)--(n18);
				\draw[shorten <=10pt, shorten >=10pt] (n18)--(n19);
				\draw[shorten <=10pt, shorten >=10pt] (n19)--(n20);
				\draw[shorten <=10pt, shorten >=10pt] (n20)--(n17);
				\draw[shorten <=10pt, shorten >=10pt] (n21)--(n22);
				\draw[shorten <=10pt, shorten >=10pt,red] (n22)--(n23);
				\draw [dashed,shorten <=10pt, shorten >=10pt] (n23)--(n24);
				\draw [dashed,shorten <=10pt, shorten >=10pt] (n24)--(n21);
				
				\draw[shorten <=10pt, shorten >=10pt,red] (n12) -- (n22);
				\draw[shorten <=10pt, shorten >=10pt] (n21) -- (n3);
				\draw[shorten <=10pt, shorten >=10pt,red] (n2) -- (n9);
				
				\draw[shorten <=10pt, shorten >=10pt] (n23) -- (n7);
				\draw[shorten <=10pt, shorten >=10pt] (n6) -- (n11);
				
				\draw[shorten <=10pt, shorten >=10pt] (n10) -- (n18);
				\draw[shorten <=10pt, shorten >=10pt] (n5) -- (n19);
				
				\draw[dashed,shorten <=10pt, shorten >=10pt] (n8) -- (n15);
				\draw[dashed,shorten <=10pt, shorten >=10pt] (n4) -- (n13);
				\draw[dashed,shorten <=10pt, shorten >=10pt] (n16) -- (n24);
				\draw[dashed,shorten <=10pt, shorten >=10pt] (n20) -- (n14);
				\draw[shorten <=10pt, shorten >=10pt] (n1) -- (n17);
				
			\end{tikzpicture}
			\caption{The graph of permutations $G(\sym{C})$ for $C=\{a,b,c,d\}$ realised in the shape of a truncated octahedron. The edges of the graph correspond to consecutive swaps. Note that the distance $\dswapnoarg$ coincides with the graph-theoretical distance $d_{G(\sym{C})}$. For example we have \mbox{$\dswap{abcd}{dbca}=5$} because $dbca=abcd\circ \tau_{a,b} \circ \tau_{a,c}\circ \tau_{a,d}\circ \tau_{c,d}\circ \tau_{b,d}$ and $dbca$ cannot be expressed as composition of $abcd$ with fewer than $5$ consecutive swaps. These consecutive swaps correspond to one of the shortest paths from $abcd$ to $dbca$ in $G(\sym{C})$ given by $abcd,bacd,bcad,bcda,bdca,dbca$. The path is highlighted in red. 
			}\label{fig:graph_of_permutations}
		\end{figure}
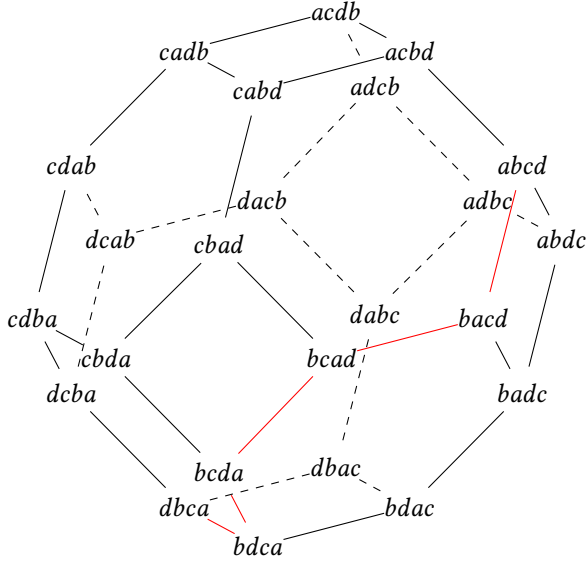
		
		\subsection{Logic}\label{subsec:logic}
		A \emph{propositional formula over set of variables $\{x_1,x_2,\ldots\}$} is a boolean expression containing the variables and logical connectives: disjunction ($\vee$), conjunction ($\wedge$), negation ($\neg$), implication ($\Rightarrow$), and equivalence ($\Leftrightarrow$). Formally, every variable is a formula and if $\varphi,\psi$ are two formulas, then $(\varphi \oplus \psi)$ and $\neg \varphi$ are also formulas, where $\oplus\in\{\vee,\wedge,\Rightarrow,\Leftrightarrow\}$. A \emph{literal} is either a variable $x$ (positive literal) or a negation of a variable $\neg x$ (negative literal). A \emph{clause} is a disjunction of literals, e.g. $x_v \vee \neg x_u \vee x_w$. A formula is in \emph{conjunctive normal form} (CNF), or \emph{CNF formula} if it is a conjunction of clauses. A (truth) assignment for formula $\varphi$ over variables $\{x_1,x_2,\ldots\}$ is a function $f\colon \{x_1,x_2,\ldots\}\to \{0,1\}$. We extend $f$ to set of all formulas as follows: $f(\neg \varphi)=1-f(\varphi),f(\varphi\wedge \psi)=\min\{f(\varphi),f(\psi)\},f(\varphi\vee\psi)=\max\{f(\varphi),f(\psi)\},f(\varphi\Rightarrow \psi)=f(\neg \varphi \vee \psi),f(\varphi \Leftrightarrow \psi)=f((\varphi\Rightarrow \psi)\wedge(\psi \Rightarrow \varphi))$. A formula $\varphi$ is satisfied by an assignment $f$ if $f(\varphi)=1$. It is not hard to observe, that if $\varphi$ is a CNF formula, then $f$ satisfies $\varphi$ if and only if for every clause $X$ of $\varphi$ there is a literal $\ell$ in $X$ such that $f(\ell)=1$.
		We say that two formulas $\varphi,\psi$ are \emph{equivalent} and write $\varphi \equiv \psi$ if for any truth assignment $f$ we have $f(\varphi)=1$ if and only if $f(\psi)=1$.

		\subsection{Elections}\label{subsec:elections}
		An \emph{election} is a pair $(C,V)$, where $C$ is the set of \emph{candidates} and $V$ is the set of \emph{voters}. Each voter $v\in V$ casts a \emph{vote} which is a permutation of $C$; or equivalently a strict linear order $\succ_v$ on $C$. In this work, we are not dealing with preference aggregation, so we use the terms \emph{vote} and \emph{voter} interchangeably and for our purposes $V$ is just a set of permutations of $C$, i.e., $V\subseteq \sym{C}$. We say that voter $v$ \emph{prefers $a$ over $b$} if $a\succ_v b$. For a vote $v$ and a candidate $c\in C$ the \emph{position} of $c$ in $v$, denoted by $\pos{v}{c}$, is defined to be $v(c)$ or equivalently $|\{d\in C \mid d\succ_v c\}| + 1$. We say that $v$ ranks $c$ \emph{first} (\emph{second},$\ldots$,\emph{last}) if $\pos{v}{c}=1$ ($\pos{v}{c}=2, \ldots, \pos{v}{c}=|C|$). 
		
		
		\paragraph*{Subelections.}
		A \emph{restriction} of a vote $v\in\sym{C}$ to $C'$ is the vote $v'\in \sym{C'}$ where for all $a,b\in C'$ we have $a\succ_{v'} b$ if and only if $a\succ_v b$. We denote the vote $v\in \sym{C}$ restricted to $C'\subseteq C$ by $v[C']$ and extend this notion to sets of votes, i.e., $V[C']=\{v[C']\mid v \in V\}$.
		\lv{
			Let $(C,V)$ be an election. Intuitively, a subelection of $(C,V)$ is an election $(C',V')$ where $C'\subseteq C$ and $V'$ can be obtained from $V$ by deleting some voters and restricting the remaining ones to the candidates from $C'$. 
			For our purposes, it will be essential to be able to capture that one election is a subelection of another election even if, for example, they have different names of candidates and voters.
		}
		Formally, an election $(C',V')$ is a subelection of $(C,V)$ if there exists an injective function $f\colon C'\cup V' \to C\cup V$ mapping candidates to candidates and voters to voters, i.e., $f(C')\subseteq C$ and $f(V')\subseteq V$ such that for any $a,b\in C'$ and $v\in V'$ we have $a\succ_v b\Leftrightarrow f(a)\succ_{f(v)}f(b)$. We alternatively say that the election $(C,V)$ \emph{contains} $(C',V')$ or that $(C',V')$ \emph{is contained in} $(C,V)$ to express the fact that $(C',V')$ is a subelection of $(C,V)$. We say that two elections $(C_1,V_1)$ and $(C_2,V_2)$ are \emph{isomorphic} if $(C_1,V_1)$ is a subelection of $(C_2,V_2)$ and vice versa. Notice that if two elections are isomorphic, the function $f$ is necessarily a bijection.

		For election $(C,V)$ and $C'\subseteq C$ the \emph{subelection of $(C,V)$ induced by the set of candidates $C'\subseteq C$} is the election $(C',V[C'])$. Note that it is possible that $|V[C']|<|V|$ as we could create duplicate votes by leaving out the candidates in $C\setminus C'$. It is not hard to verify that any subelection of $(C,V)$ induced by a set of candidates is indeed a subelection of $(C,V)$ by the definition above.
		
		\toappendix{
			\begin{example}
				Let $C_1=\{1,2,3,4\}$, $V_1=\{1324,1243,4132\}$, $C_2=\{a,b,c,d\}$, $V_2=\{bcda,dbac,bacd\}$, $C_3=\{x,y\}$, $V_3=\{xy,yx\}$. The elections $(C_1,V_1)$ and $(C_2,V_2)$ are isomorphic via the bijection $1\mapsto b, 2\mapsto c, 3\mapsto a, d\mapsto 4,1324\mapsto bacd,1243\mapsto bcda,4132\mapsto dbac$. The election $(C_3,V_3)$ is a subelection of $(C_2,V_2)$ (and also of $(C_1,V_1)$), where the injective function $f$ is given by $x\mapsto 1, y\mapsto 4,xy\mapsto 1243,yx\mapsto 4132$. Note that $(C_3,V_3)$ is isomorphic to the subelection of $(C_1,V_1)$ induced by the set of candidates $\{1,4\}\subseteq C_1$, i.e., to the election $(\{1,4\},V_1[\{1,4\}]) = (\{1,4\},\{14,41\})$.
			\end{example}
		}
		
		\paragraph*{Election Embeddings.}
		An \emph{embedding} of a finite set $X$ is an injective function $X\to \mathbb{R}^2$.
		Let $(C,V)$ be an election. We say that it is \emph{$2$-Euclidean} if there exists an embedding~$\gamma\colon C \cup V \to \mathbb{R}^2$ such that for all $v \in V$ and $a,b \in C$ we have $a \succ_v b\Rightarrow\ell_2(\gamma(v), \gamma(a)) < \ell_2(\gamma(v), \gamma(b))$. We refer to $\gamma$ as \emph{$2$-Euclidean} embedding of $(C,V)$.
		In other words, voter $v$ prefers candidate $a$ over candidate $b$ if $\gamma(v)$ is closer to~$\gamma(a)$ than to~$\gamma(b)$. Since we are dealing with strict votes then we can equivalently require $a \succ_v b\Leftrightarrow\ell_2(\gamma(v), \gamma(a)) < \ell_2(\gamma(v), \gamma(b))$.
		
		Observe that the property `being $2$-Euclidean' is hereditary under taking arbitrary subelections.
		
		\begin{observation}\label{obs:euclidean_is_hereditary}
			If $(C,V)$ is $2$-Euclidean, then any subelection of $(C,V)$ is also $2$-Euclidean.
		\end{observation}
		
		
		\paragraph*{Geometry of Embeddings.} We further explore the geometry of election embeddings and introduce several technical terms. We refer the reader to \Cref{fig:embedding_four_candidates} for illustration.
		Suppose that $(C,V)$ is 2-Euclidean and let $\gamma$ be a $2$-Euclidean embedding of $(C,V)$. For any two distinct candidates $a,b\in C$ there is perpendicular bisector of the line segment $\linesegment{\gamma(a)}{\gamma(b)}$ which we denote by $\bisector{\gamma}{a}{b}$. That is, $\bisector{\gamma}{a}{b}=\{p\in \mathbb{R}^2\mid \ell_2(p,\gamma(a))=\ell_2(p,\gamma(b))\}$. Note that $\bisector{\gamma}{a}{b}=\bisector{\gamma}{b}{a}$. The bisector $\bisector{\gamma}{a}{b}$ splits the plane into two open half-planes $\halfplane{\gamma}{a}{b},\halfplane{\gamma}{b}{a}$, where $\gamma(a)\in \halfplane{\gamma}{a}{b}$ and $\gamma(b)\in \halfplane{\gamma}{b}{a}$. Formally $\halfplane{\gamma}{a}{b}=\{p\in\mathbb{R}^2\mid \ell_2(p,\gamma(a))<\ell_2(p,\gamma(b))\}$ and $\halfplane{\gamma}{b}{a}=\{p\in\mathbb{R}^2\mid \ell_2(p,\gamma(b))<\ell_2(p,\gamma(a))\}$. For a voter $v$ we have $\gamma(v)\in \halfplane{\gamma}{a}{b}$ if $a\succ_v b$ and $\gamma(v)\in\halfplane{\gamma}{b}{a}$ if $b\succ_v a$. Note that the bisector $\bisector{\gamma}{a}{b}$ consists exactly of all the points that are equidistant to $\gamma(a)$ and $\gamma(b)$. Since we deal with votes that are strict linear orders, no voter can be embedded to any bisector, i.e., for any two candidates $a,b\in C$ we have $\gamma(V)\cap \bisector{\gamma}{a}{b}=\emptyset$. Let $\arrangement{\gamma}=\{\bisector{\gamma}{a}{b}\mid a,b\in C, a\neq b\}$ be the set of all bisectors of $\gamma$. $\arrangement{\gamma}$ induces an \emph{arrangement of lines} and it splits the plane into two-dimensional open cells which we refer to as \emph{regions}. Each region corresponds to a particular vote~$v$ and can be described as the intersection of $\binom{|C|}{2}$ open halfplanes of the form $\halfplane{\gamma}{a}{b}$. We denote the region for vote~$v$ by $\region{\gamma}{v}$. That is, $\region{\gamma}{v}=\bigcap_{a\succ_v b} \halfplane{\gamma}{a}{b}$. It is clear that if $v\in V$, then $\region{\gamma}{v}\neq \emptyset$ and $\gamma(v)\in\region{\gamma}{v}$. A region is \emph{inner} if it is bounded, otherwise it is \emph{outer}.
		
		With this terminology at hand, it is not hard to see that instead of embedding both the candidates and voters, we could only embed the candidates. It is then clear that an embedding $\gamma\colon C\to \mathbb{R}^2$ can be extended to a $2$-Euclidean embedding $\gamma \colon C\cup V \to \mathbb{R}^2$ if and only if for all $v\in V$ we have $\region{\gamma}{v}\neq \emptyset$. We use the same symbol $\gamma$ to refer to an \emph{embedding of the candidates}.
		
		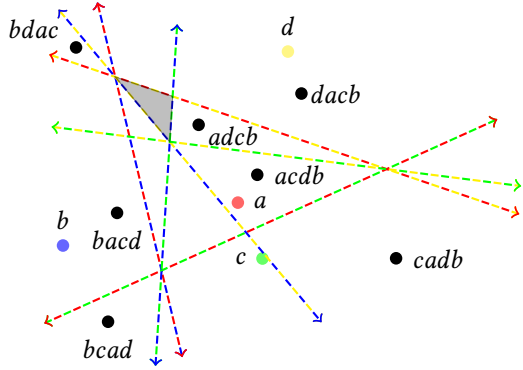
\begin{figure}
			\centering
			\begin{tikzpicture}[scale=0.3]
				\node[circle,scale=0.5,fill=red!60,label=right:$a$] (a) at (3.18,-0.78) {};
				\node[circle,scale=0.5,fill=blue!60,label=$b$] (b) at (-4.58,-2.66) {};
				\node[circle,scale=0.5,fill=green!60,label=left:$c$] (c) at (4.27,-3.23) {};
				\node[circle,scale=0.5,fill=yellow!60,label=$d$] (d) at (5.4,5.84) {};
				
				\path (a) -- (b) coordinate[midway] (Mab);
				\draw[name path=AB,postaction={draw,red,dash pattern= on 3pt off 5pt,dash phase=4pt,thick}]
				[blue,dash pattern= on 3pt off 5pt,thick,<->] ($(Mab)!6cm!270:(a)$) -- ($(Mab)!10cm!90:(a)$);
				
				\path (a) -- (c) coordinate[midway] (Mac);
				\draw[postaction={draw,red,dash pattern= on 3pt off 5pt,dash phase=4pt,thick}]
				[green,dash pattern= on 3pt off 5pt,thick,<->] ($(Mac)!12cm!270:(a)$) -- ($(Mac)!10cm!90:(a)$);
				
				\path (a) -- (d) coordinate[midway] (Mad);
				\draw[name path=AD,postaction={draw,red,dash pattern= on 3pt off 5pt,dash phase=4pt,thick}]
				[yellow,dash pattern= on 3pt off 5pt,thick,<->] ($(Mad)!10cm!270:(a)$) -- ($(Mad)!12cm!90:(a)$);
				\path (b) -- (c) coordinate[midway] (Mbc);
				\draw[name path=BC,postaction={draw,blue,dash pattern= on 3pt off 5pt,dash phase=4pt,thick}]
				[green,dash pattern= on 3pt off 5pt,thick,<->] ($(Mbc)!10cm!270:(b)$) -- ($(Mbc)!5cm!90:(b)$);
				\path (b) -- (d) coordinate[midway] (Mbd);
				\draw[name path=BD,postaction={draw,blue,dash pattern= on 3pt off 5pt,dash phase=4pt,thick}]
				[yellow,dash pattern= on 3pt off 5pt,thick,<->] ($(Mbd)!8cm!270:(b)$) -- ($(Mbd)!10cm!90:(b)$);
				\path (c) -- (d) coordinate[midway] (Mcd);
				\draw[postaction={draw,green,dash pattern= on 3pt off 5pt,dash phase=4pt,thick}]
				[yellow,dash pattern= on 3pt off 5pt,thick,<->] ($(Mcd)!10cm!270:(c)$) -- ($(Mcd)!11cm!90:(c)$);
				
				
				\node[fill,scale=0.5,circle] (bcad) at (-2.58,-6) {};
				\node[below,yshift=-3pt] at (bcad){$bcad$};
				\node[fill,scale=0.5,circle] (bacd) at (-2.18,-1.23) {};
				\node[below,yshift=-3pt] at (bacd){$bacd$};
				\node[fill,scale=0.5,circle] (bdac) at (-4,6.02) {};
				\node[above left,xshift=-3pt] at (bdac){$bdac$};
				\node[fill,scale=0.5,circle] (dacb) at (6,4) {};
				\node[right] at (dacb) {$dacb$};
				\node[fill,scale=0.5,circle] (cadb) at (10.2,-3.22) {};
				\node[right,xshift=3pt] at (cadb) {$cadb$};
				\node[fill,scale=0.5,circle] (acdb) at (4.04,0.44) {};
				\node[right,xshift=3pt] at (acdb) {$acdb$};
				\node[fill,scale=0.5,circle] (adcb) at (1.44,2.63) {};
				\node[below right,yshift=3pt] at (adcb) {$adcb$};
				
				\draw[name intersections={of=AB and AD}];
				\coordinate (x) at (intersection-1);
				\draw[name intersections={of=BC and AD}];
				\coordinate (y) at (intersection-1);
				\draw[name intersections={of=BC and BD}];
				\coordinate (z) at (intersection-1);

				\draw[fill=gray,opacity=0.5] (x) -- (y) -- (z) --cycle;
				
				%
				%

			\end{tikzpicture}
			\caption{A $2$-Euclidean embedding of the election $(C,V)$ with $C=\{a,b,c,d\}$ and $V=\{bdac,bacd,adcb,acdb,dacb,cadb,bcad\}$. The dashes lines are the bisectors between pairs of candidates. The shaded triangle is the region $\region{\gamma}{adbc}$. Note that for example the region $\region{\gamma}{cdba}$ is empty. }\label{fig:embedding_four_candidates}
		\end{figure}

	} 

	
	\subsection{Nice embeddings}\label{subsec:technical_lemmata}
	\toappendix
	{
		\sv
		{
			\subsection{Missing material from section Nice embeddigs}
		}
	}
	To prove certain properties of embeddings, it is essential to have somewhat `well-behaved' embedding. We show that we can always assume a \emph{nice} $2$-Euclidean embedding of a $2$-Euclidean election.

	\toappendix
	{
		\begin{lemma}\label{lem:technical_embedding_simplify_calculation}
			Let $X$ be a finite set and $\gamma\colon X \to \mathbb{R}^2$ an embedding. Let $x_1,x_2\in X$ and $p_1,p_2\in\mathbb{R}^2$ be such that $x_1\neq x_2$ and $p_1\neq p_2$. Then there is a transformation $T$ preserving relative distances such that $(T\circ \gamma)(x_1)=p_1$ and $(T\circ \gamma)(x_2)=p_2$. In particular for any four elements $y_1,y_2,y_3,y_4\in X$ we have $\frac{\ell_2(\gamma(y_1),\gamma(y_2))}{\ell_2(\gamma(y_3),\gamma(y_4))} = \frac{\ell_2((T\circ\gamma)(y_1),(T\circ\gamma)(y_2))}{\ell_2((T\circ\gamma)(y_3),(T\circ \gamma)(y_4))}$.
		\end{lemma}
		
		\begin{proof}
			The transformation $T$ will be given as a composition of four transformations $T_4\circ T_3 \circ T_2 \circ T_1$. $T_1$ is the uniform scaling with ratio $\frac{\ell_2(p_1,p_2)}{\ell_2(\gamma(x_1),\gamma(x_2))}$. Notice that $\ell_2((T_1\circ \gamma)(x_1),(T_1\circ \gamma)(x_2))=\ell_2(p_1,p_2)$. Next, $T_2$ is the translation given by the translation vector $t=-\gamma(x_1)$. $T_3$ is the rotation which rotates the vetor $\gamma(x_2)-\gamma(x_1)$ to the same direction as the vector $p_2-p_1$. The angle $\theta$ for $T_3$ is given by the angle that one needs to rotate $\gamma(x_2)-\gamma(x_1)$ in positive direction to obtain $p_2-p_1$. Finally, let $T_4$ be the translation given by the vector $t=p_1$. Notice that uniform scaling, translation and rotation are similarity transformations which implies that they preserve the ratios of any two straight line segments. In particular we have
			$\frac{\ell_2(q_1,q_2)}{\ell_2(q_3,q_4)}=\frac{\ell_2(Tq_1,Tq_2)}{\ell_2(Tq_3,Tq_4)}$ for any four points $q_1,q_2,q_3,q_4\in \mathbb{R}^2$ and in particular for any four points $q_1,q_2,q_3,q_4\in \gamma(X)$.
		\end{proof}
	}
	\toappendix
	{
		For an embedding $f\colon X\to \mathbb{R}^2$, $x\in X$ and $p\in\mathbb{R}^2\setminus f(X\setminus \{x\})$ we denote by $f_{x\mapsto p}$ the embedding where we send $x$ to $p$ and everything else remains unchanged. Three elements $x_1,x_2,x_3\in X$ are referred to as \emph{collinear triplet with respect to $f$} if the points $f(x_1),f(x_2),f(x_3)$ lie on a common line. 
	}
	
	\toappendix{
		\begin{lemma}\label{lem:technical_lemma_no_three_voters_on_line}
			
			Let $(C,V)$ be a $2$-Euclidean election. There is a 2-Euclidean embedding $\gamma$ of $(C,V)$ such that for any three distinct voters $v_1,v_2,v_3\in V$ the points $\gamma(v_1),\gamma(v_2),\gamma(v_3)$ do not lie on a common line.
		\end{lemma}
		
		\begin{proof}
			Let $\gamma$ be a $2$-Euclidean embedding of $(C,V)$ that minimizes the number of collinear triplets of voters with respect to $\gamma$. If $\gamma$ has no collinear triplet of voters, we are done. Suppose that there is a collinear triplet $v_1,v_2,v_3\in V$ w.r.t. $\gamma$. We show that we can move $\gamma(v_1)$ and create a new embedding $\gamma_{v_1\mapsto p}$ and has smaller number of collinear triplets of voters thus contradicting the choice of $\gamma$.
			
			Consider the region $\region{\gamma}{v_1}$. Since $\gamma$ is a $2$-Euclidean embedding of $(C,V)$ the region $\region{\gamma}{v_1}$ is nonempty and has nonzero (possibly infinite) measure. Consider all possible lines of the form $\infline{\gamma(u)}{\gamma(v)}$ for $u,v\in V\setminus \{v_1\}$. Since lines have measure zero and there is only finitely many such lines, it follows that the set $P=\region{\gamma}{v_1}\setminus \bigcup_{u,v\in V\setminus \{v_1\}}\infline{\gamma(u)}{\gamma(v)}$ has nonzero measure and is hence nonempty. Let $p\in P$ be arbitrary. It is clear that the embedding $\gamma_{v_1\mapsto p}$ is a valid $2$-Euclidean embedding of $(C,V)$ and there is strictly less collinear triplets with respect to $\gamma_{v_1\mapsto p}$ than to $\gamma$. This finishes the proof.
		\end{proof}
	}
	
	\toappendix
	{
		\begin{lemma}\label{lem:technical_lemma_no_three_candidates_on_line}
			
			Let $C$ be a finite set of candidates and let $\gamma\colon C \to \mathbb{R}^2$ be an embedding. Then there exists an embedding of candidates $\gamma^*\colon C\to\mathbb{R}^2$ such that no three points $\gamma^*(a),\gamma^*(b),\gamma^*(c)$ for $a,b,c\in C$ are collinear and for any $v\in \sym{C}$ if $\region{\gamma}{v}\neq \emptyset$, then $\region{\gamma^*}{v}\neq \emptyset$.
		\end{lemma}
		\begin{proof}
			Let $\Gamma = \{\gamma'\colon C\to\mathbb{R}^2 \mid \region{\gamma}{v}\neq \emptyset \Rightarrow \region{\gamma'}{v}\neq \emptyset \}$ be the set of all embeddings that preserve nonempty regions of $\gamma$. Let $\gamma^*\in \Gamma$ be an embedding that minimizes the number of collinear triplets. If $\gamma^*$ has no collinear triplet of candidates, we are done. Suppose that $a,b,c\in C$ is a collinear triplet w.r.t. $\gamma^*$.
			Our aim is to move $\gamma^*(a)$ to some point $p\in \openballrc{\varepsilon}{\gamma^*(a)}$ and show that such embedding is also in $\Gamma$ and has smaller number of collinear triplets than $\gamma^*$, thus contradicting the choice of $\gamma^*$.
			
			Note that for any $\varepsilon>0$ the open ball $\openballrc{\varepsilon}{\gamma^*(a)}$ has nonzero measure and hence by the very same argument as the in the proof of \Cref{lem:technical_lemma_no_three_voters_on_line} the set $\region{\gamma}{v}\setminus \bigcup_{d,e\in C\setminus \{a\}}\infline{\gamma(d)}{\gamma(e)}$ is nonempty. Hence for any $\varepsilon>0$ there is a point $p\in \openballrc{\varepsilon}{\gamma^*(a)}$ such that the mapping $\gamma^*_{a\mapsto p}$ has strictly less collinear triplets of candidates.
			
			It suffices to argue that there exists some $\varepsilon^*>0$ such that for any $p\in \openballrc{\varepsilon^*}{\gamma^*(a)}$ we have $\gamma^*_{a\mapsto p}\in \Gamma$. In other words that there is such ball $\openballrc{\varepsilon^*}{\gamma^*(a)}$ such that any movement of $\gamma^*(a)$ inside $\openballrc{\varepsilon^*}{\gamma^*(a)}$ preserves all nonempty regions of $\gamma$.
			
			Consider the movement of the point $\gamma^*(a)$ and
			let $\region{\gamma^*}{v}$ be a nonempty region such that one of its bounding bisectors is of the form $\bisector{\gamma^*}{a}{x}$ for some candidate $x\in C\setminus \{a\}$. Note that if none of the bounding bisectors of $\region{\gamma^*}{v}$ is of this form then the polygon $\boundary{\region{\gamma^*}{v}}$ and the region $\region{\gamma^*}{v}$ remains unchanged under sufficiently small movement of $\gamma^*(a)$.
			
			Any movement of $\gamma^*(a)$ induces a transformation of $\bisector{\gamma^*}{a}{x}$ which can be composed of translation and rotation. As the region $\region{\gamma^*}{v}$ has nonzero (possibly infinite) measure and there is only finitely many of bounding bisectors there is some small $\varepsilon_v>0$ such that $\region{\gamma^*_{a\mapsto p}}{v}\neq \emptyset$ for any $p\in \openballrc{\varepsilon_v}{\gamma^*(a)}$.
			
			We let
			$\varepsilon^*=\min_{\region{\gamma^*}{v}\neq \emptyset} \varepsilon_v$. By construction for any $p\in \openballrc{\varepsilon^*}{\gamma^*(a)}$ we have $\gamma^*_{a\mapsto p}\in \Gamma$ and there is also some $p^*$ such that the embedding $\gamma^*_{a\mapsto p^*}\in \Gamma$ has less collinear triplets. This finishes the proof.
		\end{proof}
	}
	
	\toappendix
	{
		\begin{lemma}\label{lem:technical_lemma_no_parallel_bisectors}
			Let $(C,V)$ be $2$-Euclidean election. Then there is $2$-Euclidean embedding $\gamma$ of $(C,V)$ such that no pair of bisectors of the form $\bisector{\gamma}{a}{b}$ and $\bisector{\gamma}{c}{d}$ are parallel.
		\end{lemma}
		\begin{proof}
			We proceed similarly as in the proof of \Cref{lem:technical_lemma_no_three_candidates_on_line}. We refer to a pair of parallel bisectors as a \emph{parallel pair}.
			We show that any $2$-Euclidean embedding of $(C,V)$ can be adjusted such that it has less parallel pairs. Hence there is an embedding with no parallel pairs.
			
			Start with any embedding that has no collinear triplets of candidates and let $\gamma^*$ minimize the number of parallel pairs among such embeddings. Suppose there is a parallel pair $\bisector{\gamma^*}{a}{b},\bisector{\gamma^*}{c}{d}$. Note that $\{a,b\}\cap \{c,d\}=\emptyset$, because otherwise we have a collinear triplet of candidates. We will move the candidate $a$ in the perpendicular direction to the line segment $\linesegment{\gamma^*(a)}{\gamma^*(b)}$. This induces a rotation of the bisector $\bisector{\gamma^*}{a}{b}$ and by similar arguments as in \Cref{lem:technical_lemma_no_three_candidates_on_line} we can ensure that no triplets of candidates become collinear and no pair of bisectors becomes a parallel pair.
		\end{proof}
	}
	\begin{definition}
		An embedding of candidates $\gamma\colon C \to \mathbb{R}^2$ is \emph{nice} if there are no two parallel bisectors of the form $\bisector{\gamma}{a}{b},\bisector{\gamma}{c}{d}$ for some $a,b,c,d\in C$. A $2$-Euclidean embedding $\gamma$ of an election $(C,V)$ is \emph{nice} if $\gamma$ restricted to $C$ is nice and it has no collinear triplet of voters.
	\end{definition}
	
	\sv
	{
		\begin{theorem}[$\star$]\label{thm:non_degenerate_embedding}
		}
		\lv
		{
			\begin{theorem}\label{thm:non_degenerate_embedding}
			}
			Every $2$-Euclidean election admits a nice\\ $2$-Euclidean embedding.
		\end{theorem}
		\toappendix
		{
			\sv
			{
				\begin{proof}[Proof of \Cref{thm:non_degenerate_embedding}]
				}
				\lv
				{
					\begin{proof}
					}
					Start with the embedding with no parallel bisectors and no collinear triplets of candidates provided by \Cref{lem:technical_lemma_no_parallel_bisectors}. Then apply the procedure described in the proof of \Cref{lem:technical_lemma_no_three_voters_on_line} to adjust the voters, but start with an embedding that already has no collinear triplets of candidates nor parallel bisectors. The resulting embedding satisfies all conditions promised in the theorem.
				\end{proof}
			}
			
			\subsection{Forbidden substructure on $3$ voters ($3$-$8$ pattern)}\label{subsec:3_8_pattern}
			\toappendix
			{
				\sv
				{
					\subsection{Missing material from section Forbidden substructure on $3$ voters ($3$-$8$ pattern)}
				}
			}
			It has been shown by Bulteau and Chen~\cite{BulteauC23} that any election with $|C|\leq 2$ or $|C|\leq 3 \wedge |V|\leq 7$ is $2$-Euclidean. Note that there exists an election with $3$ voters and $8$ candidates that is not $2$-Euclidean. This counterexample, given by Bogomolnaïa and Laslier~\cite[Proposition 11]{BogomolnaiaL07}, gives rise to a class of counterexamples for any fixed dimension $d\geq 2$. For $d=2$ the pattern is as follows (see (\ref{eq:the38pattern})). The set of voters is $V=\{v_1,v_2,v_3\}$ and the candidate set corresponds to the powerset of $V$. That is $C=\{c_{\emptyset}, c_{1},c_{2},c_{3},c_{12},c_{13},c_{23},c_{123}\}$. The preferences are as follows. The candidate $c_{\emptyset}$ is the `central' candidate and is ranked fifth in each vote. For each index $i\in[3]$ we have $c_{I}\succ_{v_i} c_{\emptyset}$ if and only if $i\in I$ and the remaining preferences may be arbitrary. We will refer to this pattern as the \emph{$3$-$8$ pattern}. 
			\begin{align}\label{eq:the38pattern}\tag{$*$}
				\minibox{$v_1\colon$ \\ $v_2\colon$ \\ $v_3\colon$}
				\minibox[frame]{$c_{123}\succ c_{12}\succ c_{13}\succ c_1\succ$ \\
					$c_2 \succ c_{23}\succ c_{123}\succ c_{12}\succ$ \\ $ c_{13}\succ c_3\succ c_{123}\succ c_{23}\succ$}
				\minibox[frame]{$c_\emptyset$ \\ $c_\emptyset$ \\ $c_\emptyset$}
				\minibox[frame]{$\succ c_2\succ c_{23}\succ c_3$\\ $\succ c_{13}\succ c_1\succ c_3$ \\ $\succ c_1 \succ c_2 \succ c_{12}$}
			\end{align}
			
			\sv
			{
				\begin{lemma}[$\star$]\label{lem:38patterrec}
				}
				\lv
				{
					\begin{lemma}\label{lem:38patterrec}
					}
					Given an election $(C,V)$, we can in $O(|V|^3|C|^2)$ time decide whether $(C,V)$ contains the $3$-$8$ pattern.
				\end{lemma}
				\toappendix
				{
					\sv
					{
						\begin{proof}[Proof of \Cref{lem:38patterrec}]
						}
						\lv{
							\begin{proof}
							}
							Fix $\{v_1,v_2,v_3\}\subseteq V$ and fix one of the candidates in $C$ as $c_\emptyset$. For any other candidate $d\in C\setminus \{c_\emptyset\}$ compute the vector $x_d=(x_1,x_2,x_3)$ where $x_i=1$ if and only if $d\succ_{v_i} c_\emptyset$, otherwise $x_i=0$. Let $X=\{x_d\mid d\in C\setminus \{c_\emptyset\}\}$ be the set of all bit vectors. If $|X\setminus \{(0,0,0)\}|=7$, then the subelection $(C,\{v_1,v_2,v_3\})$ contains the forbidden $3$-$8$ pattern. To see this, note that $(x_1,x_2,x_3)\in X$ if and only if there is a candidate $c\in C\setminus \{c_\emptyset\}$ that corresponds to the candidate $c_I$ for $I=\{i\mid x_i = 1\}$. Hence if $X$ contains all the $7$ patterns corresponding to $c_1,c_2,c_3,c_{12},c_{13},c_{23},$ and $c_{123}$, the input election contains the $3$-$8$ pattern.
							
							The algorithm deciding whether $(C,V)$ contains the $3$-$8$ pattern is as follows. Try all $|C|$ candidates as $c_\emptyset$ and all triplets $\{v_1,v_2,v_3\}\subseteq V$ of voters and do the above procedure. It is not hard to see that if $(C,V)$ contains the $3$-$8$ pattern, then the above algorithm finds it.
							
							To finish the proof, note that for each choice of $c_\emptyset,v_1,v_2,v_3$ we can in time $O(|C|)$ construct the set $|X|$. For each $d\in C$ we only need to test $d\succ_{v_i} c_\emptyset$ and this is equivalent to $\pos{v_i}{d}<\pos{v_i}{c_\emptyset}$. Hence the entire algorithm runs in $O(|V|^3|C|^2)$ time, as promised.
						\end{proof}
					}
					
					\toappendix
					{
						\sv
						{
							\section{Missing Material from section Convex Hull}
						}
					}
					\section{Convex Hull}\label{sec:forbidden_substructures}
					Suppose that $(C,V)$ is $2$-Euclidean and let $\gamma$ be a nice $2$-Euclidean embedding of $(C,V)$. Since there are finitely many voters, the convex hull $\conv{\gamma(V)}$ is a convex polygon. Let $\gamma(V)\cap \boundary{\conv{\gamma(V)}}=\{p_1,p_2,\ldots,p_k\}$ be the list of vertices of the polygon sorted in counterclockwise order. Note that since $\gamma$ is nice, no two consecutive sides of the polygon are parallel. Two points $p_i,p_j$ are \emph{consecutive} if $|j-i|=1\vee \{i,j\}=\{1,k\}$. We refer to voters $v$ with $\gamma(v)=p_i$ as being \emph{on the convex hull} and two voters corresponding to two consecutive points as \emph{consecutive}.
					
					\begin{definition}
						Let $V'\subseteq V$ be a set of voters and $a,b\in C$ two candidates. We say that $V'$ is \emph{controversial for $a$ over $b$} if and only if all voters in $V'$ prefer $a$ over $b$ while all voters in $V\setminus V'$ prefer $b$ over $a$. We say that $V'$ is \emph{controversial} if it is controversial for $a$ over $b$ for some $a,b\in C$. A voter $v$ is controversial if the singleton $\{v\}$ is controversial.
					\end{definition}
					
					\sv
					{
						\begin{lemma}[$\star$]\label{lem:controversial_lies_on_convex_hull}
						}
						\lv
						{
							\begin{lemma}\label{lem:controversial_lies_on_convex_hull}
							}
							Let $(C,V)$ be $2$-Euclidean election and $\gamma$ be a nice $2$-Euclidean embedding of $(C,V)$. If $v\in V$ is controversial, then $v$ is on the convex hull.
						\end{lemma}
						\toappendix
						{
							\sv
							{
								\begin{proof}[Proof of \Cref{lem:controversial_lies_on_convex_hull}]
								}
								\lv
								{
									\begin{proof}
									}
									Let $P=\conv{\gamma(V)}$ be the convex polygon corresponding to the convex hull of the voters. Let $a,b\in C$ be two candidates such that $v$ is controversial for $a$ over $b$. Note that this assumption implies that $v$ is the only voter with $\gamma(v)\in \halfplane{\gamma}{a}{b}$. Suppose that $\gamma(v)$ is not on the boundary of $P$. We argue that this implies that there is a vertex $p$ of $P$ that also lies in $\halfplane{\gamma}{a}{b}$. To see this note that if all vertices of $P$ were in $\halfplane{\gamma}{b}{a}$, then the entire polygon $P$ is in $\halfplane{\gamma}{b}{a}$, which contradicts the fact that $\gamma(v)\in \halfplane{\gamma}{a}{b}\cap P$. Note that no voter can be embedded to the bisector $\bisector{\gamma}{a}{b}$. The vertex $p$ corresponds to a voter $\gamma^{-1}(p)$ that also prefers $a$ over $b$ and is distinct from~$v$, thus contradicting the assumption that $v$ was the only one who prefers $a$ over $b$.
								\end{proof}
							}
							\sv
							{
								\begin{lemma}[$\star$]\label{lem:two_controversial_consecutive_convex_hull}
								}
								\lv
								{
									\begin{lemma}\label{lem:two_controversial_consecutive_convex_hull}
									}
									Let $(C,V)$ be $2$-Euclidean election and $\gamma$ be a nice $2$-Euclidean embedding of $(C,V)$. Let $u,v\in V$ be two voters such that the three sets $\{u,v\},\{u\},\{v\}$ are controversial. Then $u$ and $v$ are consecutive on the convex hull.
								\end{lemma}
								\toappendix
								{
									\sv
									{
										\begin{proof}[Proof of \Cref{lem:two_controversial_consecutive_convex_hull}]
										}
										\lv
										{
											\begin{proof}
											}
											Let $P=\conv{\gamma(V)}$ be the convex polygon corresponding to the convex hull of the voters. By assumption $u$ and $v$ are controversial. By \Cref{lem:controversial_lies_on_convex_hull} both $u$ and $v$ are on the boundary of $P$. We just have to show that $u$ and $v$ are consecutive. Suppose, for the sake of contradiction that they aren't and let $v_1,v_2\in V$ be two voters such that $v_1,v$ and $v_2,v$ are consecutive. Let $a,b\in C$ be two candidates such that $\{u,v\}$ is controversial for $a$ over $b$. This implies that $u$ and $v$ are the only voters with $\gamma(u),\gamma(v)\in \halfplane{\gamma}{a}{b}$. Hence $\gamma(v_1),\gamma(v_2)\in \halfplane{\gamma}{b}{a}$. We refer the reader to~\Cref{fig:pf_2_controversial_adjacent} for an illustration of the proof.
											
											Let $x_1$ and $x_2$ be the points of intersection of the line segments $\overline{\gamma(v)\gamma(v_1)}$ and $\overline{\gamma(v)\gamma(v_2)}$ with the bisector $\bisector{\gamma}{a}{b}$. Note that $x_1\neq x_2$ since we can assume that the three points $\gamma(v),\gamma(v_1),\gamma(v_2)$ do not lie on a common line.
											
											Let $T=\triangle \gamma(v) x_1 x_2$ and let $\alpha$ be the interior angle of $T$ at the vertex $\gamma(v)$. By convexity of $P$ we have $P\subseteq \alpha$. Since $u$ is on the boundary of $P$ and also in the halfplane $\halfplane{\gamma}{a}{b}$, this implies that $\gamma(u)$ lies on one of the line segments $\linesegment{\gamma(v)}{x_1},\linesegment{\gamma(v)}{x_2}$. But this contradicts the fact that $v_1,v$ and $v_2,v$ are consecutive. This finishes the proof.
										\end{proof}
										
										\begin{figure}[ht]
											\center
											\begin{tikzpicture}
												
												\coordinate (gv) at (0,4);
												\fill (gv) circle (2pt);
												\node[above left] at (gv){$\gamma(v)$};
												
												\coordinate (gv1) at (4,4);
												\fill (gv1) circle (2pt);
												\node[below right] at (gv1){$\gamma(v_1)$};	
												
												\coordinate (gv2) at (3,1);
												\fill (gv2) circle (2pt);
												\node[below] at (gv2){$\gamma(v_2)$};
												
												\coordinate (a) at (3.5714285714286,6);
												\coordinate (b) at (1,0);
												\draw[<->] (a) -- (b);
												\node[left] at (a){$\bisector{\gamma}{a}{b}$};
												
												\draw(gv) -- (gv1);
												\draw(gv) -- (gv2);

												\coordinate (x1) at (2.7142857142,4);
												\coordinate (x2) at (1.9,2.1);

												\fill (x1) circle (2pt);
												\node[below right] at (x1){$x_1$};	
												\fill (x2) circle (2pt);
												\node[right] at (x2){$x_2$};
												
												\draw[fill=black!20] (gv) -- ($(gv)!8mm!(x1)$) to [bend left] ($(gv)!8mm!(x2)$) -- cycle;
												\node[xshift=15,yshift=-5] at (gv){$\alpha$};
												
											\end{tikzpicture}
											\caption{The situation in the proof of \Cref{lem:two_controversial_consecutive_convex_hull}.}\label{fig:pf_2_controversial_adjacent}
										\end{figure}
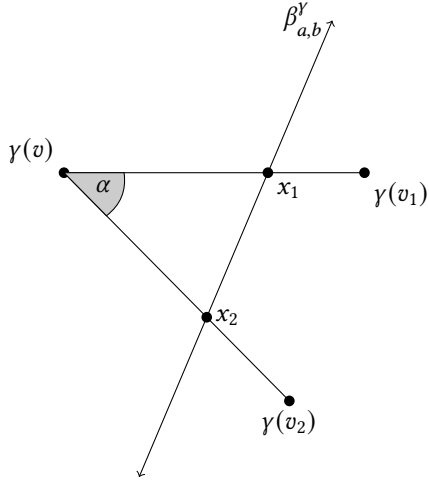
									}
									For an election $(C,V)$ we define the \emph{controversity graph} $\controversity{C}{V}$ as follows. The vertex set $V(\controversity{C}{V})$ consists of controversial voters and there is an edge $\{u,v\}\in E(\controversity{C}{V})$ between two controversial voters if and only if the set $\{u,v\}\subseteq V$ is controversial. Rephrasing in terms of the convex hull and \Cref{lem:controversial_lies_on_convex_hull,lem:two_controversial_consecutive_convex_hull}, the vertices of $\controversity{C}{V}$ are the voters that are necessarily on the convex hull and the edges correspond to consecutiveness on the convex hull for a nice $2$-Euclidean embedding $\gamma$ of $(C,V)$ (if it exists). This allows us to give a simple graph-theoretical characterization of $\controversity{C}{V}$.
									
									\begin{theorem}\label{thm:controversity_graph_property}
										Let $(C,V)$ be $2$-Euclidean election and let $G=\controversity{C}{V}$ be the controversity graph for $(C,V)$. Then $\Delta(G)\leq 2$ and if $G$ contains a cycle, then it is connected.
									\end{theorem}
									\begin{proof}
										Since $(C,V)$ is $2$-Euclidean, by \Cref{thm:non_degenerate_embedding} there is a nice $2$-Euclidean embedding of $(C,V)$. Suppose that there is a vertex $v\in V(G)$ of degree at least $3$ and let $v_1,v_2,v_3$ be three distinct neighbors of $v$ in $G$. Note that this implies that the pairs $(v,v_1),(v,v_2),(v,v_3)$ are consecutive on the convex hull. However, this is impossible. Hence $\Delta(G)\leq 2$.
										
										For the second property suppose that $G$ contains a cycle $Y$. Then this cycle uniquely determines the convex hull $\conv{\gamma(V)}$ and any other voter cannot be on the boundary. Hence if $G$ was disconnected and $v$ was a voter not on $Y$, then $v$ is, by \Cref{lem:controversial_lies_on_convex_hull}, on the convex hull. Contradiction.
									\end{proof}
									
									\Cref{thm:controversity_graph_property} gives us a refutation procedure for $2$-Euclideanness of an election. Given $(C,V)$, construct the controversity graph $G=\controversity{C}{V}$. If $G$ does not satisfy the properties of \Cref{thm:controversity_graph_property}, then $(C,V)$ is not $2$-Euclidean.
									
									
									\toappendix
									{
										We \lv{now} present examples where the controversity graph allows us to easily show that given election is not $2$-Euclidean.
										
										\begin{example}\label{example:refuting_using_controversity_graph}
											Consider the election with seven candidates $C=\{a,b,c,d,e,f,g\}$ and four voters $V=\{v_1,v_2,v_3,v_4\}$ where the preferences are:
											
											\begin{minipage}{0.35\columnwidth}
												\begin{align*}
													v_1&: d\succ g\succ c\succ f\succ a\succ e\succ b \\
													v_2&: g\succ c\succ b\succ a\succ e\succ d\succ f \\
													v_3&: c\succ b\succ a\succ d\succ f\succ g\succ e \\
													v_4&: d\succ b\succ a\succ e\succ g\succ c\succ f \\
												\end{align*}
											\end{minipage}
											\begin{minipage}{0.35\columnwidth}
												\begin{tikzpicture}
													\coordinate (v1) at (0,0);
													\fill (v1) circle (4pt);
													\node[left,xshift=-5] at (v1){$v_1:a \sqsupset b$};

													\coordinate (v2) at (2,0);
													\fill (v2) circle (4pt);
													\node[right,xshift=5] at (v2){$v_2: g\sqsupset d$};

													\coordinate (v3) at (0,2);
													\fill (v3) circle (4pt);
													\node[above, yshift=5] at (v3){$v_3: c\sqsupset g$};

													\coordinate (v4) at (2,-2);
													\fill (v4) circle (4pt);
													\node[below,yshift=-5] at (v4){$v_4: a\sqsupset c$};
													
													\draw (v1) -- (v2) node[midway,above] {$g\sqsupset a$};
													\draw (v1) -- (v4) node[sloped,midway,below] {$d\sqsupset a$};
													\draw (v2) -- (v4) node[sloped,midway,above] {$e\sqsupset f$};
													\draw (v3) -- (v1) node[sloped,midway,below] {$f\sqsupset e$};
													\draw (v3) -- (v2) node[sloped,midway,above] {$a\sqsupset d$};
												\end{tikzpicture}
											\end{minipage}
											
											The controversity graph $G=\controversity{C}{V}$ is depicted on the right. The notation $v_i:x\sqsupset y$ represents the fact that $v_i$ is controversial for $x$ over $y$. Note that $x$ and $y$ might not be unique. For example $v_1$ is also controversial for $f$ over $a$. The label $x\sqsupset y$ on an edge from $v_i$ to $v_j$ corresponds to the fact that $\{v_i,v_j\}\subseteq V$ is controversial for $x$ over $y$. We can see that $\deg_G v_1 = 3$ and hence by \Cref{thm:controversity_graph_property}, $(C,V)$ is not $2$-Euclidean. \qed
										\end{example}
										
										\begin{example}\label{example:refuting_using_controversity_graph_2}
											Consider extending the election from \Cref{example:refuting_using_controversity_graph} by three voters $v_5,v_6,v_7$:
											
											\begin{minipage}{0.35\columnwidth}
												\begin{align*}
													v_1&: d\succ g\succ c\succ f\succ a\succ e\succ b \\
													v_2&: g\succ c\succ b\succ a\succ e\succ d\succ f \\
													v_3&: c\succ b\succ a\succ d\succ f\succ g\succ e \\
													v_4&: d\succ b\succ a\succ e\succ g\succ c\succ f \\
													v_5&: d\succ c\succ b\succ g\succ f\succ e\succ a \\
													v_6&: c\succ d\succ b\succ a\succ e\succ g\succ f \\
													v_7&: d\succ g\succ c\succ a\succ b\succ e\succ f \\
												\end{align*}
											\end{minipage}
											\begin{minipage}{0.35\columnwidth}
												\begin{tikzpicture}
													\coordinate (v1) at (0,0);
													\fill (v1) circle (4pt);
													\node[above,yshift=5] at (v1){$v_1:b\sqsupset e$};
													
													\coordinate (v2) at (2,0);
													\fill (v2) circle (4pt);
													\node[above,yshift=5] at (v2){$v_2: e\sqsupset d$};
													
													\coordinate (v4) at (4,0);
													\fill (v4) circle (4pt);
													\node[above,yshift=5] at (v4){$v_4: a\sqsupset c$};
													
													\coordinate (v5) at (0,-2);
													\fill (v5) circle (4pt);
													\node[below,yshift=-5] at (v5){$v_5: e \sqsupset a$};
													
													\coordinate (v3) at (2,-2);
													\fill (v3) circle (4pt);
													\node[below,yshift=-5] at (v3){$v_3: f\sqsupset g$};
													
													\draw (v1) -- (v5) node[sloped,midway,above] {$f\sqsupset a$};
													\draw (v2) -- (v3) node[sloped,midway,above] {$a\sqsupset d$};
												\end{tikzpicture}
											\end{minipage}
											
											Here the controversity graph has only $5$ vertices and two edges (depicted on the right). The labels on the vertices and edges have the same meaning as in \Cref{example:refuting_using_controversity_graph}. \Cref{thm:controversity_graph_property} cannot be directly applied to refute the instance $(C,V)$. However, we already know that $(C,V)$ is not $2$-Euclidean.
											\qed
										\end{example}
									}
									
									\toappendix
									{
										\sv
										{
											\section{Missing Material from section Reducing the number of candidates}
										}
									}
									\section{Reducing the number of candidates}\label{sec:reducing_the_number_of_candidates}
									
									In this section, we describe how to reduce the number of candidates in an instance while preserving $2$-Euclideanness. We utilize the framework of \emph{reduction rules}. Formally a reduction rule is an algorithm that, given an instance $(C,V)$, outputs a new instance $(C',V')$ that is $2$-Euclidean if and only if $(C,V)$ is $2$-Euclidean. The goal of a reduction rule is to simplify the structure of the input instance, for example decrease the number of candidates or voters. This will be the case in this section.
									
									We emphasize that all our reduction rules are constructive. That is, given either of $(C,V)$ or $(C',V')$ our proofs also provide an algorithm how to construct an embedding of one election given the embedding of the other one, i.e., we provide the so-called \emph{solution lifting} algorithm.
									
									
									\dtoappendix
									{
										\subsection{Removing candidates that behave similarly}
									}
									We begin with a simple observation that a candidate that is always ranked the last in all votes can be safely removed. This is because this candidate can be always embedded `far enough' from all other candidates and voters.
									\begin{reductionrule}\label{rr:rr1}
										Let $(C,V)$ be the input election. If there is a candidate $a\in C$ such that for all $b\in C\setminus \{a\}$ and all voters $v\in V$ we have $b\succ_v a$, then output the subelection induced by $C\setminus \{a\}$.
									\end{reductionrule}
									We now generalize \Cref{rr:rr1} to allow for the removal of more candidates who are ranked last. For this purpose we introduce new terminology and notation. Let $(C,V)$ be an election. For $C'\subseteq C$ and $V'\subseteq V$ let $\pos{V'}{C'}=\bigcup_{v\in V'}\bigcup_{c\in C'}\left\{\pos{v}{c}\right\}$. For a set of indices $I\subseteq [|C|]$ let $\posi{I}=\{c\in C \mid \exists v\in V: \pos{v}{c}\in I\}$.
									We say that a set of candidates $S\subseteq C$ induces a \emph{block} if the set of positions where~$S$ occurs in a vote is the same for each vote and moreover constitutes an interval. More formally, a nonempty $S\subseteq C$ is a \emph{block} (in $(C,V)$) if there are two indices $i,j\in[|C|]$, $i\leq j$, such that $\pos{V}{S}=[i,j]$. We equivalently refer to a block by the discrete interval $[i,j]$ such that the set of candidates occuring between positions $i$ and $j$ is the same for all votes. Formally $[i,j]$ is a block (in $(C,V)$) if $|\posi{[i,j]}| =|[i,j]|=j-i+1$. We say that $S$ is a \emph{tail block} if $\max\pos{V}{S}=|C|$, that is $j=|C|$. 
									\begin{example}
										If $a\in C$ is a candidate ranked the last in all votes, then $\{a\}$ is a tail block of size $1$. If $|V|=1$, then any subinterval $I\subseteq [1,|C|]$ is a block. The interval $[1,|C|]$ is always a block in any election $(C,V)$.
									\end{example}
									Let $(C,V)$ be an election and $C_1,C_2\subseteq C$ two sets of candidates. We say that $C_1$ \emph{copies} $C_2$ if there is a bijection $f\colon C_2 \to C_1$ such that for all $a,b\in C_2$ and all votes $v\in V$ we have $a\succ_v b$ if and only if $f(a)\succ_v f(b)$.
									\begin{reductionrulep}{\ref*{rr:rr1}$+$}\label{rr:rr1plus}
										Let $(C,V)$ be the input election. If there is a tail block $S$ of size at most $3$ and there is $S'\subseteq C\setminus S$ that copies $S$, then output the subelection induced by $C\setminus S$.
									\end{reductionrulep}

									\sv
									{
										\begin{proposition}[$\star$]\label{prop:rr1plus}
										}
										\lv
										{
											\begin{proposition}\label{prop:rr1plus}
											}
											\Cref{rr:rr1plus} is correct.
										\end{proposition}
										\toappendix
										{
											\sv
											{
												\begin{proof}[Proof of \Cref{prop:rr1plus}]
												}
												\lv
												{
													\begin{proof}
													}
													Let $(C',V')$ denote the subelection induced by $C\setminus S$. We show that $(C,V)$ is $2$-Euclidean if and only if $(C',V')$ is.
													If $(C,V)$ is $2$-Euclidean, then $(C',V')$ is $2$-Euclidean by \Cref{obs:euclidean_is_hereditary}. For the other direction suppose that $(C',V')$ is $2$-Euclidean and let $\gamma'\colon C'\cup V' \to \mathbb{R}^2$ be a $2$-Euclidean embedding of $(C',V')$. We define a $2$-Euclidean embedding $\gamma \colon C\cup V\to \mathbb{R}^2$ as an extension of $\gamma'$. For $c\in C\setminus S$ let $\gamma(c)=\gamma'(c)$ and for $v\in V$ let $\gamma(v)=\gamma'(v[C'])$. Note that $|V[C']| = |V|$, because if $u,v\in V$ are distinct, then $u[C']\neq v[C']$. It follows that $\gamma$ is injective. It is now left to embed the candidates from~$S$.
													
													Let $B'$ be an open ball that has the points in $\gamma'(S')$ on the boundary and let $h$ denote its center. Note that if $|S'|=1$, then $h$ can be chosen arbitrarily. If $|S'|=2$, then $h$ lies on the bisector between the two candidates in $S'$. In the case $|S'|=3$ the center $h$ is given by the circumcenter of the triangle with vertices $\gamma'(S')$ and $\boundary{B'}$ is the circumcircle of this triangle. Note that we can assume that no three candidates are embedded to a line by \Cref{thm:non_degenerate_embedding} so in the case $|S'|=3$ the three points in $\gamma'(S')$ indeed form a triangle. We refer the reader to \Cref{fig:rr3_proof} for visualisation of the proof.
													
													Since $S'$ copies $S$, there is a bijection $f\colon S \to S'$ such that for $a,b\in S$ and $v\in V$ we have $a\succ_v b$ if and only if $f(a)\succ_v f(b)$.
													Our goal is to embed the candidates in $S$ such that for $a,b\in S$ the bisectors $\bisector{\gamma}{a}{b}$ coincide with the bisectors $\bisector{\gamma'}{f(a)}{f(b)}$. Note that any homothety with center $h$ and ratio $k>0$ preserves the perpendicular bisectors between pairs of points that are equidistant from $h$, in particular between pairs of points from $\gamma'(S')$. This is true because all such bisectors contain the center $h$ and lines containing the center are invariant under homotheties. We shall now define the homothety $T$ and we set $\gamma(c)=(T\circ \gamma')(f(c))$ for $c\in S$. 
													
													The center for $T$ is $h$ and it is only left to determine the ratio. Note that we must ensure that for all candidates $c'\in C\setminus S$ and $c\in S$ and all voters $v$ we have $c\succ_v c'$ so we must place the voters in $S$ far enough. The set $\gamma'(C'\cup V')\subseteq \mathbb{R}^2$ is bounded, so by definition, there is an open ball $B_r(h)$ centered at $h$ with radius $r$ such that $\gamma'(C'\cup V')\subseteq B_r(h)$. We let the ratio for the homothety $T$ be equal to $\frac{3r+1}{\ell_2(h,\gamma(f(a)))}$.
													
													We now verify that $\gamma$ is a valid $2$-Euclidean embedding of $(C,V)$. By assumption $S'$ copies $S$ and for any $a,b\in S$ we have $\bisector{\gamma}{a}{b} = \bisector{\gamma}{f(a)}{f(b)}$. It follows that for any $a,b\in S,v\in V$:
													\begin{align*}
														&\ell_2(\gamma(v),\gamma(a))<\ell_2(\gamma(v),\gamma(b)) \\ &\Leftrightarrow 
														\ell_2(\gamma(v),\gamma(f(a)))<\ell_2(\gamma(v),\gamma(f(b))) \\  &\Leftrightarrow \ell_2(\gamma'(v[C']),\gamma'(f(a)))<\ell_2(\gamma'(v[C']),\gamma'(f(b))) \\&\Leftrightarrow
														f(a)\succ_{v[C']} f(b) \\ &\Leftrightarrow f(a)\succ_v f(b) \\ &\Leftrightarrow
														a \succ_v b 
													\end{align*}
													The first equivalence is due to the fact that $\bisector{\gamma}{a}{b}=\bisector{\gamma}{f(a)}{f(b)}$, the second is the definition of $\gamma$. Third is the definition of $2$-Euclidean embedding. Fourth is the definition of a restricted vote and finally, fifth is the property of $f$.
													
													It now remains to argue that for all $c\in S$ and $c'\in C\setminus S$ and all voters $v\in V$ we have $\ell_2(\gamma(v),\gamma(c'))<\ell_2(\gamma(v),\gamma(c))$. To see this note that all voters and candidates except $S$ are embedded inside $B_r(h)$. Thus $\ell_2(\gamma(v),\gamma(c')) < 2r$. The homothety $T$ scaled all distances from $h$ by $k$. Since $h$ is the circumcenter of the triangle $\triangle\gamma(a')\gamma(b')\gamma(c')$, we have $\ell_2(h,\gamma(a))=\ell_2(h,\gamma(b))=\ell_2(h,\gamma(c))= k \cdot \ell_2(h,\gamma(f(a)))=3r+1$. Hence all three candidates $a,b,c$ are at distance $3r+1$ from $h$. In particular at distance at least $2r+1$ from any point in $B_r(h)$. Hence 
													$\ell_2(\gamma(v),\gamma(c))\geq 2r+1$ for $c\in S$ and any $v\in V$. This finishes the proof.
												\end{proof}
											}
											\toappendix
											{
												\begin{figure}[ht]
													\centering
													\begin{tikzpicture}[scale=0.6]
														
														\coordinate (h) at (0,0);
														
														\coordinate (gap) at (1,0);
														\coordinate (gbp) at (-0.8,0.6);
														\coordinate (gcp) at (-0.4705882352941,-0.8823529411765);
														
														\fill (gap) circle (2pt);
														\node[right] at (gap){$\gamma(a')$};
														
														\fill (gbp) circle (2pt);
														\node [left] at (gbp) {$\gamma(b')$}; 
														
														\fill (gcp) circle (2pt);
														\node[below] at (gcp){$\gamma(c')$};
														
														\draw[opacity=0.5] (gap) -- (gbp);
														\draw[opacity=0.5] (gbp) -- (gcp);
														\draw[opacity=0.5] (gcp) -- (gap);
														
														\coordinate (ga) at (5.5,0);
														\coordinate (gb) at (-4.4,3.3);
														\coordinate (gc) at (-2.5882352941176,-4.8529411764706);
														
														\fill (ga) circle (2pt);
														\node[right] at (ga){$\gamma(a)$};
														
														\fill (gb) circle (2pt);
														\node[above left] at (gb){$\gamma(b)$};
														
														\fill (gc) circle (2pt);
														\node[below] at (gc){$\gamma(c)$};
														
														\draw[help lines] (ga) -- (gb);
														\draw[help lines] (gb) -- (gc);
														\draw[help lines] (gc) -- (ga);

														\coordinate (bab1) at (1.3333333333333,4);
														\coordinate (bab2) at (-2,-6);
														
														\coordinate (bac1) at (-2.4,4);
														\coordinate (bac2) at (3.6,-6);
														
														\coordinate (bbc1) at (-6,-1.33333333333);
														\coordinate (bbc2) at (6,1.33333333333);
														
														\draw[help lines,dashed] (bab1) -- (bab2);
														\draw[help lines,dashed] (bac1) -- (bac2);
														\draw[help lines,dashed] (bbc1) -- (bbc2);

														\fill (h) circle (2pt);
														\node[above,yshift=10] at (h){$h$};

														\draw[thick,red,fill=red,fill opacity=0.1] (h) circle (1.5);
														\node [above left,red] at (-1.42942312,0.454697186){$\openballrc{r}{h}$};

														\draw (h) circle (1);
														\node [above,yshift=27pt] at (0,0) {$\boundary{B'}$};
														
														

													\end{tikzpicture}
													\caption{Situation in the proof of correctness of \Cref{rr:rr1plus} where $S=\{a,b,c\}$ and $S'=\{a',b',c'\}$, where $a'=f(a),b'=f(b),c'=f(c)$. Note that $h$ -- the circumcenter of the triangle $\triangle\gamma'(a')\gamma'(b')\gamma'(c')$ is given by the intersection of the three bisectors $\bisector{\gamma'}{a'}{b'}\cap\bisector{\gamma'}{b'}{c'}\cap \bisector{\gamma'}{a'}{c'}$ (represented as dashed lines). The black circle $\boundary{B'}$ is the circumcircle of the triangle $\triangle\gamma'(a')\gamma'(b')\gamma'(c')$ and the red ball $\openballrc{r}{h}$ contains all the points of $\gamma'(C'\cup V')$. }\label{fig:rr3_proof}
												\end{figure}
											}
											We proceed to generalize this rule even further and introduce the notion of block decomposition. With \Cref{rr:rr1plus} we are able to remove tail blocks of size at most $3$. With block decomposition we will be able to remove such blocks even if they are not tail.
											
											\begin{definition}[Block decomposition]
												Let $(C,V)$ be an election and $k$ a positive integer. We say that a sequence of blocks $\mathcal{I}=([i_1,j_1],[i_2,j_2],\ldots,[i_t,j_t])$ in $(C,V)$ forms a \emph{$k$-block decomposition (of $(C,V)$)} if
												\begin{enumerate}
													\item all blocks $[i_\ell,j_\ell]$ are of size at most $k$,
													\item the blocks are consecutive, that is, for all $\ell\in[t-1]$ we have $j_{\ell}+1=i_{\ell+1}$, and
													\item $[i_t,j_t]$ is a tail block.
												\end{enumerate}
											\end{definition}

											Note that it is possible that $i_1\neq 1$. Observe that some elections may admit multiple $k$-block decompositions for some $k$. For example the single interval $([1,|C|])$ is always a $|C|$-block decomposition for any election $(C,V)$. And if $|V|=1$, then there are as many $k$-block decompositions of $(C,V)$ as there are ways to partition the set $[|C|]$ into nonempty intervals of size at most $k$. However, if we require that the decomposition contains maximum number of blocks, then it is unique (\Cref{lem:maximality_block}) and we refer to it as the \emph{maximal $k$-block decomposition} of $(C,V)$.

											\sv
											{
												\begin{lemma}[$\star$]\label{lem:maximality_block}
												}
												\lv
												{
													\begin{lemma}\label{lem:maximality_block}
													}
													Let $(C,V)$ be an election and $k$ a positive integer. Then the $k$-block decomposition containing maximum number of blocks is unique.
												\end{lemma}
												\toappendix
												{
													\lv{
														\begin{proof}	
														}
														\sv
														{
															\begin{proof}[Proof of \Cref{lem:maximality_block}]
															}
															Let $t$ be the maximum size of a $k$-block decomposition of $(C,V)$ and for the sake of contradiction suppose that there are two distinct $k$-block decompositions $\mathcal{I} = ([i_1,j_1],\ldots,[i_t,j_t])$ and $\mathcal{I}'=([i'_1,j'_1],\ldots,[i'_t,j'_t])$. Let $\ell$ be the largest index such that $[i_{\ell},j_{\ell}] \neq [i'_{\ell},j'_{\ell}]$. Note that $j_{\ell}=j'_{\ell}$ otherwise it is a contradiction with the choice of $\ell$. Without loss of generality assume that $i_{\ell}<i'_{\ell}$. 
															
															We create a new $k$-block decomposition $\mathcal{I}^*$ from $\mathcal{I}$ by replacing the block $[i_{\ell},j_{\ell}]$ by two blocks $[i_{\ell},i'_{\ell}-1],[i'_{\ell},j_{\ell}]$ (see \Cref{fig:lemma_maximality_block} for visualisation). We now verify that $\mathcal{I}^*$ is also a $k$-block decomposition thus contradicting the maximality of both $\mathcal{I}$ and $\mathcal{I}'$. The new blocks are still of size at most $k$. It remains to argue that the two new intervals are indeed blocks in $(C,V)$. Note that $[i'_{\ell},j_{\ell}]=[i'_{\ell},j'_{\ell}]$ is a block since it is a block in $\mathcal{I}'$. Denote $I = [i_\ell,i_{\ell}'-1], I_1 = [i_\ell,j_\ell], I_2 = [i'_{\ell-1},j'_{\ell-1}]$. We show that $|\posi{I}|=|I|$. Note that 
															$\posi{I}=\posi{I_1\cap I_2} = \{c\in C | \exists v \in V: \pos{v}{c}\in I_1\cap I_2\}=\posi{I_1}\cap\posi{I_2}=C_1\cap C_2$. Hence $|\posi{I}|=|C_1\cap C_2|$.
															
															The inequality $|\posi{I}|\geq |I|$ holds for any $I\subseteq [|C|]$ as there cannot be less than $|I|$ candidates on $|I|$ positions. On the other hand, we show that $|I| = |C_1\cap C_2|$. To see this note that $C_1\cup C_2$ is also a block since $\pos{V}{C_1\cup C_2}=\pos{V}{C_1}\cup\pos{V}{C_2}=I_1\cup I_2 = [i'_{\ell - 1},j'_{\ell}]$. Hence $|C_1\cup C_2|=|[i'_{\ell-1},j'_{\ell}]|=|I_1|+|I_2|-|I_1\cap I_2|=|I_1|+|I_2|-|I|=|C_1|+|C_2|-|I|$. Since $|C_1\cup C_2|=|C_1|+|C_2|-|C_1\cap C_2|$ it follows that $|I|=|C_1\cap C_2|$.
															
															Putting this together we obtain $|I|=|C_1\cap C_2|=|\posi{I}|\geq |I|$. Hence, the inequality hold with an equality and thus $I$ is a block and this finishes the proof.
														\end{proof}
														\begin{figure}
															\centering
															\begin{tikzpicture}
																\draw (0,0) rectangle (6,0.5);
																\node[above] at (6,0.5){$j'_\ell$};
																
																\draw (4,0.5)--(4,0);
																
																\draw (1,0.5)--(1,0);
																\node[above] at (4,0.5){$i'_\ell$};
																\node[above] at (1,0.5){$i'_{\ell-1}$};
																
																\draw (3,0.5)--(3,0);
																\node[below] at (3,0){$i_\ell$};
																\node[below] at (6,0){$j_\ell$};

																\draw[decorate,decoration={brace,amplitude=5pt,raise=15pt}] (1,0.5) -- (4,0.5) node [midway,yshift=25pt]{$I_2$};
																
																\draw[decorate,decoration={brace,amplitude=5pt,raise=15pt}] (6,0) -- (3,0) node [midway,yshift=-26pt]{$I_1$};
																
															\end{tikzpicture}
															\caption{Visualisation of the proof of \Cref{lem:maximality_block}.}\label{fig:lemma_maximality_block}
														\end{figure}
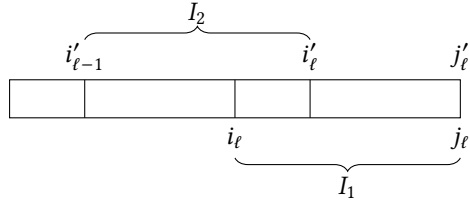
														
													}

													\begin{reductionrulep}{\ref*{rr:rr1}$++$}\label{rr:rr1plusplus}
														Let $(C,V)$ be the input election and let $\mathcal{I} = (I_1,I_2,\ldots,I_t)$ be its maximal $3$-block decomposition. Denote $S_\ell = \posi{I_\ell}$. If there is a block $I_\ell$ such that there is $S'\subseteq C\setminus S_\ell$ that copies $S_\ell$, then output the subelection induced by $C\setminus S_\ell$.
													\end{reductionrulep}

													\sv{
														\begin{proposition}[$\star$]\label{prop:correctness_of_1pp}
														}
														\lv
														{
															\begin{proposition}\label{prop:correctness_of_1pp}
															}
															
															\Cref{rr:rr1plusplus} is correct.
														\end{proposition}
														
														\toappendix
														{
															\lv
															{
																\begin{proof}
																}
																\sv
																{
																	\begin{proof}[Proof of \Cref{prop:correctness_of_1pp}]
																	}
																	
																	Let $(C',V')$ denote the subelection induced by $C\setminus S_\ell$. Clearly if $(C,V)$ is $2$-Euclidean then $(C',V')$ is also $2$-Euclidean by \Cref{obs:euclidean_is_hereditary}. Now suppose that $(C',V')$ is $2$-Euclidean and let $\gamma'$ be a $2$-Euclidean embedding of $(C',V')$. We have, by the same argument as in the proof of correctness of \Cref{rr:rr1plus} that $|V|=|V[C']|$. We construct a $2$-Euclidean embedding $\gamma\colon C\cup V \to \mathbb{R}^2$ for $(C,V)$ as follows. For all $v\in V$ we let $\gamma(v)=\gamma'(v[C'])$. For all candidates~$c$ that are ranked before the candidates in $S_\ell$ in all votes we let $\gamma(c)=\gamma'(c)$. As in the proof of correctness of \Cref{rr:rr1plus} we let $h$ be the center of the open ball $B'$ that has the points in $\gamma'(S')$ on the boundary. Our task is to embed $S_\ell$ in a similar way as in \Cref{rr:rr1plus} via a homothety with center $h$. Since $S'$ copies $S_\ell$ we let $f\colon S_\ell \to S'$ be the bijection such that $\forall v \in V: a\succ_v b \Leftrightarrow f(a)\succ_v f(b)$. Let $\openballrc{r}{h}$ be the open ball containing the so-far embedded candidates (i.e., those that are before the block $I_\ell$) and all voters and we let $\gamma(a)=(T\circ\gamma')(f(a))$, where $T$ is homothety with center $h$ and ratio $\frac{3r+1}{\ell_2(h,\gamma'(f(a)))}$. 
																	
																	Now it only remains to embed the candidates from the remaining blocks $S_{\ell+1},\ldots,S_{t}$. We can do this in a similar fashion by embedding them one by one, but the `reference' points are the points $\gamma'(S_i)$. To be more precise, to embed the candidates $S_i$ for $i\in\{\ell+1,\ldots,t\}$, we consider the points $\gamma'(S_i)$ and let $B_i$ be the ball containing the points $\gamma'(S_i)$ on the boundary and let $h_i$ be its center. Let $\openballcr{h_i}{r_i}$ be an open ball containing all candidates before the $i$-th block and all the voters. For $a\in S_i$ we let $\gamma(a)=(T_i\circ \gamma')(a)$ where $T_i$ is the homothety with center $h_i$ and ratio $\frac{3r_i+1}{\ell_2(h_i,\gamma'(a))}$. By doing this procedure for all $i\in\{\ell+1,\ldots,t\}$ in increasing order, we obtain an embedding $\gamma$ of $(C,V)$.
																	
																	It remains to verify that $\gamma$ is a valid $2$-Euclidean embedding of $(C,V)$. Note that if $a,b\in S_\ell$ we have for all voters $v\in V$:
																	\[
																	\ell_2(\gamma(v),\gamma(a))<\ell_2(\gamma(v),\gamma(b))\Leftrightarrow a\succ_v b
																	\]
																	by the same arguments as in the proof of \Cref{rr:rr1plus}.
																	
																	We now verify the positions of the candidates from blocks $S_{\ell+1},\ldots,S_t$. There are two cases to consider for $a,b\in C$ and a vote $v$:
																	\begin{description}
																		\item[Case 1:] $a,b$ are in two distinct blocks, say $a\in S_i,b\in S_j,j>i$. By the definition of blocks we have $a\succ_v b$ for all $v\in V$. However, note that $\gamma(a)$ and all voters are inside the ball $B=\openballcr{h_{j}}{r_{j}}$. Hence $\forall p,q\in B$ we have $\ell_2(p,q)<2r_j$ and by definition of $T_j$ we have $\forall p \in B: \ell_2(p,\gamma(b))\geq 2r_j$. In particular, for any voter $v$ we have $\ell_2(\gamma(v),\gamma(a))<2r_j\leq \ell_2(\gamma(v),\gamma(b))$. 
																		\item[Case 2:] $a,b$ are from the same block, say $a,b\in S_i$. Note that $\gamma'$ was $2$-Euclidean embdeding for the subelection $(C',V')$ and $S_i\subseteq C'$. Hence we have $a\succ_v b\Leftrightarrow \ell_2(\gamma'(v),\gamma'(a))<\ell_2(\gamma'(v),\gamma'(b))$. But since $\gamma$ is defined as a composition of $\gamma'$ and a homothety with center $h_i$ and all the points in~$\gamma'(S_i)$ have the same distance to $h_i$ (by definition of $h_i$ and the ball $B_i$), we have $\bisector{\gamma}{a}{b}=\bisector{\gamma'}{a}{b}$. Thus 
																		\[
																		\ell_2(\gamma'(v),\gamma'(a))<\ell_2(\gamma'(v),\gamma'(b))\Leftrightarrow \ell_2(\gamma(v),\gamma(a))<\ell_2(\gamma(v),\gamma(b)).
																		\]
																	\end{description}
																\end{proof}
															}
															\dtoappendix
															{
																\subsection{Removing candidates ranked next to each other}
															}
															This reduction rule deals with the scenario when there are two candidates that are ranked tightly next to each other in each vote. That is, there are two candidates $b,c$ such that for every $v\in V$ we have $b\succ_v c$ and there is no $d$ and voter $v$ such that $b\succ_v d \succ_v c$. Moreover, for the reduction rule to work there must be a candidate $a$ that is ranked before $b$ in all votes. Symbolically, each vote may be described by the regular expression $?a?bc?$, where $?$ denotes arbitrary (possibly empty) sequence of candidates (other than $a,b,c$).

															\begin{reductionrule}\label{rr:rr2}
																Let $(C,V)$ be the input election. If there are candidates $b,c\in C$ such that:
																\begin{enumerate}
																	\item for all votes $v$ we have $b\succ_v c$,
																	\item for all votes $v$ there is no $d\in C\setminus \{b,c\}$ such that $b\succ_v d$ and $d\succ_v c$, and
																	\item there exists candidate $a$ such that $a\succ_v b$ for all votes $v$,	
																\end{enumerate}
																then output the subelection induced by $C\setminus \{b\}$.
															\end{reductionrule}

															\lv{		
																\begin{proposition}\label{prop:rr2}
																}
																\sv{
																	\begin{proposition}[$\star$]\label{prop:rr2}	
																	}
																	\Cref{rr:rr2} is correct.
																\end{proposition}

																\toappendix
																{
																	\lv{
																		\begin{proof}
																		}
																		\sv{
																			\begin{proof}[Proof of \Cref{prop:rr2}]
																			}
																			Let $(C',V')$ be the subelection induced by $C\setminus \{b\}$. If $(C,V)$ is $2$-Euclidean, then~$(C',V')$ is clearly $2$-Euclidean. On the other hand, let $(C',V')$ be $2$-Euclidean and $\gamma'\colon C'\cup V'\to \mathbb{R}^2$ a $2$-Euclidean embedding of $(C',V')$. We define a $2$-Euclidean embedding $\gamma\colon C \cup V\to\mathbb{R}^2$ of $(C,V)$ as an extension of $\gamma'$. For $c\in C'=C\setminus \{b\}$ we let $\gamma(c)=\gamma'(c)$ and for $v\in V$ we let $\gamma(v)=\gamma'(v[C'])$. Note that $|V[C']| = |V|$, because if $u,v\in V$ are distinct, then $u[C']\neq v[C']$. We are left with the task of determining $\gamma(b)$. The idea is as follows.
																			Let~$L$ be the line segment with endpoints $\gamma(a),\gamma(c)$. Our aim is to embed~$b$ close to~$c$~on~$L$. We now formally describe how to do it.
																			
																			Let $d\in C\setminus \{b,c\}$ be any candidate (note that we may have $d=a$) and $v\in V$ an arbitrary vote. We define the \emph{feasible region} $F_d(v)$ as follows:
																			\[
																			F_d(v) = 
																			\begin{cases}
																				\left\{p\in \mathbb{R}^2\mid \ell_2(\gamma(v),p)<\ell_2(\gamma(v),\gamma(c))\right\} & \text{if $c\succ_v d$} \\
																				\left\{p\in \mathbb{R}^2 \mid \ell_2(\gamma(v),\gamma(d))<\ell_2(\gamma(v),p)<\ell_2(\gamma(v),\gamma(c)\right\} & \text{if $d\succ_v c$}
																			\end{cases}
																			\]
																			
																			Let $F=\bigcap_{d\in C\setminus \{b,c\}}\bigcap_{v\in V'}F_d(v)$. Note that condition $2$ can be rephrased as follows: For all candidates $d\in C\setminus \{b,c\}$ we have $d\succ_v c$ if and only if $d\succ_v b$. Notice that if we embed $\gamma(b)\in F\setminus \gamma(C'\cup V)$, then we are done. It remains to show that $F\setminus \gamma(C'\cup V)\neq \emptyset$. To achieve this, we show that $F$ is nonempty and infinite. Since $\gamma(C'\cup V)$ is finite, then $F\setminus \gamma(C'\cup V)$ is also nonempty.
																			
																			We show that for each feasible region $F_d(v)$ there is an open line segment $L_d(v)\subseteq L\cap F_d(v)$ whose one endpoint is $\gamma(c)$. We emphasize that the line segment is \emph{open}, i.e., it does not contain its endpoints. In other words, the line segment $L_d(v)$ is of the form 
																			\[
																			L_d(v) = \{\gamma(c)+ t\cdot (\gamma(a)-\gamma(c))\mid 0< t < \varepsilon \} 
																			\]
																			for some $\varepsilon=\varepsilon(d,v)$.
																			
																			We only need to distinguish two cases for a voter $v$ and the candidate $d$. See \Cref{fig:rr5_proof} for illustration.
																			\begin{description}
																				\item[Case 1:] $c\succ_v d$. In this case, the feasible region $F_d(v)$ is the open ball $B=\openballcp{\gamma(v)}{\gamma(c)}$. Note that $a\succ_v c$ implies that both $\gamma(a),\gamma(c)$ lie in the closed ball $\closure{B}$. By convexity of $\closure{B}$ we have $\linesegment{\gamma(a)}{\gamma(c)} = L\subseteq \closure{B}$. Hence we set $L_d(v)=L\setminus \{\gamma(a),\gamma(c)\}\subseteq B=F_d(v)$, thus $\varepsilon(d,v)=1$.
																				\item[Case 2:] $d\succ_v c$. In this case, the feasible region $F_d(v)$ is the open annulus $\annuluscpp{\gamma(v)}{\gamma(d)}{\gamma(c)}$. By the very same argument as in the previous case we have $L\subseteq \closure{\openballcp{\gamma(v)}{\gamma(c)}}$. Note that $d\succ_v c$ is equivalent to $\ell_2(\gamma(v),\gamma(d))<\ell_2(\gamma(v),\gamma(c))$. Hence $\annuluscpp{\gamma(v)}{\gamma(d)}{\gamma(c)} = \openballcp{\gamma(v)}{\gamma(c)}\setminus \closure{\openballcp{\gamma(v)}{\gamma(d)}}$. 
																				
																				Let $L'=L\setminus \{\gamma(a),\gamma(c)\}\setminus \closure{\openballcp{\gamma(v)}{\gamma(d)}}$. In other words, $L'$ is the set of points of $L$ excluding the endpoints and without the points that are in the closed inner ball of the annulus. There are three subcases to consider depending on the number of intersection points of the circle $\boundary{\openballcp{\gamma(v)}{\gamma(d)}}$ with the line segment $L$. See \Cref{fig:rr5_three_subcases} for illustration.
																				\begin{description}
																					\item[Subcase 2.1:] $L\cap \boundary{\openballcp{\gamma(v)}{\gamma(d)}}=\emptyset$. In this case, the entire line $L'$ is contained in the annulus $\annuluscpp{\gamma(v)}{\gamma(d)}{\gamma(c)}=F_d(v)$. Hence we set $L_d(v)=L'$. Therefore, $\varepsilon(d,v)=1$.
																					
																					\item[Subcase 2.2:] $L\cap \boundary{\openballcp{\gamma(v)}{\gamma(d)}}=\{p\}$. In this case, $L'$ is just an open line segment with one endpoint $\gamma(c)$. Hence we set $L_d(v)=L'$. Thus $\varepsilon(d,v)=\frac{\ell_2(\gamma(c),p)}{\ell_2(\gamma(c),\gamma(a))}$.
																					
																					\item[Subcase 2.3:] $L\cap \boundary{\openballcp{\gamma(v)}{\gamma(d)}}=\{p_1,p_2\}$. In this case, $L'$ is a disjoint union of two open line segments $L_1\cup L_2$ where (without loss of generality) $L_1$ has $\gamma(c)$ as one of its endpoints. We set $L_d(v)=L_1$. Thus, if $\ell_2(p_1,\gamma(c))<\ell_2(p_2,\gamma(c))$, then $\varepsilon(d,v) = \frac{\ell_2(\gamma(c),p_1)}{\ell_2(\gamma(c),\gamma(a))}$.
																				\end{description}
																			\end{description}
																			
																			By the two cases above, for each feasible region $F_d(v)$ there is a nonempty open subsegment $L_d(v)\subseteq F_d(v)\cap L$ of the form $\{\gamma(c)+ t\cdot (\gamma(a)-\gamma(c))\mid 0< t < \varepsilon \}$. Let $\varepsilon^*=\min_{d,v}\varepsilon(d,v)$. Then the open line segment $\{\gamma(c)+t\cdot (\gamma(a)-\gamma(c))\mid 0<t<\varepsilon^*\}$ is contained in all $L_d(v)$, in particular it is contained in $F$. Hence $F$ is infinite. By the arguments above this implies that $F\setminus \gamma(V)\neq \emptyset$ and the proof is finished.
																		\end{proof}
																		\begin{figure}[ht]
																			\center
																			\begin{tikzpicture}
																				
																				\coordinate (a) at (0,0);
																				\fill (a) circle (2pt);
																				\node[above left] at (a){$\gamma(a)$};
																				
																				\coordinate (c) at (1,0);
																				\draw (c) circle (2pt);
																				\node[below right] at (c){$\gamma(c)$};
																				
																				\draw (0.5,0) circle (2pt);

																				\coordinate (v1) at (-0.5,1);
																				\fill[blue] (v1) circle (2pt);
																				\node[below,blue] at (v1){$\gamma(v_1)$};

																				\coordinate (v2) at (0,1.5);
																				\fill[red] (v2) circle (2pt);
																				\node[above left,red] at (v2){$\gamma(v_2)$};

																				\coordinate (d) at (-0.5,3);
																				\fill (d) circle (2pt);
																				\node[above,shift={(0,0.2)}] at (d){$\gamma(d)$};
																				
																				\draw[help lines] (c) -- (d);

																				\draw[thick,blue,fill=blue,fill opacity=0.1] (v1) circle (1.8027756377319946);
																				
																				\draw[red,fill=red,fill opacity=0.1,even odd rule] (v2) circle (1.581138830084188) circle (1.8027756377319946);
																				
																				\coordinate (bcd1) at (-3,-0.125);
																				\coordinate (bcd2) at (2,2.375);
																				\draw[dashed,help lines] (bcd1) -- (bcd2);

																				\coordinate (bac1) at (0.5,-1);
																				\coordinate (bac2) at (0.5,4);
																				\draw[dashed,help lines] (bac1) -- (bac2);

																				\draw[ultra thick] (0.5,0) -- (1,0);
																				
																				\draw (0,0) -- (1,0);
																				\draw[decorate,decoration={brace,amplitude=5pt,raise=5pt}] (0,0) -- (1,0) node [midway,yshift=15pt]{$L$};

																			\end{tikzpicture}
																			\caption{The situation in \Cref{rr:rr2}. The blue open ball is the feasible region $F_d(v_1)=\openballcp{\gamma(v_1)}{\gamma(c)}$ and this corresponds to the Case 1 in the proof, i.e., $c\succ_{v_1} d$. The open annulus given by the two red circles is the feasible region $F_d(v_2)=\annuluscpp{\gamma(v_2)}{\gamma(d)}{\gamma(c)}$ and this corresponds to the Case 2 in the proof, i.e., $d\succ_{v_2}c$. The thick black line is the segment $L_d(v_2)$ (note that it does not contain the endpoints). The dashed lines represent the bisectors $\bisector{\gamma}{a}{c}$ and $\bisector{\gamma}{c}{d}$.}\label{fig:rr5_proof}
																		\end{figure}
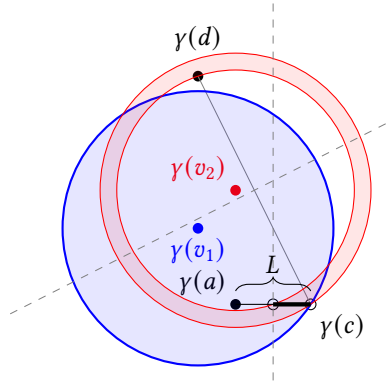
																		
																		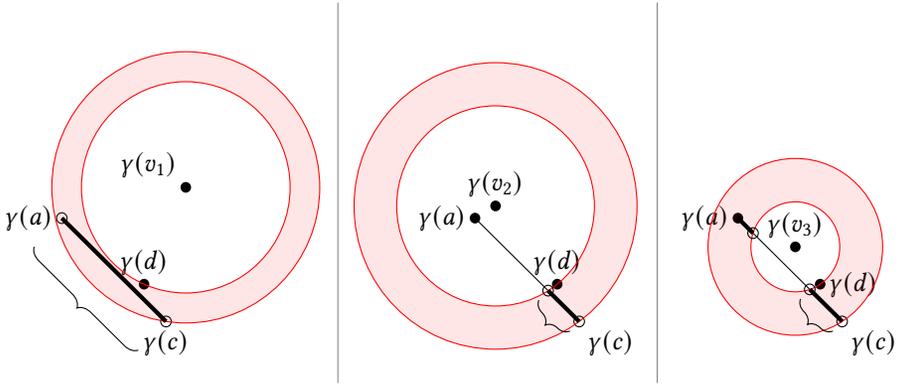
\begin{figure}[ht]
																			\centering
																			\begin{tikzpicture}[scale=0.3]
																				\begin{scope}
																					\coordinate (a) at (1.1500664776888,3.3341147983);
																					\draw (a) circle (4pt);
																					\node[left] at (a){$\gamma(a)$};
																					
																					\coordinate (d) at (3.339169250,1.59792984022);
																					\fill (d) circle (4pt);
																					\node [above] at (d){$\gamma(d)$};
																					
																					\coordinate (c) at (3.924188095546,0.6166079073822);
																					\draw (c) circle (4pt);
																					\node [below] at (c){$\gamma(c)$};

																					\draw (a) -- (c);
																					
																					\coordinate (v1) at (4.45070567,4.1437053435);
																					\fill (v1) circle (4pt);
																					\node [above left] at (v1){$\gamma(v_1)$};
																					
																					\draw[red,fill=red,fill opacity=0.1,even odd rule] (v1) circle (2.7778561895712492) circle (3.5661795779168637
																					);
																					
																					\draw[ultra thick] (a) -- (c);
																					
																					\draw[decorate,decoration={brace,amplitude=5pt,raise=15pt}] (c) -- (a);

																				\end{scope}
																				\draw[help lines] (8.5,-1) -- (8.5,9);
																				\draw[help lines] (17,-1) -- (17,9);
																				
																				\begin{scope}[xshift=11cm]
																					\coordinate (a) at (1.1500664776888,3.3341147983);
																					\fill (a) circle (4pt);
																					\node[left] at (a){$\gamma(a)$};
																					
																					\coordinate (d) at (3.339169250,1.59792984022);
																					\fill (d) circle (4pt);
																					\node [above] at (d){$\gamma(d)$};
																					
																					\coordinate (c) at (3.924188095546,0.6166079073822);
																					\draw (c) circle (4pt);
																					\node [below right] at (c){$\gamma(c)$};
																					
																					\draw (a) -- (c);
																					
																					\coordinate (v2) at (1.697342176,3.654934083);
																					\fill (v2) circle (4pt);
																					\node [above] at (v2){$\gamma(v_2)$};
																					
																					\draw[red,fill=red,fill opacity=0.1,even odd rule] (v2) circle (2.6318913978354046) circle (3.7669952928030055);
																					
																					\node(cc2) at (v2)[circle through={(d)}]{};
																					\coordinate (int2) at (intersection 1 of a--c and cc2);
																					
																					\draw (int2) circle (4pt);
																					\draw[ultra thick] (int2) -- (c);
																					
																					\draw[decorate,decoration={brace,amplitude=5pt,raise=5pt}] (c) -- (int2);

																				\end{scope}
																				\begin{scope}[xshift=18cm]
																					\coordinate (a) at (1.1500664776888,3.3341147983);
																					\fill (a) circle (4pt);
																					\node[left] at (a){$\gamma(a)$};
																					
																					\coordinate (d) at (3.339169250,1.59792984022);
																					\fill (d) circle (4pt);
																					\node [right] at (d){$\gamma(d)$};
																					
																					\coordinate (c) at (3.924188095546,0.6166079073822);
																					\draw (c) circle (4pt);
																					\node [below right] at (c){$\gamma(c)$};
																					
																					\draw (a) -- (c);
																					
																					\coordinate (v3) at (2.67866410,2.57925177);
																					\fill (v3) circle (4pt);
																					\node [above] at (v3){$\gamma(v_3)$};
																					
																					\draw[red,fill=red,fill opacity=0.1,even odd rule] (v3) circle (1.1829031174381104
																					) circle (2.3245001520757103);
																					
																					\node(cc3) at (v3)[circle through={(d)}]{};
																					\coordinate (int3) at (intersection 1 of a--c and cc3);
																					\coordinate (int3extra) at (1.55217858,2.94020905); 
																					
																					\draw (int3) circle (4pt);
																					\draw (int3extra) circle (4pt);
																					
																					\draw[ultra thick] (int3) -- (c);
																					
																					\draw[decorate,decoration={brace,amplitude=5pt,raise=5pt}] (c) -- (int3);
																					
																					\draw[ultra thick] (int3extra) -- (a);
																					
																				\end{scope}

																			\end{tikzpicture}
																			\caption{The three subcases in Case 2 in the proof of Reduction Rule 2. The black line is the line $L=\linesegment{\gamma(a)}{\gamma(c)}$ and the thick parts correspond to the set $L'$. The braces represent the line segments $L_d(v_i)$. From left to right the number of intersection points of $L$ with the inner circle $\boundary{\openballcp{\gamma(v_i)}{\gamma(d)}}$ is $0,1,$ and $2$, respectively. The red concentric circles represent the annulus $\annuluscpp{\gamma(v_i)}{\gamma(d)}{\gamma(c)}$.}\label{fig:rr5_three_subcases}
																		\end{figure}
																	}
																	
																	\dtoappendix
																	{
																		\subsection{Limitation of generalization of reduction rules}
																	}
																	We show that our reduction rules are optimally constrained. In particular, we prove that \Cref{rr:rr1plus} cannot be applied even to tail blocks of size $2$ when no copy of the tail block exists (\Cref{prop:cannot_remove_two_candidates_without_copy}). We then further explore the generalization of \Cref{rr:rr1plus} and show that it is not possible to remove arbitrarily large tail blocks even with copies (\Cref{prop:tail_block_seven}). Finally, we show that condition 3 in \Cref{rr:rr2} cannot be omitted (\Cref{prop:cannot_gen_rr2}).
																	
																	\lv
																	{
																		\begin{proposition}\label{prop:cannot_remove_two_candidates_without_copy}
																		}
																		\sv
																		{
																			\begin{proposition}[$\star$]\label{prop:cannot_remove_two_candidates_without_copy}
																			}
																			There exists an election $(C,V)$ with \mbox{$|C|=14$},  \mbox{$|V|=4$} with tail block $S\subseteq C$ of size $2$ and the subelection induced by $C\setminus S$ is $2$-Euclidean while $(C,V)$ is not $2$-Euclidean.
																		\end{proposition}
																		\toappendix
																		{
																			\sv
																			{
																				\begin{proof}[Proof of \Cref{prop:cannot_remove_two_candidates_without_copy}]
																				}
																				\lv
																				{
																					\begin{proof}
																					}
																					The election is as follows. Let $C=\{c_1,\ldots,c_{14}\}$ and $V=\{v_1,v_2,v_3,v_4\}$, where:
																					\[
																					\begin{array}{cccccccccccccc}
																						v_1: & c_2 c_1 & c_3 c_4 & c_5 c_6 & c_8 c_7 & c_{10} c_9  & c_{11} c_{12} & c_{13} c_{14}\\
																						v_2: & c_1 c_2 & c_4 c_3 & c_5 c_6 & c_8 c_7 & c_{9} c_{10} & c_{12} c_{11} & c_{13} c_{14}\\
																						v_3: & c_1 c_2 & c_3 c_4 & c_6 c_5 & c_7 c_8 & c_{10} c_9   & c_{12} c_{11} & c_{13} c_{14}\\
																						v_4: & c_1 c_2 & c_3 c_4 & c_5 c_6 & c_7 c_8 & c_{9} c_{10} & c_{11} c_{12} & c_{14} c_{13}\\		
																					\end{array}
																					\]
																					We omit the symbol $\succ$ due to space constraints and emphasize the blocks of two candidates $c_{2k-1},c_{2k}$. To show that $(C,V)$ is not $2$-Euclidean we use the controversity graph. Note that the voter~$v_k$ is controversial for~$c_{2k}$ over~$c_{2k-1}$ for $k\in[3]$ and $v_4$ is controversial for~$c_{14}$ over~$c_{13}$ and due to candidates $c_7,\ldots,c_{12}$ for any pair of distinct voters $u,v$ the set $\{u,v\}$ is controversial. Hence the controversity graph  $\controversity{C}{V}$ is a complete graph on $4$ vertices and thus by \Cref{thm:controversity_graph_property} the election $(C,V)$ cannot be $2$-Euclidean.
																					
																					To prove that the subelection induced by $C\setminus \{c_{13},c_{14}\}$ is $2$-Euclidean we explicitly construct a $2$-Euclidean embedding $\gamma$ for it.
																					
																					Start by embedding $\gamma(v_1),\gamma(v_2),\gamma(v_3)$ to the vertices of an equilateral triangle $T$ and $\gamma(v_4)$ to the circumcenter of $T$. Next, let $L_{1,2},L_{3,4},L_{5,6},L_{7,8},L_{9,10},L_{11,12}$ be the six lines as depicted in \Cref{fig:pf_thm_no_copied_order}. The line $L_{2k-1,2k}$ separates the voters with respect to their preference about candidates $c_{2k-1} c_{2k}$. For example $v_1$ is the only candidate that prefers $c_2$ over $c_1$, thus $L_{1,2}$ separates $\gamma(v_1)$ from $\gamma(v_2),\gamma(v_3),\gamma(v_4)$.
																					
																					We now describe how to embed the candidates. The aim is to embed the candidates in such a way that the line $L_{2k-1,2k}$ coincides with the bisector $\bisector{\gamma}{c_{2k-1}}{c_{2k}}$. In this way the embedding of the candidates will agree with the preferences about $c_{2k-1},c_{2k}$ of the individual voters. What is left is to ensure that the placement agrees with the preferences $c_k \succ_{v} c_{k+i}$ for $i>1$ which holds for all voters $v$.
																					
																					We inductively place the pairs $c_{2k-1},c_{2k}$ for all $k\in[6]$.
																					Start by placing $c_1,c_2$ arbitrarily such that $\bisector{\gamma}{c_1}{c_2}=L_{1,2}$. Suppose now that first $2k$ candidates were placed for $k\geq 1$. To place $c_{2k+1},c_{2k+2}$ consider the closed balls of the form $\closure{\openballcp{\gamma(v)}{\gamma(c_{k'})}}$ for $v\in V$ and $k'\in[2k]$. In order for the embedding $\gamma$ to agree with $c_{2k+1}\succ_{v} c_{k'}$ and $c_{2k+2}\succ_v c_{k'}$ for all $k'\in[2k],v\in V$ it is enough to place the two candidates $c_{2k+1},c_{2k+2}$ outside the union
																					$
																					\bigcup_{k'\in[2k],v\in V}\closure{\openballcp{\gamma(v)}{\gamma(c_{k'})}}
																					$
																					and such that $L_{2k+1,2k+2}=\bisector{\gamma}{c_{2k+1}}{c_{2k+2}}$. This is indeed possible as the region given by the finite union of closed balls is bounded, however, all the lines $L_{i,i+1}$ extend to infinity.
																				\end{proof}

																				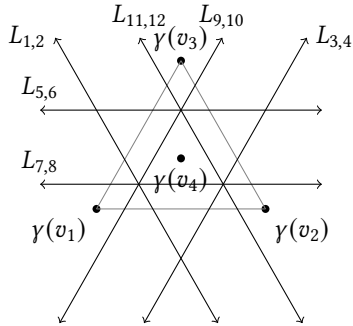
\begin{figure}[ht]
																					\center
																					\begin{tikzpicture}[scale=0.75]
																						\coordinate (gv1) at (0,0);
																						\coordinate (gv2) at (3,0);
																						\coordinate (gv3) at (1.5,2.598076);
																						
																						\coordinate (gv4) at (1.5,0.88602540);
																						
																						\fill (gv1) circle (2pt);
																						\node[below left] at (gv1){$\gamma(v_1)$};
																						
																						\fill (gv2) circle (2pt);
																						\node [below right] at (gv2) {$\gamma(v_2)$}; 
																						
																						\fill (gv3) circle (2pt);
																						\node[above] at (gv3){$\gamma(v_3)$};

																						\fill (gv4) circle (2pt);
																						\node[below] at (gv4){$\gamma(v_4)$};
																						
																						\draw[help lines] (gv1) -- (gv2);	
																						\draw[help lines] (gv1) -- (gv3);	
																						\draw[help lines] (gv3) -- (gv2);

																						\coordinate (x1) at (-0.7320508075689,3);
																						\coordinate (x2) at (2.15470053, -2);
																						\draw[<->] (x1) -- (x2);
																						\node[left] at (x1){$L_{1,2}$};
																						
																						\coordinate (x3) at (3.73205080,3);
																						\coordinate (x4) at (0.8452994616,-2);
																						\draw[<->] (x3) -- (x4);
																						\node[right] at (x3){$L_{3,4}$};
																						
																						\coordinate (x5) at (-1,1.7320508075689);
																						\coordinate (x6) at (4,1.7320508075689);
																						\draw[<->] (x5) -- (x6);
																						\node[above] at (x5){$L_{5,6}$};

																						\coordinate (x7) at (-1,0.43301270);
																						\coordinate (x8) at (4,0.43301270);
																						\draw[<->] (x7) -- (x8);
																						\node[above] at (x7){$L_{7,8}$};

																						\coordinate (x9) at (2.23205080,3);
																						\coordinate(x10) at (-0.6547005,-2);
																						\draw[<->] (x9) -- (x10);
																						\node[above] at (x9){$L_{9,10}$};

																						\coordinate (x11) at (0.7679491924311,3);
																						\coordinate (x12) at (3.65470053,-2);
																						\draw[<->] (x11) -- (x12);
																						\node[above] at (x11){$L_{11,12}$};
																						
																					\end{tikzpicture}
																					\caption{Realization of the six lines $L_{1,2},\ldots,L_{11,12}$ in the proof of \Cref{prop:cannot_remove_two_candidates_without_copy}. The gray triangle is the triangle $T$. The concrete realization of the lines is not relevant, the important property is that $L_{1,2}$ separates $\gamma(v_1)$ from $\gamma(v_2),\gamma(v_3),\gamma(v_4)$, $L_{3,4}$ separates $\gamma(v_2)$ from $\gamma(v_1),\gamma(v_3),\gamma(v_4)$ and so on.}\label{fig:pf_thm_no_copied_order}
																				\end{figure}
																			}
																			\sv
																			{
																				\begin{proposition}[$\star$]\label{prop:tail_block_seven}
																				}
																				\lv
																				{
																					\begin{proposition}\label{prop:tail_block_seven}
																					}
																					There exists an election $(C,V)$ with \mbox{$|C|=14$}, \mbox{$|V|=7$} with tail block $S\subseteq C$ of size $7$ and the set $S'=C\setminus S$ copies~$S$ but the subelection induced by $C\setminus S$ is $2$-Euclidean and $(C,V)$ is not.
																				\end{proposition}
																				\toappendix{
																					\sv
																					{
																						\begin{proof}[Proof of \Cref{prop:tail_block_seven}]
																						}
																						\lv
																						{
																							\begin{proof}
																							}
																							Consider the instance $(C,V)$ given by $V=\{v_1,v_2,v_3\}$ and $C=\{c_0,c_1,\ldots,c_6,c_0',c_1',\ldots,c_6'\}$. The three voters in $V$ are:
																							\[
																							\begin{array}{ccc}
																								v_1: c_6 c_4 c_1 c_0 c_2 c_3 c_5 & c_6' c_4' c_1' c_0' c_2' c_3' c_5' \\
																								v_2: c_5 c_4 c_2 c_0 c_1 c_3 c_6 & c_5' c_4' c_2' c_0' c_1' c_3' c_6' \\
																								v_3: c_6 c_5 c_3 c_0 c_1 c_2 c_4 & c_6' c_5' c_3' c_0' c_1' c_2' c_4' \\
																							\end{array}
																							\]
																							To see why $(C,V)$ is not $2$-Euclidean, note that this instance contains the forbidden $3$-$8$ pattern by choosing $c_\emptyset$ as $c_0'$. The tail block is $S=\{c_0',\ldots,c_6'\}$. Note that subelection induced by $C\setminus S$ has $7$ candidates and $3$ voters and by the result of Bulteau and Chen (\cite[Theorem 3]{BulteauC23}) any such election is $2$-Euclidean.
																						\end{proof}
																					}
																					We know that we can always remove tail blocks with copies of size at most $3$ by \Cref{rr:rr1plus} and \Cref{prop:tail_block_seven} gives an upper bound on tail block size that certainly cannot be removed safely even with copies. We conjecture that \Cref{rr:rr1plus} cannot be extended to remove more than $3$ candidates.
																					
																					\sv
																					{
																						\begin{proposition}[$\star$]\label{prop:cannot_gen_rr2}
																						}
																						\lv
																						{
																							\begin{proposition}\label{prop:cannot_gen_rr2}
																							}
																							There exists an election $(C,V)$ with \mbox{$|C|=8$}, \mbox{$|V|=3$} with two candidates $b,c\in C$ such that:
																							\begin{enumerate}
																								\item for all votes $v$ we have $b\succ_v c$,
																								\item for all votes $v$ there is no $d\in C\setminus \{b,c\}$ such that $b\succ_v d$ and $d\succ_v c$,
																								\item $(C,V)$ is not $2$-Euclidean, and 
																								\item the subelection induced by  \mbox{$C\setminus \{b\}$} is $2$-Euclidean.
																							\end{enumerate}
																							In other words, we cannot omit condition $3$ in \Cref{rr:rr2}.
																						\end{proposition}
																						\toappendix
																						{
																							\sv
																							{
																								\begin{proof}[Proof of \Cref{prop:cannot_gen_rr2}]
																								}
																								\lv
																								{
																									\begin{proof}
																									}
																									The instance $(C,V)$ is the forbidden $3$-$8$ pattern:
																									\begin{align*}
																										v_1 &\colon   c_{1}\succ c_{12} \succ c_{13} \succ c_{123} \succ  c_\emptyset\succ c_2 \succ c_{23} \succ c_3\\
																										v_2 &\colon  c_{2} \succ c_{23}\succ c_{12}\succ c_{123}\succ c_\emptyset\succ c_{13}\succ c_{1}\succ c_{3}\\
																										v_3 &\colon  c_{13}\succ c_3 \succ c_{23}\succ c_{123}\succ c_\emptyset\succ c_{1}\succ c_{2}\succ c_{12}
																									\end{align*}
																									where $b=c_{123}$ and $c=c_\emptyset$. The subelection induced by $C\setminus \{b\}$ has $3$ voters and $7$ candidates and any such election is $2$-Euclidean by the result of Bulteau and Chen (\cite[Theorem 3]{BulteauC23}).
																								\end{proof}
																							}
																							\toappendix
																							{
																								\sv
																								{
																									\section{Missing Material from section Refuting the existence of an embedding with an ILP}
																								}
																							}
																							\section{Refuting the existence of an embedding with an ILP}\label{sec:implied_regions}
																							The property of being $2$-Euclidean imposes restriction on the election itself but also on the embedding. A trivial observation is that whenever $(C,V)$ admits a $2$-Euclidean embedding $\gamma$, then $\gamma$ must have at least $|V|$ nonempty regions, one for each $v\in V$. 
																							
																							It is not hard to observe that for $|C|=3$ the embedding that embeds the $3$ candidates to the vertices of a non-degenerate triangle induces $6=3!$ distinct regions. Thus, any election with at most $3$ candidates is $2$-Euclidean. In general, an election with $m$ candidates can have up to $m!$ votes. However, these must correspond to some nonempty regions induced by the bisectors. Bennet and Hays~\cite{Bennett1960} gave a recursive formula to compute the maximum number of nonempty regions induced by an embedding of $m$ candidates in $d$ dimensions. For $d=2$ the upper bound is as follows.

																							\begin{corollary}\label{thm:ubm}
																								Let $\gamma\colon C\to\mathbb{R}^2$ be an embedding of candidates. Let $\mathcal{R}= \{v \mid \region{\gamma}{v}\neq \emptyset\}$ be the set of nonempty regions of $\gamma$. Then $|\mathcal{R}|\leq \ub(|C|)$, where
																								\begin{equation}
																									\ub(m) = \frac{m(3m-10)(m+1)(m-1)}{24}+m(m-1)+1.\tag{UB}\label{eqn:ubm}
																								\end{equation}
																							\end{corollary}
																							
																							A trivial consequence of \Cref{thm:ubm} is the following.  Whenever $|V|>\ub(|C|)$, then $(C,V)$ is clearly not $2$-Euclidean as otherwise a $2$-Euclidean embedding would induce too many regions. \Cref{thm:ubm} is one of many combinatorial properties that the election must satisfy in order to even have a chance to be $2$-Euclidean. Towards designing an efficient refutation procedure, we derive additional properties that only depend on the combinatorial structure of the election itself or hold for \emph{any} $2$-Euclidean embedding of the given election.
																							
																							For convenience, we rephrase the geometrical terminology of bisectors and regions into the language of plane graphs and introduce the \emph{embedding graph} for a $2$-Euclidean embedding of an election in \Cref{subsec:embedding_graph}. We prove important properties of the embedding graph in \Cref{subsec:properties_of_embeddings} and then utilize them to design an integer linear program refuting the existence of $2$-Euclidean embedding in \Cref{subsec:the_algorithm}. To avoid degenerate or trivial cases, we will assume, without loss of generality, that $|C|\geq 4$. We already know that any election with at most three candidates is $2$-Euclidean.

																							\dtoappendix
																							{
																								\subsection{Embedding Graph}
																							}\label{subsec:embedding_graph}
																							Let $\gamma \colon C \to \mathbb{R}^2$ be a nice embedding of candidates. The arrangement of lines $\arrangement{\gamma}$ induces a plane graph $\primalg{\gamma}=(V(\primalg{\gamma}),E(\primalg{\gamma}))$ (the \emph{primal graph}) as follows (refer to \Cref{fig:embedding_graph} for an illustration).
																							Let $I^\gamma=\{p\in \mathbb{R}^2 \mid \exists \beta_1,\beta_2, \in \arrangement{\gamma}: \beta_1\neq \beta_2, p \in \beta_1\cap \beta_2\}$ be the set of all pairwise intersection points of distinct bisectors. Let $B^\gamma$ be an open ball containing $I^\gamma$. We let $V(\primalg{\gamma}) = I^\gamma \cup \left(\bigcup\arrangement{\gamma} \cap \partial B^\gamma\right)$. The arcs $E(\primalg{\gamma})$ are either the straight line segments (parts of bisectors) or circular arcs (parts of $\partial B^\gamma$) connecting two intersection points. Observe that the bounded faces of $\primalg{\gamma}$ are exactly the regions $\region{\gamma}{v}$ (the outer regions are clipped by $B^\gamma$). Moreover the circle $\boundary{B^\gamma}$ is partitioned by the bisectors into several circular arcs and there is a one to one correspondence between these circular arcs and the outer regions of the embedding.
																							
																							The \emph{embedding graph for $\gamma$}, denoted by $\dualg{\gamma}$, is the weak dual of the plane graph $\primalg{\gamma}$, i.e., $\dualg{\gamma}$ has a vertex for each bounded face of $\primalg{\gamma}$ and there is an edge between two faces $F_1,F_2$ for each common boundary edge of $F_1$ and $F_2$. Each bounded face of $\primalg{\gamma}$ in turn corresponds to a nonempty region $\region{\gamma}{v}$. By a slight abuse of notation we identify vertices $v\in V(\dualg{\gamma})$ with the underlying permutation $v\in \sym{C}$ and with the nonempty region $\region{\gamma}{v}$. This correspondence allows us to rephrase \Cref{thm:ubm} in terms of the embedding graph:
																							\begin{corollary}\label{cor:ubm_embgr}
																								Let $(C,V)$ be a \mbox{$2$-Euclidean} election, and $\gamma$ its \mbox{$2$-Euclidean} embedding. Then the number of vertices of $\dualg{\gamma}$ is at most $\ub(|C|)$.
																							\end{corollary}
																							
																							\begin{figure}
																								\centering
																								\begin{tikzpicture}[scale=0.25]
																									\node[circle,scale=0.5,fill=red!60,label=right:$a$] (a) at (3.18,-0.78) {};
																									\node[circle,scale=0.5,fill=blue!60,label=$b$] (b) at (-4.58,-2.66) {};
																									\node[circle,scale=0.5,fill=green!60,label=left:$c$] (c) at (4.27,-3.23) {};
																									\node[circle,scale=0.5,fill=yellow!60,label=$d$] (d) at (5.4,5.84) {};
																									
																									\path (a) -- (b) coordinate[midway] (Mab);
																									\draw[postaction={draw,red,dash pattern= on 3pt off 5pt,dash phase=4pt,thick}]
																									[blue,dash pattern= on 3pt off 5pt,thick,<->] ($(Mab)!11cm!270:(a)$) -- ($(Mab)!11cm!90:(a)$);
																									
																									\path (a) -- (c) coordinate[midway] (Mac);
																									\draw[name path=AC,postaction={draw,red,dash pattern= on 3pt off 5pt,dash phase=4pt,thick}]
																									[green,dash pattern= on 3pt off 5pt,thick,<->] ($(Mac)!12cm!270:(a)$) -- ($(Mac)!12cm!90:(a)$);
																									
																									\path (a) -- (d) coordinate[midway] (Mad);
																									\draw[postaction={draw,red,dash pattern= on 3pt off 5pt,dash phase=4pt,thick}]
																									[yellow,dash pattern= on 3pt off 5pt,thick,,<->] ($(Mad)!11cm!270:(a)$) -- ($(Mad)!12cm!90:(a)$);
																									\path (b) -- (c) coordinate[midway] (Mbc);
																									\draw[postaction={draw,blue,dash pattern= on 3pt off 5pt,dash phase=4pt,thick}]
																									[green,dash pattern= on 3pt off 5pt,thick,,<->] ($(Mbc)!14cm!270:(b)$) -- ($(Mbc)!9cm!90:(b)$);
																									\path (b) -- (d) coordinate[midway] (Mbd);
																									\draw[postaction={draw,blue,dash pattern= on 3pt off 5pt,dash phase=4pt,thick}]
																									[yellow,dash pattern= on 3pt off 5pt,thick,,<->] ($(Mbd)!8cm!270:(b)$) -- ($(Mbd)!15cm!90:(b)$);
																									\path (c) -- (d) coordinate[midway] (Mcd);
																									\draw[name path=CD, postaction={draw,green,dash pattern= on 3pt off 5pt,dash phase=4pt,thick}]
																									[yellow,dash pattern= on 3pt off 5pt,thick,,<->] ($(Mcd)!13cm!270:(c)$) -- ($(Mcd)!11cm!90:(c)$);

																									\draw[name path=circ] (a) circle (11);
																									\draw[name intersections={of=circ and CD}];
																									\coordinate (x) at (intersection-2);
																									\draw[ultra thick] (x) arc(161:211:11); 
																									\draw[ultra thick] (x) arc(161:210.5:11.1);	
																									\draw[ultra thick] (x) arc(161:210.5:11.2);		
																									
																									\node[] (Bg) at (10,10) {$\boundary{B^\gamma}$};
																									
																									\node[fill,scale=0.5,circle] (abcd) at (-0.45,1.30) {};
																									\node[fill,scale=0.5,circle] (acbd) at (0.81,-0.62) {};
																									\node[fill,scale=0.5,circle] (cabd) at (2.47,-5.67) {};
																									\node[fill,scale=0.5,circle] (cbad) at (0,-7) {};
																									\node[fill,scale=0.5,circle] (bcad) at (-2.58,-6) {};
																									\node[fill,scale=0.5,circle] (bacd) at (-2.18,-1.23) {};
																									\node[fill,scale=0.5,circle] (badc) at (-2.92,3.69) {};
																									\node[fill,scale=0.5,circle] (bdac) at (-4,6.02) {};
																									\node[fill,scale=0.5,circle] (dbac) at (-3.5,7.23) {};
																									\node[fill,scale=0.5,circle] (dabc) at (-0.91,6.54) {};
																									\node[fill,scale=0.5,circle] (dacb) at (6,4) {};
																									\node[fill,scale=0.5,circle] (dcab) at (13,1) {};
																									\node[fill,scale=0.5,circle] (cdab) at (14,-0.2) {};
																									\node[fill,scale=0.5,circle] (cadb) at (10.2,-3.22) {};
																									\node[fill,scale=0.5,circle] (acdb) at (4.04,0.44) {};
																									\node[fill,scale=0.5,circle] (adcb) at (1.44,2.63) {};
																									\node[fill,scale=0.5,circle] (adbc) at (-0.38,3.48) {};
																									\node[fill,scale=0.5,circle] (abdc) at (-1.15,2.61) {};
																									
																									\draw[thick] (abcd) -- (acbd);
																									\draw[thick] (acbd) -- (cabd);
																									\draw[thick] (cabd) -- (cbad);
																									\draw[thick] (cbad) -- (bcad);
																									\draw[thick] (bcad) -- (bacd);
																									\draw[thick] (bacd) -- (badc);
																									\draw[thick] (badc) -- (bdac);
																									\draw[thick] (bdac) -- (dbac);
																									\draw[thick] (dbac) -- (dabc);
																									\draw[thick] (dabc) -- (dacb);
																									\draw[thick] (dacb) -- (dcab);
																									\draw[thick] (dcab) -- (cdab);
																									\draw[thick] (cdab) -- (cadb);
																									\draw[thick] (cadb) -- (acdb);
																									\draw[thick] (acdb) -- (adcb);
																									\draw[thick] (adcb) -- (adbc);
																									\draw[thick] (adbc) -- (abdc);
																									\draw[thick] (abdc) -- (abcd);
																									
																									\draw[thick] (abcd) -- (bacd);
																									\draw[thick] (abdc) -- (badc);
																									\draw[thick] (adbc) -- (dabc);
																									\draw[thick] (adcb) -- (dacb);
																									\draw[thick] (acbd) -- (acdb);
																									\draw[thick] (cadb) -- (cabd);
																								\end{tikzpicture}
																								\caption{The embedding graph $\dualg{\gamma}$ (in black). 
																									The highlighted circular arc corresponds to the region $\region{\gamma}{bacd}$. }\label{fig:embedding_graph}
																								\vspace{-1em}
																							\end{figure}
																							
																							\dtoappendix
																							{
																								\subsubsection{Properties of Embeddings and the Embedding Graph}
																							}\label{subsec:properties_of_embeddings}
																							We begin with a simple observation that the edges of the embedding graph correspond to consecutive swaps in the underlying permutations. 
																							
																							\sv
																							{
																								\begin{observation}[$\star$]\label{obs:edge_iff_neighboring_swap}
																								}
																								\lv
																								{
																									\begin{observation}\label{obs:edge_iff_neighboring_swap}
																									}
																									In the embedding graph $\dualg{\gamma}$, there is an edge between two vertices $u$ and $v$ if and only if $\dswap{u}{v}=1$.
																								\end{observation}
																								\toappendix
																								{
																									\sv
																									{
																										\begin{proof}[Proof of \Cref{obs:edge_iff_neighboring_swap}]
																										}
																										\lv
																										{
																											\begin{proof}
																											}
																											Let $e=\{u,v\}\in E(\dualg{\gamma})$. By definition there is an edge of $\primalg{\gamma}$ that $e$ corresponds to. This edge is a part of some bisector $\bisector{\gamma}{a}{b}$ for some candidates $a,b\in C$. Note that for any two distinct candidates $c,d\in C$ such that $\{c,d\}\neq \{a,b\}$ we have  $\region{\gamma}{u}\subseteq \halfplane{\gamma}{c}{d} \Leftrightarrow \region{\gamma}{v}\subseteq \halfplane{\gamma}{c}{d}$.
																											In other words the regions $\region{\gamma}{u}$ and $\region{\gamma}{v}$ are on the same side of any other bisector $\bisector{\gamma}{c}{d}$. This implies that $c\succ_v d\Leftrightarrow c\succ_u d$. Note that this implies that $a,b$ must be consecutive in $v$. To see this, suppose without loss of generality that $a\succ_v b$ and towards a contradiction suppose that there is some $x\in C\setminus \{a,b\}$ such that $a\succ_v x \succ_v b$. This implies $b\succ_u x \succ_u a$. However, this contradicts the equivalence $c\succ_v d\Leftrightarrow c\succ_u d$ by choosing $c=a$ and $d=x$.
																											
																											For the other direction, suppose that $u,v$ are two votes that differ by a consecutive swap of candidates $a,b\in C$. That is, $u=c_{i_1}c_{i_2}\ldots ab \ldots c_{i_\ell}$ and $v=c_{i_1}c_{i_2}\ldots ba\ldots c_{i_\ell}$. Then $\region{\gamma}{u}\subseteq \halfplane{\gamma}{a}{b}$ but $\region{\gamma}{v}\subseteq \halfplane{\gamma}{b}{a}$. For any other two distinct candidates $c,d\in C$ such that $\{c,d\}\neq \{a,b\}$ we have
																											$c\succ_v d\Leftrightarrow c\succ_u d$. In other words	 $\region{\gamma}{u}\subseteq \halfplane{\gamma}{c}{d}\Leftrightarrow\region{\gamma}{v}\subseteq \halfplane{\gamma}{c}{d}$. This implies that in the arrangement $\arrangement{\gamma}\setminus \bisector{\gamma}{a}{b}$ the permutations $u,v$ would end up in the same region. However, by adding $\bisector{\gamma}{a}{b}$ they are different, thus $\region{\gamma}{u}$ and $\region{\gamma}{v}$ are necessarily adjacent regions. Therefore, $\{u,v\}\in E(\dualg{\gamma})$.
																										\end{proof}
																									}
																									\lv
																									{
																										\begin{observation}\label{obs:d_has_no_multiedges}
																										}
																										\sv
																										{
																											\begin{observation}[$\star$]\label{obs:d_has_no_multiedges}
																											}
																											The embedding graph $\dualg{\gamma}$ has no loops and multiedges.
																										\end{observation}
																										\toappendix
																										{
																											\lv
																											{
																												\begin{proof}
																												}
																												\sv
																												{
																													\begin{proof}[Proof of \Cref{obs:d_has_no_multiedges}]
																													}
																													Note that a loop would imply that $v=v\circ \tau_{a,b}$ for some transposition $\tau_{a,b}$ and vote $v$ which is impossible. A multiedge would imply that $u=v\circ \tau_{a,b}$ and $u=v\circ \tau_{c,d}$ for some $\{c,d\}\neq \{a,d\}$ and votes $u,v$. But then clearly $\tau_{a,b}=\tau_{c,d}$ which is not true.
																												\end{proof}
																											}
																											Let $G(\sym{C})$ be a graph with vertex set $\sym{C}$ and two vertices share an edge if and only if they differ by a consecutive swap. It is not hard to see that $\dualg{\gamma}$ is a subgraph of $G(\sym{C})$. Moreover, it is a distance-preserving subgraph of $G(\sym{C})$, i.e., for any two vertices $u,v\in V(H)$ we have $d_H(u,v)=d_G(u,v)$.
																											
																											\lv
																											{
																												\begin{lemma}\label{lem:dualg_preserves_distances}
																												}
																												\sv
																												{
																													\begin{lemma}[$\star$]\label{lem:dualg_preserves_distances}
																													}
																													The graph $\dualg{\gamma}$ is distance-preserving subgraph of $G(\sym{C})$ of minimum degree at least $2$.
																												\end{lemma}
																												\toappendix
																												{
																													\lv
																													{
																														\begin{proof}
																														}
																														\sv
																														{
																															\begin{proof}[Proof of \Cref{lem:dualg_preserves_distances}]
																															}
																															The fact that $\dualg{\gamma}$ is a subgraph of $G(\sym{C})$ is a direct corollary of \Cref{obs:edge_iff_neighboring_swap,obs:d_has_no_multiedges}. Recall that the graph distance in $G(\sym{C})$ is equivalent to the swap distance of two permutations $u,v$. This in turn is equivalent to the Kendall Tau distance which is the number of discordant pairs of candidates, i.e, $D=|\{a,b\}\subseteq C\mid \pos{u}{a}<\pos{u}{b}\wedge \pos{v}{b}<\pos{v}{a} \text{ or }\pos{u}{b}<\pos{u}{a}\wedge \pos{v}{a}<\pos{v}{b}\}|$.
																															
																															Let $u,v\in V(\dualg{\gamma})$ be arbitrary. We show that $d_{\dualg{\gamma}}(u,v)=d_{G(\sym{C})}(u,v)$. Since $u,v\in V(\dualg{\gamma})$ we have $\region{\gamma}{u}\neq \emptyset,\region{\gamma}{v}\neq \emptyset$. Consider arbitrary points $p_u\in \region{\gamma}{u},p_v\in \region{\gamma}{v}$ and the straight line segment $L=\linesegment{p_u}{p_v}$. Let $\mathcal{B}=\{\bisector{\gamma}{a}{b}\mid \{a,b\}\in D\}$ be the set of bisectors corresponding to the discordant pairs of candidates with respect to $u$ and $v$. We have $|\mathcal{B}|=\dswap{u}{v}=d_{G(\sym{C})}(u,v)$. We argue that $L$ crosses exactly the bisectors from $\mathcal{B}$. To see this, consider travelling along $L$ from $p_u$ to $p_v$ and maintain the current set $D'$ of discordant pairs. We initialize $D'$ by all $\{a,b\}$ such that $\bisector{\gamma}{a}{b}\in\mathcal{B}$. Whenever we cross a bisector of the form $\bisector{\gamma}{c}{d}$ we either add the pair $\{c,d\}$ to $D'$ if $\{c,d\}\notin D'$ or we remove it from $D'$ if $\{c,d\}\in D'$. Note that any pair will be added or removed exactly once because we cannot cross a bisector twice. When we arrive to $p_v$, the set $D'$ must be empty. Hence each crossing of a bisector must remove exactly one pair. Thus $L$ crosses exactly the bisectors from $\mathcal{B}$. Each bisector that $L$ crossed corresponds to some edge of $\dualg{\gamma}$ and these edges induce a path from $u$ to $v$ of length $|\mathcal{B}|=\dswap{u}{v}=d_{G(\sym{C})}(u,v)$ in $\dualg{\gamma}$. There cannot be any shorter $u$-$v$ path in $\dualg{\gamma}$ because $\dualg{\gamma}$ is a subgraph of $G(\sym{C})$. Hence, $d_{G(\sym{C})}(u,v) = d_{\dualg{\gamma}}(u,v)$, as claimed.
																															
																															For the minimum degree, recall that we assume that $|C|\geq 4$. Hence $\dualg{\gamma}$ has at least $3$ vertices. It is not hard to see that $\dualg{\gamma}$ is connected. Hence it does not contain vertices of degree $0$. Suppose for the sake of contradiction that $v\in V(\dualg{\gamma})$ is a vertex of degree $1$. By definition this means that the boundary of the region $\region{\gamma}{v}$ consists of exactly one straight line segment. That is clearly impossible because $\region{\gamma}{v}$ is bounded.
																														\end{proof}
																													}

																													\sv
																													{
																														\begin{lemma}[$\star$]\label{lem:implied_neighbors}
																														}
																														\lv
																														{
																															\begin{lemma}\label{lem:implied_neighbors}
																															}
																															Let $(C,V)$ be a $2$-Euclidean election and $\gamma$ any $2$-Euclidean embedding of $(C,V)$ and let $G=G(\sym{C})$. Let $u,v\in V$ be two votes and let $v_1,\ldots,v_k\in N_G(u)$ be all the neighbors of $u$ such that $d_{G}(v_i,v)=d_G(u,v)-1$. Then there is index $i$ such that $v_i\in V(\dualg{\gamma})$.
																														\end{lemma}
																														\toappendix
																														{
																															\sv
																															{
																																\begin{proof}[Proof of \Cref{lem:implied_neighbors}]
																																}
																																\lv
																																{
																																	\begin{proof}
																																	}
																																	Since $u,v\in V$ are votes in $(C,V)$, there are the corresponding vertices $u,v\in V(\dualg{\gamma})$. Since $\dualg{\gamma}$ is connected and by \Cref{lem:dualg_preserves_distances} preserves shortest paths, there is an $u$-$v$ path $P$ of length $\dswap{u}{v}$ in $\dualg{\gamma}$. The second vertex on $P$ corresponds to one of the permutations $v_i$, hence $v_i\in V(\dualg{\gamma})$.
																																\end{proof}
																															}
																															In other words, let $(C,V)$ is $2$-Euclidean and $u,v\in V$ two votes and let $v_1,\ldots, v_k$ be all the votes that differ from $u$ by a consecutive swap and $\dswap{v_i}{v}=\dswap{u}{v}-1$. Then $\region{\gamma}{v_i}\neq \emptyset$ for some~$i$.
																															%
																															
																															
																															\toappendix
																															{
																																\begin{example}\label{ex:implied_neighbors}
																																	We show how already the property proven in \Cref{lem:implied_neighbors} can help us with refuting an existence of $2$-Euclidean embedding for a particular election. Consider the election $(C,V)$ with $C=\{a,b,c,d\}$ and $V=\{abcd,dcba,bdca,cabd,dabc,cdab\}$. Intuition suggests that $(C,V)$ might be $2$-Euclidean, because since $|V|=6$ is much smaller than $18=\ub(|C|)$. It turns out, that it is not the case.
																																	
																																	Consider, for example, two votes $abcd$ and $cabd$ and apply \Cref{lem:implied_neighbors}. The only possible $v_i$ is the vote $acbd$. Hence, in any $2$-Euclidean embedding of $(C,V)$ we have $\region{\gamma}{acbd}\neq \emptyset$. We can proceed in a similar fashion as follows.
																																	\begin{center}
																																		\begin{tikzpicture}[scale=0.7]
																																			\tikzmath{\xoff = 2; \yoff=-2.5;}
																																			
																																			\node (abcd) at ($(0*\xoff,0*\yoff)$){$abcd$};
																																			\node (dcba) at ($(1*\xoff,0*\yoff)$){$dcba$};
																																			\node (bdca) at ($(2*\xoff,0*\yoff)$){$bdca$};
																																			\node (cabd) at ($(3*\xoff,0*\yoff)$){$cabd$};
																																			\node (dabc) at ($(4*\xoff,0*\yoff)$){$dabc$};
																																			\node (cdab) at ($(5*\xoff,0*\yoff)$){$cdab$};
																																			
																																			\node (acbd) at ($(0*\xoff,1*\yoff)$){$acbd$};
																																			\node (abdc) at ($(1*\xoff,1*\yoff)$){$abdc$};
																																			\node (bcda) at ($(2*\xoff,1*\yoff)$){$bcda$};
																																			\node (dbca) at ($(3*\xoff,1*\yoff)$){$dbca$};
																																			\node (cadb) at ($(4*\xoff,1*\yoff)$){$cadb$};
																																			\node (dcab) at ($(5*\xoff,1*\yoff)$){$dcab$};
																																			
																																			\node (bacd) at ($(0*\xoff,2*\yoff)$){$bacd$};
																																			\node (badc) at ($(1*\xoff,2*\yoff)$){$badc$};
																																			\node (adbc) at ($(2*\xoff,2*\yoff)$){$adbc$};
																																			\node (cbda) at ($(3*\xoff,2*\yoff)$){$cbda$};
																																			\node (dbac) at ($(4*\xoff,2*\yoff)$){$dbac$};
																																			\node (dacb) at ($(5*\xoff,2*\yoff)$){$dacb$};
																																			
																																			\node (acdb) at ($(0*\xoff,3*\yoff)$){$acdb$};
																																			\node (adcb) at ($(1*\xoff,3*\yoff)$){$adcb$};
																																			\node (bcad) at ($(2*\xoff,3*\yoff)$){$bcad$};
																																			\node (bdac) at ($(3*\xoff,3*\yoff)$){$bdac$};
																																			\node (cbad) at ($(4*\xoff,3*\yoff)$){$cbad$};
																																			\node (cdba) at ($(5*\xoff,3*\yoff)$){$cdba$};		
																																			
																																			
																																			
																																			\draw[->]($(abdc.north)+(0,0.3)$) --(abdc.north);
																																			\draw[->]($(acbd.north)+(0,0.3)$) --(acbd.north);
																																			\draw[->]($(acdb.north)+(0,0.3)$) --(acdb.north);
																																			\draw[->]($(adbc.north)+(0,0.3)$) --(adbc.north);
																																			\draw[->]($(adcb.north)+(0,0.3)$) --(adcb.north);
																																			\draw[->]($(bacd.north)+(0,0.3)$) --(bacd.north);
																																			\draw[->]($(badc.north)+(0,0.3)$) --(badc.north);
																																			\draw[->]($(bcad.north)+(0,0.3)$) --(bcad.north);
																																			\draw[->]($(bcda.north)+(0,0.3)$) --(bcda.north);
																																			\draw[->]($(bdac.north)+(0,0.3)$) --(bdac.north);
																																			\draw[->]($(cadb.north)+(0,0.3)$) --(cadb.north);
																																			\draw[->]($(cbad.north)+(0,0.3)$) --(cbad.north);
																																			\draw[->]($(cbda.north)+(0,0.3)$) --(cbda.north);
																																			\draw[->]($(cdba.north)+(0,0.3)$) --(cdba.north);
																																			\draw[->]($(dacb.north)+(0,0.3)$) --(dacb.north);
																																			\draw[->]($(dbac.north)+(0,0.3)$) --(dbac.north);
																																			\draw[->]($(dbca.north)+(0,0.3)$) --(dbca.north);
																																			\draw[->]($(dcab.north)+(0,0.3)$) --(dcab.north);

																																			\draw[out=270,in=90] (abcd) to ($(acbd.north) + (0,0.3)$);
																																			\draw[out=220,in=90,looseness=0.5] (cabd) to ($(acbd.north) + (0,0.3)$);
																																			\draw[out=270,in=90] (abcd) to ($(abdc.north) + (0,0.3)$);
																																			\draw[out=270,in=90,looseness=0.5] (dabc) to ($(abdc.north) + (0,0.3)$);
																																			\draw[out=270,in=90] (bdca) to ($(bcda.north) + (0,0.3)$);
																																			\draw[out=270,in=90] (cabd) to ($(bcda.north) + (0,0.3)$);
																																			\draw[out=270,in=90] (bdca) to ($(dbca.north) + (0,0.3)$);
																																			\draw[out=300,in=110,looseness=0.5] (dcba) to ($(dbca.north) + (0,0.3)$);
																																			\draw[out=270,in=90] (cabd) to ($(cadb.north) + (0,0.3)$);
																																			\draw[out=270,in=90] (cdab) to ($(cadb.north) + (0,0.3)$);
																																			\draw[out=270,in=90] (cdab) to ($(dcab.north) + (0,0.3)$);
																																			\draw[out=270,in=90] (dabc) to ($(dcab.north) + (0,0.3)$);
																																			
																																			\draw[out=300,in=70] (abcd) to ($(bacd.north) + (0,0.3)$);
																																			\draw[out=270,in=70] (bcda) to ($(bacd.north) + (0,0.3)$);
																																			\draw[out=270,in=90] (abdc) to ($(badc.north) + (0,0.3)$);
																																			\draw[out=270,in=90] (bdca) to ($(badc.north) + (0,0.3)$);
																																			\draw[out=270,in=90] (abdc) to ($(adbc.north) + (0,0.3)$);
																																			\draw[out=220,in=70] (dabc) to ($(adbc.north) + (0,0.3)$);
																																			\draw[out=270,in=90] (bcda) to ($(cbda.north) + (0,0.3)$);
																																			\draw[out=270,in=90,looseness=2] (cdab) to ($(cbda.north) + (0,0.3)$);
																																			\draw[out=300,in=70] (dabc) to ($(dbac.north) + (0,0.3)$);
																																			\draw[out=270,in=90] (dbca) to ($(dbac.north) + (0,0.3)$);
																																			\draw[out=270,in=90] (dabc) to ($(dacb.north) + (0,0.3)$);
																																			\draw[out=270,in=90] (dcab) to ($(dacb.north) + (0,0.3)$);
																																			
																																			\draw[out=240,in=110] (acbd) to ($(acdb.north) + (0,0.3)$);
																																			\draw[out=220,in=20] (dacb) to ($(acdb.north) + (0,0.3)$);
																																			\draw[out=270,in=90] (adbc) to ($(adcb.north) + (0,0.3)$);
																																			\draw[out=250,in=90,looseness=1.5] (cadb) to ($(adcb.north) + (0,0.3)$);
																																			\draw[out=270,in=90] (bacd) to ($(bcad.north) + (0,0.3)$);
																																			\draw[out=300,in=70] (bcda) to ($(bcad.north) + (0,0.3)$);
																																			\draw[out=300,in=90] (badc) to ($(bdac.north) + (0,0.3)$);
																																			\draw[out=270,in=90,looseness=0.5] (bdca) to ($(bdac.north) + (0,0.3)$);
																																			\draw[out=290,in=110] (cabd) to ($(cbad.north) + (0,0.3)$);
																																			\draw[out=270,in=110] (cbda) to ($(cbad.north) + (0,0.3)$);
																																			\draw[out=300,in=90] (cbda) to ($(cdba.north) + (0,0.3)$);
																																			\draw[out=300,in=70] (cdab) to ($(cdba.north) + (0,0.3)$);
																																		\end{tikzpicture}
																																	\end{center}
																																	
																																	The first row consists of votes of the original election $(C,V)$. Second to last row contains the new votes that must inevitably have their corresponding region nonempty in any $2$-Euclidean embedding of $(C,V)$. Any new vote $v_i$ is always connected to two votes $u,v$ from previous rows from which it was derived. Formally, $v_i$ is the unique vote such that $\dswap{u}{v_i}=1$ and $\dswap{v_i}{v}=\dswap{u}{v}-1$. Hence if $(C,V)$ was $2$-Euclidean, any $2$-Euclidean embedding of $(C,V)$ would have all $24$ nonempty regions. In other words, the embedding graph has at least $24$ vertices. But since $24>18=\ub(4)$, this is impossible. Hence $(C,V)$ is not $2$-Euclidean. \qed
																																\end{example}
																															}

																															\sv
																															{
																																\begin{observation}[$\star$]\label{obs:vote_ranking_candidate_first}
																																}
																																\lv
																																{
																																	\begin{observation}\label{obs:vote_ranking_candidate_first}
																																	}
																																	Let $C$ be a set of candidates and $\gamma$ an embedding of $C$. For each candidate $c\in C$ there is a vote $v\in\sym{C}$ such that $\pos{v}{c}=1$ and $\region{\gamma}{v}\neq \emptyset$.
																																\end{observation}
																																\toappendix
																																{
																																	\lv
																																	{
																																		\begin{proof}
																																		}
																																		\sv
																																		{
																																			\begin{proof}[Proof of \Cref{obs:vote_ranking_candidate_first}]
																																			}
																																			Consider the point $\gamma(c)$. Since $\gamma$ is injective, there is some $\varepsilon>0$ such that no other candidate is embedded into $\openballrc{\varepsilon}{\gamma(c)}$. Hence any voter embedded to $\openballrc{\frac{\varepsilon}{2}}{\gamma(c)}$ ranks $c$ first. However, there is always a region $\region{\gamma}{v}$ that intersects the open ball $\openballrc{\frac{\varepsilon}{2}}{\gamma(c)}$ and since all points in the region correspond to the same vote, the entire region corresponds to a vote that ranks $c$ first.
																																		\end{proof}
																																	}
																																	We now turn our attention to the inner and outer regions of the embedding. 

																																	\sv
																																	{
																																		\begin{observation}[$\star$]\label{obs:num_outer_regions}
																																		}
																																		\lv
																																		{
																																			\begin{observation}\label{obs:num_outer_regions}
																																			}
																																			Let $\gamma \colon C \to \mathbb{R}^2$ be a nice embedding of a candidate set. Then there are exactly $2\binom{|C|}{2}$ nonempty outer regions.
																																		\end{observation}
																																		\toappendix
																																		{
																																			\lv
																																			{
																																				\begin{proof}
																																				}
																																				\sv
																																				{
																																					\begin{proof}[Proof of \Cref{obs:num_outer_regions}]
																																					}
																																					There are exactly $\binom{|C|}{2}$ distinct bisectors of the form $\bisector{\gamma}{a}{b}$. Consider the open ball $B^\gamma$ containing all pairwise intersections of bisectors. Note that since no two bisectors are parallel (this is because $\gamma$ is nice), all bisectors intersect the boundary $\boundary{B^\gamma}$ exactly twice.\footnote{In fact, the only scenario where the circle $B^\gamma$ might not contain any part of a bisector is when all the bisectors are parallel.} Hence the circle $\boundary{B^\gamma}$ is partitioned into exactly $2\binom{|C|}{2}$ circular arcs. Since these are in one to one correspondence with the outer regions of $\gamma$, there are exactly $2\binom{|C|}{2}$ nonempty outer regions, as claimed. 
																																				\end{proof}
																																			}
																																			
																																			\sv
																																			{
																																				\begin{lemma}[$\star$]\label{lem:inner_regions_do_not_have_reverse}
																																				}
																																				\lv
																																				{
																																					\begin{lemma}\label{lem:inner_regions_do_not_have_reverse}
																																					}
																																					Let $\gamma$ be an embedding of a candidate set $C$ and let $v\in \sym{C}$. The following conditions are equivalent:
																																					\begin{enumerate}
																																						\item The region $\region{\gamma}{v}$ is outer.
																																						\item Both $\region{\gamma}{v}$ and $\region{\gamma}{v^R}$ are nonempty.
																																					\end{enumerate}
																																				\end{lemma}
																																				\toappendix
																																				{
																																					\lv
																																					{
																																						\begin{proof}
																																						}
																																						\sv
																																						{
																																							\begin{proof}[Proof of \Cref{lem:inner_regions_do_not_have_reverse}]
																																							}
																																							$1\Rightarrow 2$: Suppose that $\region{\gamma}{v}$ is a an outer region. Let $s$ be the circular arc on $\boundary{B^\gamma}$ that $\region{\gamma}{v}$ corresponds to. Since there is even number of circular arcs on $\boundary{B^\gamma}$ induced by the bisectors (precisely $2\binom{|C|}{2}$) there is an opposite arc $s'$ to $s$. Consider traveling along $\boundary{B^\gamma}$ in counterclockwise direction from $s$ to $s'$. Since $s'$ is opposite to $s$, we must pass exactly $\binom{|C|}{2}$ arcs (including $s'$ and not including $s$). In other words, we must cross exactly $\binom{|C|}{2}$ bisectors leaving $\boundary{B^\gamma}$. We claim that we in fact cross every bisector exactly once. For the sake of contradiction suppose that we crossed some bisector $\beta$ twice while traveling from $s$ to $s'$. Let $p_1,p_2$ be the points of intersection of $\beta$ with $\boundary{B^\gamma}$. Since all intersections of bisectors are inside $B^\gamma$, all intersections of $\beta$ with other bisectors are on the line segment $\linesegment{p_1}{p_2}$. Hence all bisectors leave the ball $B^\gamma$ at the circular arc given by $p_1$ and $p_2$ (in counterclockwise order). Hence when we crossed $\beta$ twice, we must have crossed every other bisector at least once. Hence we crossed at least $\binom{|C|}{2}-1+2$ bisectors. But that contradicts the fact that we traversed exactly $\binom{|C|}{2}$ circular arcs before hitting $s'$.
																																							
																																							Note that since we crossed every bisector exactly once this exactly corresponds to swapping all preferences of the voter $v$. Hence the region that corresponds to the arc $s'$ corresponds to the vote $v^R$ and hence $\region{\gamma}{v^R}\neq \emptyset$.
																																							
																																							$2\Rightarrow 1$: Suppose for the sake of contradiction that both $\region{\gamma}{v}$ and $\region{\gamma}{v^R}$ are nonempty and $\region{\gamma}{v}$ is inner. Consider arbitrary points $p\in \region{\gamma}{v},q\in\region{\gamma}{v^R}$. Notice that the line segment $\linesegment{p}{q}$ must cross each bisector exactly once because $v^R$ is the reverse vote of $v$. We now show that this is impossible. We show a slightly stronger claim: The infinite ray $\ray{p}{q}$ does not cross all bisectors.
																																							
																																							To simplify calculations, without loss of generality assume that $p$ is the origin, i.e., $p=(0,0)$ and the ray $\ray{p}{q}$ coincides with the positive direction of the $x$-axis. The closure of the region $\region{\gamma}{v}$ is a convex polygon, denote it by $P$. Let $p_1,\ldots,p_k$ be the vertices of $P$ given in counterclockwise order. Let $p_i$ be the index of a vertex with largest $y$-coordinate. Consider the line segment $\linesegment{p_{i}}{p_{i+1}}$ and let $\beta$ be the bisector that contains $\linesegment{p_{i}}{p_{i+1}}$. We show that $\ray{p}{q}$ does not intersect the bisector $\beta$.
																																							
																																							For the sake of contradiction suppose that this is not the case and that $\beta$ intersects $\ray{p}{q}$. Note that $y(p_i)>0$ since otherwise $p$ is not in the interior of $P$. The point of intersection of $\beta$ and $\ray{p}{q}$ is of the form $p'=(x,0)$ for some $x>0$. Notice that the angle $\alpha$ formed by the rays $\ray{p_i}{p_{i-1}}$ and $\ray{p_i}{p'}$ coincides with the interior angle of $P$ at $p_i$. Notice that $p\notin \alpha$. But by convexity of $P$ the entire polygon must lie in $\alpha$. A contradiction to the choice of $p$.	
																																						\end{proof}
																																					}

																																							{
																																								\subsection{The algorithm}\label{subsec:the_algorithm}
																																							}
																																							Let $(C,V)$ be $2$-Euclidean and let $\gamma$ be its nice $2$-Euclidean embedding. The properties derived in \Cref{subsec:properties_of_embeddings} impose constraints on both the embedding $\gamma$ and the graph $\dualg{\gamma}$. Conversely, if can show, given an election $(C,V)$, that all the constraints cannot be satisfied, then clearly $(C,V)$ is not $2$-Euclidean. For this purpose we formulate our problem as an integer linear program (ILP). The core variables of the program are binary variables $x_v$ and $\iota_v$ for each $v\in \sym{C}$ and the correspondence is as follows. If $(C,V)$ is $2$-Euclidean and $\gamma$ is a nice $2$-Euclidean embedding of $(C,V)$, then by setting
																																							$x_v=1\Leftrightarrow \region{\gamma}{v}\neq \emptyset$ and $\iota_v=1\Leftrightarrow (\region{\gamma}{v}\neq \emptyset \wedge \region{\gamma}{v} \text{ is an inner region})$, we obtain a feasible solution.																																																																		We emphasize that the property of the ILP is meant to be only unidirectional. Not every feasible solution to the ILP should correspond to some embedding $\gamma$ with $\region{\gamma}{v}\neq \emptyset$ if and only if $x_v=1$.
																																							
																																							In the following paragraphs we describe how to transform the properties from \Cref{subsec:properties_of_embeddings} to ILP constraints.

																																									\dtoappendix
																																									{
																																										\subsubsection{Basic constraints}
																																									}
																																									We begin with the obvious constraints:
																																									\begin{align}
																																										\restatableeq{\ilpcbasic}{x_v &= 1 & \forall v \in V}{ilpc:basic}.
																																									\end{align}
																																									Clearly for every $v\in V$ the region $\region{\gamma}{v}$ is nonempty, given that $(C,V)$ is $2$-Euclidean. Next, there cannot be more than $\ub(|C|)$ nonempty regions (see \Cref{thm:ubm}).
																																									\begin{align}
																																										\restatableeq{\ilpcubm}{\sum_{v\in \sym{C}}x_v &\leq \ub(|C|)}{ilpc:ubm}
																																									\end{align}
																																									
																																									\Cref{lem:implied_neighbors} gives a constraint that can be interpreted as the formula $(x_u\wedge x_v) \Rightarrow \bigvee_{i=1}^k x_{v_i}$, where $v_1,\ldots,v_k$ are all votes such that $\dswap{u}{v_i}=1$ and $\dswap{v_i}{v}=\dswap{u}{v}-1$. This can be encoded as follows:
																																									\begin{align}
																																										\restatableeq{\ilpcimpliedvotes}{(1-x_u) + (1-x_v) + \sum_{i=1}^k x_{v_i} &\geq 1 & \forall \{u,v\}\subseteq \sym{C}}{ilpc:implied_votes}.
																																									\end{align}
																																									\toappendix
																																									{
																																										\begin{lemma}\label{lem:con_impl_votes}
																																											
																																											The constraints (\ref{ilpc:implied_votes}) are correct.
																																										\end{lemma}

																																										\begin{proof}
																																											
																																											To show correctness, suppose that $(C,V)$ admits a $2$-Euclidean embedding $\gamma$ and let $u,v\in \sym{C}$ be arbitrary. If one of $\region{\gamma}{u},\region{\gamma}{v}$ is empty (without loss of generality $\region{\gamma}{u}\neq \emptyset$), then we can immediately see that by plugging in $x_u=0$, the left hand side is at least $1$, hence the constraint is satisfied. If both $\region{\gamma}{u},\region{\gamma}{v}$ are nonempty, then by \Cref{lem:implied_neighbors} there is $i\in[k]$ such that $\region{\gamma}{v_i}\neq \emptyset$. Hence we can set $x_{v_i}=1$ and thus the sum on the left hand side is again at least $1$ and hence the constraint is satisfied.
																																										\end{proof}
																																									}
																																									We proceed by introducing a constraint for \Cref{obs:vote_ranking_candidate_first}.
																																									\begin{align}
																																										\restatableeq{\ilpcvoterankingcandidatefirst}{\sum_{v\in \sym{C},\pos{v}{c}=1}x_v&\geq 1 & \forall c \in C}{ilpc:vote_ranking_candidate_first}
																																									\end{align}

																																									\toappendix
																																									{
																																										\begin{lemma}\label{lem:vote_rcf}
																																											
																																											The constraints (\ref{ilpc:vote_ranking_candidate_first}) are correct.
																																										\end{lemma}
																																										\lv
																																										{
																																											\begin{proof}
																																											}
																																											\sv
																																											{
																																												\begin{proof}[Proof of \Cref{lem:vote_rcf}]
																																												}
																																												Let $\gamma$ be a $2$-Euclidean embedding for $(C,V)$ and let $c\in C$ be any candidate. By \Cref{obs:vote_ranking_candidate_first} there is a vote $v\in \sym{C}$ such that $\pos{v}{c}=1$ and $\region{\gamma}{v}\neq \emptyset$. Hence we can set $x_v=1$ and thus satisfy this constraint (\ref{ilpc:vote_ranking_candidate_first}).
																																											\end{proof}
																																										}
																																										\dtoappendix
																																										{
																																											\subsubsection{Inner and Outer regions}
																																										}
																																										Recall that the variables $\iota_v$ correspond to nonempty regions that are also inner. Hence we add the obvious constraint $\iota_v \Rightarrow x_v$ which is captured by the following linear constraint:
																																										\begin{align}
																																											\restatableeq{\ilpcbasiciota}{(1-\iota_v) + x_v &\geq 1 & \forall v \in \sym{C}}{ilpc:basic_iota}
																																										\end{align}
																																										Next constraint is due to \Cref{obs:num_outer_regions}:
																																										\begin{align}
																																											\restatableeq{\ilpcnumouterregions}{\sum_{v\in \sym{C}}(x_v-\iota_v) &= 2\binom{|C|}{2}}{ilpc:num_outer_regions}
																																										\end{align}
																																										\sv
																																										{
																																										}
																																										\lv
																																										{
																																											
																																										}
																																										\toappendix
																																										{
																																											\begin{lemma}\label{lem:num_out_regions}
																																												
																																												Constraint (\ref{ilpc:num_outer_regions}) is correct.
																																											\end{lemma}
																																											\sv
																																											{
																																												\begin{proof}[Proof of \Cref{lem:num_out_regions}]
																																												}
																																												\lv
																																												{
																																													\begin{proof}
																																													}
																																													Let $(C,V)$ be $2$-Euclidean and $\gamma$ a nice $2$-Euclidean embedding of $(C,V)$. Set $\iota_v=1$ if $\region{\gamma}{v}$ is inner and $\iota_v = 0$ otherwise. Note that $x_v-\iota_v = 1$ if and only if $\region{\gamma}{v}$ is nonempty and is outer. Hence the sum on the left hand side counts the number of outer regions and right hand side is, by \Cref{obs:num_outer_regions}, also the number of outer regions of $\gamma$.
																																												\end{proof}
																																											}
																																											Next, we apply \Cref{lem:inner_regions_do_not_have_reverse}:
																																											\begin{align}
																																												\restatableeq{\ilpcnoreverse}{\iota_v + \iota_{v^R} &\leq 1 & \forall v \in \sym{C}}{ilpc:no_reverse}
																																											\end{align}

																																											\toappendix
																																											{
																																												\begin{lemma}\label{lem:no_rev}
																																													
																																													The constraints (\ref{ilpc:no_reverse}) are correct.
																																												\end{lemma}
																																												
																																												\begin{proof}
																																													
																																													Let $\gamma$ be a nice $2$-Euclidean embedding of $(C,V)$. By \Cref{lem:inner_regions_do_not_have_reverse} we know that if $\region{\gamma}{v}$ is inner, then $\region{\gamma}{v^R}=\emptyset$. Hence in any assignment where $\iota_v=1$ if and only if $\region{\gamma}{v}$ is inner we satisfy the constraint (\ref{ilpc:no_reverse}).
																																												\end{proof}
																																											}
																																											By using the second part of \Cref{lem:inner_regions_do_not_have_reverse} we can also see that if the region $\region{\gamma}{v}$ is outer, then the region $\region{\gamma}{v^R}$ is also nonempty and is outer. This fact can be captured by the formula $(x_v\wedge \neg\iota_v) \Rightarrow (x_{v^R}\wedge \neg \iota_{v^R})$ which is translated to the following linear constraints:
																																											\begin{align}
																																												\restatableeq{\ilpcouterregionhasreverseone}{(1-x_v)	+ \iota_v +  x_{v^R} &\geq 1 & \forall v\in\sym{C}}{ilpc:outer_region_has_reverse1}\\
																																												\restatableeq{\ilpcouterregionhasreversetwo}{(1-x_v) + \iota_v +  (1-\iota_{v^R}) &\geq 1 & \forall v \in \sym{C}}{ilpc:outer_region_has_reverse2}
																																											\end{align}
																																											
																																											Since minimum degree of $\dualg{\gamma}$ is at least $2$, we have:
																																											\begin{align}
																																												\restatableeq{\ilpcdegreestwoandthree}{\sum_{v_i\in N_{G(\sym{C})}(v)} x_{v_i} &\geq 2 x_v + \iota_v & \forall v \in \sym{C}}{ilpc:degrees_2_and_3}
																																											\end{align}

																																											\toappendix
																																											{
																																												\begin{lemma}\label{lem:degrees_2_and_3}
																																													
																																													Constraints (\ref{ilpc:degrees_2_and_3}) are correct.
																																												\end{lemma}
																																												
																																												\begin{proof}
																																													
																																													Let $\gamma$ be a nice $2$-Euclidean embedding of $(C,V)$ and consider the embedding graph $\dualg{\gamma}$. If $\region{\gamma}{v}=\emptyset$, then $x_v=0$ and $\iota_v=0$ and since all variables are binary the constraint is trivially satisfied. If $\region{\gamma}{v}\neq \emptyset$, then $x_v=1$. If $\region{\gamma}{v}$ is outer region, i.e., $\iota_v=0$, then the left hand side is required to be at lest $2$. But note that the sum on the left hand side is just the degree of the vertex $v$ in $\dualg{\gamma}$. By \Cref{lem:dualg_preserves_distances} the minimum degree of $\dualg{\gamma}$ is at least $2$, hence this constraint is satisfied. It remains to show that if $\region{\gamma}{v}$ is an inner region (i.e., $\iota_v=1$), then the vertex $v$ has degree at least $3$ in $\dualg{\gamma}$. Note that the inner regions are convex polygons and since convex polygons have at least three sides and these sides correspond to edges of $\primalg{\gamma}$, there are at least $3$ neighboring regions of $\region{\gamma}{v}$, i.e., at least $3$ neighbors of $v$ in $\dualg{\gamma}$.
																																												\end{proof}
																																											}
																																											The description of the neighbors can be made more precise. Every outer region has exactly two neighboring regions that are also outer. For a vote $v$, consider the sum $S(v)=\sum_{v_i\in N_{G(\sym{C})}}(x_{v_i}-\iota_{v_i})$. $S(v)$ counts the number of neighbors of $v$ that are also outer regions. Our aim is to express the following. If $x_v=1$ and $\iota_v = 0$, then $S(v)=2$.
																																											This can be encoded by the following:
																																											\begin{align}
																																												\restatableeq{\ilpcsvone}{S(v) &\geq 2 \cdot (x_v - \iota_v) & \forall v \in \sym{C} }{ilpc:svone}\\
																																												\restatableeq{\ilpcsvtwo}{S(v)&\leq 2 + \iota_v \cdot |C|&\forall v \in \sym{C}}{ilpc:svtwo}
																																											\end{align}

																																													\begin{theorem}[$\star$]
																																														Let $(C,V)$ be an election. Suppose that $(C,V)$ is $2$-Euclidean and let $\gamma$ be a nice $2$-Euclidean embedding of $(C,V)$. Then by plugging in $x_v=1$ if and only if $\region{\gamma}{v}\neq \emptyset$ and $\iota_v$ if and only if $x_v=1$ and the region $\region{\gamma}{v}$ is an inner region, then all the constraints (\ref{ilpc:basic}) up to (\ref{ilpc:svtwo}) are satisfied.
																																													\end{theorem}
																																													
																																													Note that the formulation of the ILP has $|C|!$ variables and may seem impractical at first sight. We propose a lazy approach applicable in practice which is described in the long version.

																																													\toappendix
																																													{
																																														\sv
																																														{
																																															\section{Missing material from section QCP approach}
																																														}
																																													}
																																													\section{QCP approach}\label{sec:qcp}					
																																																						In this section, we propose the Quadratically Constrained Program (QCP) approach, which arises from the definition of $2$-Euclidean elections. By definition, an election $(C, V)$ is $2$-Euclidean if and only if there exist points $\gamma(t) \in \mathbb{R}^2$ for each $t \in C \cup V$ such that, whenever $a \succ_v b$, we have $\ell_2(\gamma(v), \gamma(a)) < \ell_2(\gamma(v), \gamma(b))$.
																																													
																																													\dtoappendix{
																																														The main challenge here is that QCP programs and in particular most practical QCP solvers require the inequalities to be non-strict. However, this is manageable by ensuring that the distances of voters from the bisectors are sufficiently large. 
																																													}
																																													\sv
																																													{
																																														This is dealt with by adding error terms into the inequalities. More details can be found in the long version.
																																													}
																																													\toappendix{
																																														For this purpose, we define \emph{$\varepsilon$-scattered $2$-Euclidean embedding}. A $2$-Euclidean embedding $\gamma$ is \emph{$\varepsilon$-scattered} if for any $a,b\in C,v\in V$ we have
																																														\[
																																														a\succ_v b \Rightarrow \ell_2(\gamma(v),\gamma(a)) + \varepsilon \leq \ell_2(\gamma(v),\gamma(b))
																																														\]

																																														\sv
																																														{
																																															\begin{lemma}[$\star$]\label{lem:scattered_embedding}
																																															}
																																															\lv
																																															{
																																																\begin{lemma}\label{lem:scattered_embedding}
																																																}
																																																Let $\varepsilon > 0$. An election is $2$-Euclidean if and only if it admits an $\varepsilon$-scattered $2$-Euclidean embedding.
																																															\end{lemma}
																																															\toappendix
																																															{
																																																\sv
																																																{
																																																	\begin{proof}[Proof of \Cref{lem:scattered_embedding}]
																																																	}
																																																	\lv
																																																	{
																																																		\begin{proof}
																																																		}
																																																		Clearly any $\varepsilon$-scattered $2$-Euclidean embedding is also a $2$-Euclidean embedding for $(C,V)$. On the other hand, suppose that $(C,V)$ is $2$-Euclidean and let $\gamma$ be any $2$-Euclidean embedding of $(C,V)$. Let $\delta = \min_{a\succ_v b}\{\ell_2(\gamma(v),\gamma(b))-\ell_2(\gamma(v),\gamma(a))\}$. Define $\gamma'$ to be $\gamma'(x)=\gamma(\frac{\varepsilon}{\delta}x)$. Clearly $\gamma'$ is a $2$-Euclidean embedding of $(C,V)$ and we have
																																																		\begin{align*}
																																																			\ell_2(\gamma'(v),\gamma'(a))+\varepsilon&\leq\ell_2(\gamma'(v),\gamma'(b)) & \Leftrightarrow\\
																																																			\frac{\varepsilon}{\delta}\ell_2(\gamma(v),\gamma(a))+\varepsilon&\leq\frac{\varepsilon}{\delta}\ell_2(\gamma(v),\gamma(b)) & \Leftrightarrow \\
																																																			\varepsilon&\leq \frac{\varepsilon}{\delta}\left(\ell_2(\gamma(v),\gamma(b)) - \ell_2(\gamma(b),\gamma(a))\right) &
																																																		\end{align*}
																																																		the last inequality holds because $\delta \leq \ell_2(\gamma(v),\gamma(b))-\ell_2(\gamma(v),\gamma(a))$.
																																																	\end{proof}
																																																}
																																																For the purposes of QCP we wish to remove the square root from the inequality and instead consider the constraints of the form
																																																\begin{align}\label{eq:equivalent_cond}
																																																	\ell_2(\gamma(v),\gamma(a))^2+ \varepsilon^*\leq \ell_2(\gamma(v),\gamma(b))^2 & & \forall a,b\in C, v \in V: a\succ_v b
																																																\end{align}
																																																for some $\varepsilon^*>0$.
																																																We again show that for any $\varepsilon^*>0$ the condition (\ref{eq:equivalent_cond}) can be replaced in the definition of $2$-Euclidean embedding.
																																																\sv{
																																																	\begin{lemma}[$\star$]\label{lem:sq_scattered_embedding}
																																																	}
																																																	\lv{
																																																		\begin{lemma}\label{lem:sq_scattered_embedding}
																																																		}
																																																		Let $\varepsilon^* > 0$. An election $(C,V)$ is $2$-Euclidean if and only if there exists an embedding $\gamma'\colon C \cup V$ such that for any $a,b\in C,v\in V$ we have
																																																		\[
																																																		a\succ_v b \Rightarrow \ell_2(\gamma'(v),\gamma'(a))^2+\varepsilon^* \leq \ell_2(\gamma'(v),\gamma'(b))^2
																																																		\]
																																																	\end{lemma}
																																																	\toappendix
																																																	{
																																																		\sv
																																																		{
																																																			\begin{proof}[Proof of \Cref{lem:sq_scattered_embedding}]
																																																			}
																																																			\lv
																																																			{
																																																				\begin{proof}
																																																				}
																																																				On one hand, any embedding $\gamma'$ satisfying $a\succ_v b \Rightarrow \ell_2(\gamma'(v),\gamma'(a))^2+\varepsilon^* \leq \ell_2(\gamma'(v),\gamma'(b))^2$ is a $2$-Euclidean embedding of $(C,V)$.
																																																				
																																																				On the other hand, suppose that $(C,V)$ is $2$-Euclidean and pick $\varepsilon$-scattered embedding $\gamma'$ for $\varepsilon = \sqrt{\varepsilon^*}$.
																																																				
																																																				Denote by $A=\ell_2(\gamma'(v),\gamma'(a))$ and $B=\ell_2(\gamma'(v),\gamma'(b))$. By assumption, we have
																																																				\begin{align*}
																																																					a\succ_v b & &\Rightarrow \\
																																																					A+\sqrt{\varepsilon^*}&\leq B&\Leftrightarrow \\
																																																					A^2+2A\sqrt{\varepsilon^*}+\varepsilon^*&\leq B^2 &\Rightarrow \\
																																																					A^2 + \varepsilon^* \leq B^2
																																																				\end{align*}
																																																			\end{proof}
																																																		}
																																																		We are now ready to formulate the QCP. The variables are the corresponding $x$- and $y$-coordinates of the points $\gamma(v),\gamma(c)$ for $v\in V$ and $c\in C$. To simplify notation, let $x_t = x(\gamma(t)),y_t = y(\gamma(t))$ be the $x$- and $y$- coordinates of the point $\gamma(t)$ for $t\in C \cup V$. We obtain the following contraints:
																																																		\begin{align*}
																																																			\ell_2(\gamma(v),\gamma(a))^2+ \varepsilon^*&\leq \ell_2(\gamma(v),\gamma(b))^2 & \forall a,b\in C, v \in V: a\succ_v b
																																																		\end{align*}
																																																		or equivalently
																																																		\begin{align}\label{eq:qcp_init}
																																																			(x_v-x_a)^2 + (y_v-y_a)^2 + \varepsilon^* \leq \\ \leq (x_v-x_b)^2 + (y_v-y_b)^2 & & \forall a,b\in C, v \in V: a\succ_v b
																																																		\end{align}
																																																		and to avoid floating point arithmetic precision issues we choose $\varepsilon^*=1$.
																																																		
																																																	} 
																																																	
																																																	The key improvement over the QCP approach of Escoffier, Spanjaard and Tydrichová~\cite{EscoffierST23} is the addition of the following constraints:																																													\begin{align*}
																																																		-x_{\max}\leq x_t\leq x_{\max} & & \forall t \in C \cup V \\
																																																		-y_{\max}\leq y_t\leq y_{\max} & & \forall t \in C \cup V
																																																	\end{align*}
																																																	In other words, we bound the range of the coordinates of the embedding to the rectangle $2x_{\max}\times 2y_{\max}$. In practice, if the bounding box is too small, the QCP solver can quickly recognize infeasibility. In such case we keep doubling the bounding box until a solution is found.
																																																	
																																																	\toappendix
																																																	{
																																																		\sv
																																																		{
																																																			\section{Implementation details}\label{app:sec:implementation}
																																																		}
																																																	}
																																																	\section{Experiments}\label{sec:experiments}
																																																	\toappendix{
																																																		\subsection{Convex Hull}
																																																		From the description of the convex hull it is not yet clear how it should be used to refute $2$-Euclideanness of a particular election. Trying all possible subsets of voters might be too expensive. For practical usage, it seems to be enough if one considers only subsets of size $4$.
																																																		Note that the structure of the controversity graph is rich when the number of voters is small (see \Cref{example:refuting_using_controversity_graph,example:refuting_using_controversity_graph_2}). 
                                                                                                                                                                                                        We implement two variants, one that tests all subelections $(C,V')$ of $(C,V)$, where $|V'|=4,V'\subseteq V$ and the a second tests all subelections $(C,V')$ of $(C,V)$ where $V'\subseteq V$. 
                                                                                                                                                                                                        Every PrefLib instance solved by the more general approach is also solved by the simplified and faster implementation testing only 4 voters.
																																																	}
																																																	\toappendix{
																																																		\subsection{ILP}
																																																		First obstacle for practical usage of the ILP is the fact that $|\sym{C}|=|C|!$. Having a variable for each possible vote quickly becomes untractable. To avoid this issue, we adopt the following lazy approach. We start by creating variables for $v\in V$ (in order to satisfy constraints (\ref{ilpc:basic})). The idea is that the variables that do not exist yet attain the value of $0$. We then repeat the following process. We run the solver. If the solver outputs NO, then we know the answer and can safely say that the input election is not $2$-Euclidean. Otherwise, the existing variables with value $1$ induce a potential embedding graph which is also divided to inner and outer vertices (according to values of variables $\iota_v$). We then manually check whether this graph satisfies all the remaining constraints and any violated constraint is added to the solver. Then we either repeat the process \texttt{Unknown} if maximum number of iterations is reached (in our implementation we use $20$ iterations) or no new constraints are added. In this way we avoid writing long constraints and keep the size of the program manageable. Note that the constrains with the sum $\sum_{v\in\sym{C}}$ and $\leq$ will now range over the existing variables. The constraints with $\sum_{v\in \sym{C}}$ and $\geq$ cannot be written unless all the variables exist (e.g., constraint (\ref{ilpc:vote_ranking_candidate_first})). Last thing that needs to be modified is the constraint (\ref{ilpc:num_outer_regions}), we write $\leq$ instead of $=$ because some of the variables may not exist. Finally, to be able to track progress of the solver, we minimize the number of vertices in the embedding graph, i.e. $\min \sum_{v}x_v$.
																																																		
																																																		Similar to the approach with the convex hull, we do not input the entire election into the solver. Here, however, we work with subsets of candidates (rather than voters), and the solver can return NO for any subset. In our implementation, we consider all possible subelections induced by sets of at least five candidates, progressively increasing in size--starting with all subsets of size five, then moving onto subsets of size six, and so on until we either hit a subset of candidates for which the above procedure says NO or a global timeout is reached. Smaller sets are not relevant because any election with at most three candidates is automatically $2$-Euclidean and $2$-Euclidean elections with four candidates are fully characterized by three maximal profiles (\cite{KamiyaTT11}). 
																																																		
																																																		To speed up the process, for each subset, we first perform a quick check to determine if the sub-election might be $2$-Euclidean before running the ILP machinery. This quick check involves applying reduction rules, followed by a three-second run of the second phase of the EST algorithm.
																																																	}
																																																	\toappendix{
																																																		\subsection{QCP}
																																																		Our implementation uses the values of $x_{\max}$ and $y_{\max}$ as follows. We start with $x_{\max}=y_{\max}=100$ and we run the QCP solver with timeout of $10$ seconds. If the solver outputs yes, we have an answer. Otherwise we multiply both $x_{\max}$ and $y_{\max}$ by $10$ and double the timeout and run the solver again and repeat the previous step until a global timeout is reached.
																																																	}
																																																	
																																																	\begin{table}[ht]
																																																		\resizebox{!}{0.12\textheight}{
																																																			\begin{tabular}{lrrr}
																																																				\toprule
																																																				dataset & unknown & non-$2$-Euclidean & $2$-Euclidean \\
																																																				\midrule
																																																				00006 - skate & 0({\color{darkgreen}-16}) & 18({\color{darkgreen}+14}) & 2({\color{darkgreen}+2}) \\
																																																				00015 - cleanweb & 6({\color{darkgreen}-35}) & 72({\color{darkgreen}+34}) & 1({\color{darkgreen}+1}) \\
																																																				00032 - education & 0({\color{darkgreen}-1}) & 1({\color{darkgreen}+1}) & 0 \\
																																																				00042 - boxing & 29({\color{darkgreen}-60}) & 11({\color{darkgreen}+10}) & 56({\color{darkgreen}+50}) \\
																																																				00043 - cycling & 0({\color{darkgreen}-2}) & 121 & 2({\color{darkgreen}+2}) \\
																																																				00044 - tabletennis & 0({\color{darkgreen}-2}) & 36({\color{darkgreen}+2}) & 0 \\
																																																				00048 - spotifycountry & 0({\color{darkgreen}-4}) & 640({\color{darkgreen}+3}) & 2({\color{darkgreen}+1}) \\
																																																				00049 - mylaps & 3({\color{darkgreen}-47}) & 584({\color{darkgreen}+38}) & 23({\color{darkgreen}+9}) \\
																																																				00052 - f1seasons & 2({\color{darkgreen}-6}) & 61({\color{darkgreen}+2}) & 4({\color{darkgreen}+4}) \\
																																																				00054 - weeksport & 0({\color{darkgreen}-4}) & 951({\color{darkgreen}+4}) & 0 \\
																																																				00056 - seasonsport & 20({\color{darkgreen}-106}) & 3766({\color{darkgreen}+104}) & 193({\color{darkgreen}+2}) \\
																																																				00062 - orderaltexpe\footnote{The dataset 00062 - orderaltexpe was not present in PrefLib at the time of the paper~\cite{EscoffierST23} and it is also not counted towards the 212 newly resolved no-instances.} & 0 & 2({\color{darkgreen}+2}) & 0 \\
																																																				other & 0 & 1037 & 100 \\
																																																				\bottomrule
																																																				total & 60({\color{darkgreen}-283}) & 7300({\color{darkgreen}+212}) & 383({\color{darkgreen}+71}) \\
																																																				\bottomrule
																																																			\end{tabular}
																																																		}
																																																		\caption{Solved instances per PrefLib dataset by using all solvers combined (including \texttt{EST}). The green numbers indicate the improvement over the \texttt{EST} algorithm alone. 
																																																						}\label{tab:solved_per_dataset}
																																																					\vspace{-2.9em}
																																																	\end{table}
																																																	
																																																	\begin{figure}
																																												
							\centering
							\includegraphics[width=0.47\textwidth]{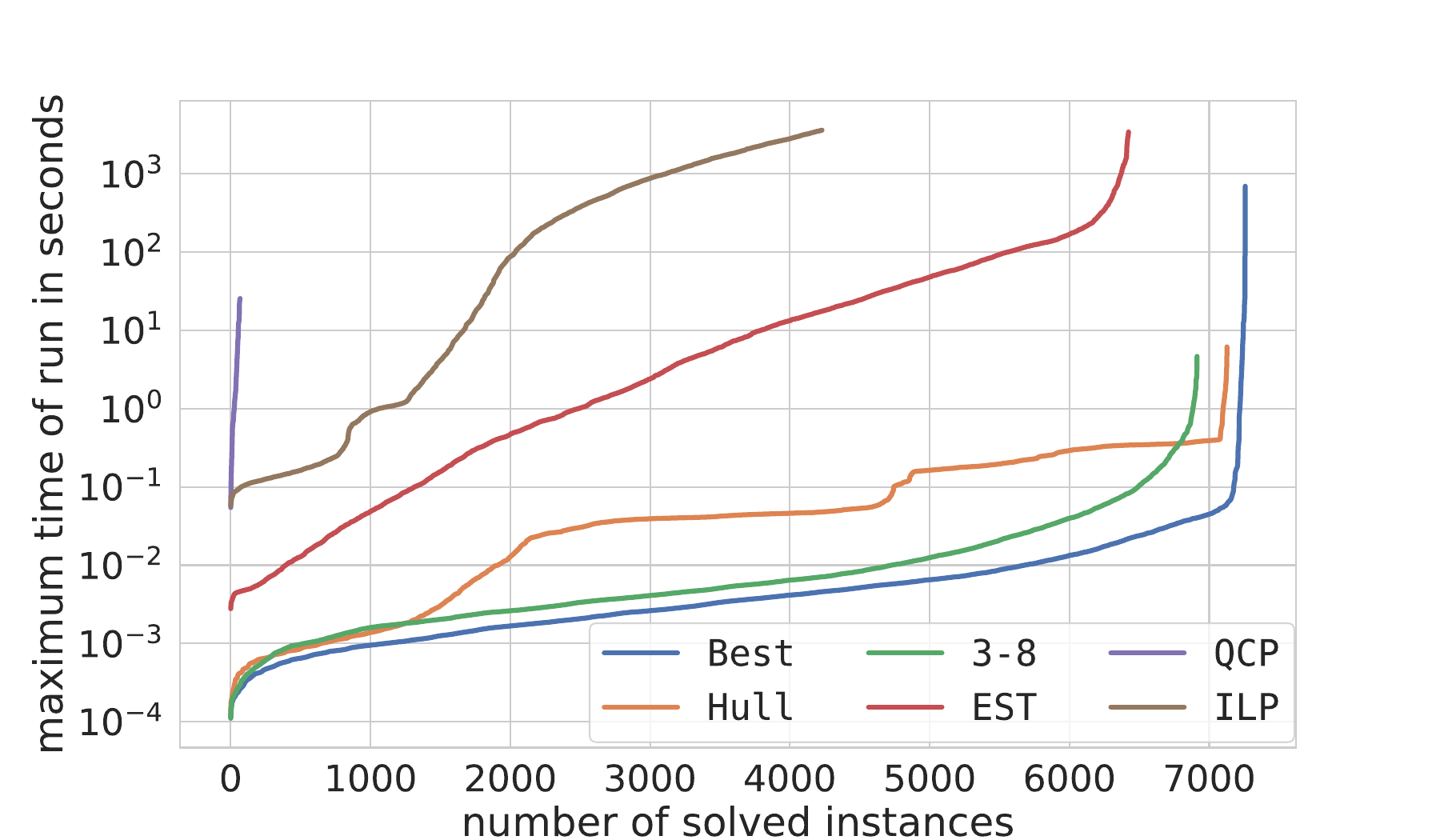}
							\caption{The plot shows how many instances each solver solves within given time limit. \texttt{Best} represents ideal parallel execution, stopping all solvers once the first finishes.}
						\label{fig:solved_per_time}						\vspace{-1em}																									\end{figure}
																																																	
																																																	\begin{table}[ht]
																																																		\centering
																																
																																																			\begin{tabular}{lrrrr}
                                                                                                                                                                                                                \toprule
                                                                                                                                                                                                                 & solved & fastest & max time (s) & median time (s) \\
                                                                                                                                                                                                                \midrule
                                                                                                                                                                     										   \texttt{3-8} & 6914 & 5213 & 4.653348 & 0.004962 \\
                                                                                                                                                                                                                \texttt{EST} & 6424 & 96 & 3422.947932 & 3.857373 \\
                                                                                                                                                                                                                \texttt{Hull} & 7128 & 1908 & 6.150966 & 0.042965 \\
                                                                                                                                                                                                                \texttt{ILP} & 4231 & 3 & 3598.020537 & 140.411390 \\
                                                                                                                                                                                                                \texttt{QCP} & 68 & 39 & 25.507809 & 1.578934 \\
                                                                                                                                                                                                                \texttt{trivial} & 424 & 424 & - & - \\
                                                                                                                                                                                                                \bottomrule
                                                                                                                                                                                                            \end{tabular}
																																
																																																		\caption{Running times of individual algorithms. First column indicates the number of instances solved by the algorithm, second column indicates for how many instances this approach was the fastest and the last two columns show the maximum and median times in seconds.}\label{tab:running_times}
		\vspace{-3em}
			\end{table}
																																															
																																																\begin{figure}
																																																\centering
																																																\includegraphics[width=0.47\textwidth]{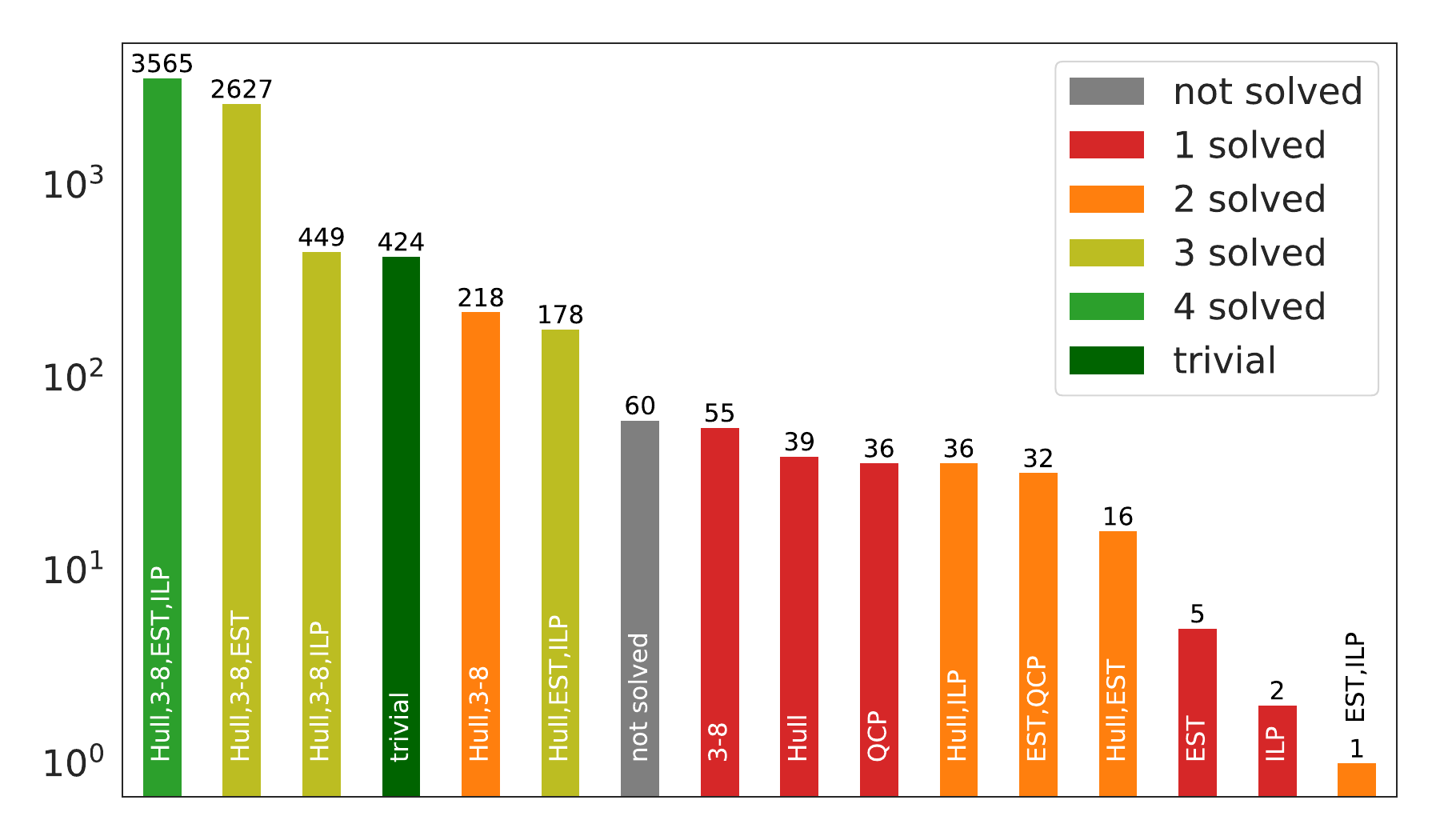}
							\vspace{-2.5em}
																																																\caption{
																																																	Number of instances that were solved by exactly the solvers in the corresponding set.}
										
										\vspace{-2em}																																						\label{fig:powersets}
	\end{figure}

We implemented\footnote{The source code and related material is available at \href{https://zenodo.org/records/18541197}{\texttt{zenodo.org/records/18541197}}.} the algorithms \texttt{3-8} (based on \Cref{subsec:3_8_pattern}), \texttt{Hull} (based on \Cref{sec:forbidden_substructures}), \texttt{ILP} (based on \Cref{sec:implied_regions}), and \texttt{QCP} (based on \Cref{sec:qcp}).
																																																As the source of the experiments, we took all the complete strict order datasets (\texttt{.soc} files) from PrefLib~\cite{Preflib} including datasets from \cite{bennet2007netflix,mao2013better,green1972applied,boehmer2022collecting,bindas2025serial}. As of October 2025 there were $7743$ such instances in PrefLib.
																																																The experiments were carried out on a machine with \texttt{AMD EPYC 7742 64-Core Processor}.
																																																For each instance each heuristic (including the \texttt{EST} algorithm) was run on a single thread with time limit of 1 hour and memory limit of 32 GiB. This memory limit was not enough for 235 instances (in all cases this was the case of the \texttt{EST} algorithm).
																																																For both ILP and QCP programs we used the Gurobi Optimizer~\cite{gurobi} (note that for QCP it is the same as in the previous paper~\cite{EscoffierST23}). \Cref{tab:solved_per_dataset} shows the absolute number of solved instances divided among datasets of PrefLib by combining all approaches. \Cref{tab:running_times} shows the running times of each approach separately. We denote by \texttt{trivial} the solver which (trivially) resolves exactly the instances with $|C|\leq 3$ or $|V|\leq 2$ or $|V|\leq 3 \wedge |C|\leq 7$ all of which are $2$-Euclidean.
																																																
																																																\paragraph*{Running times}
																																																We evaluated the performance of a solver by counting the number of instances that are solvable within a given time (see \Cref{fig:solved_per_time}). Out of a total of $7683$ instances we solved, only $39$ of them took longer than one second to solve by the fastest solver. This corresponds to the \texttt{Best} curve in \Cref{fig:solved_per_time}. Note that only the running time for successfully solved instances is displayed in \Cref{fig:solved_per_time}, rather than the time spent on instances that could not be decided. 
							Experiments suggest, that \texttt{3-8} approach is the fastest for most instances. The only contender is the algorithm \texttt{Hull} as it solves $214$ more instances and it also solves most of the instances under one second. The next fastest solver is the \texttt{EST} algorithm, although it is slower by few orders of magnitude. In few instances, it almost exceeds the time limit. However, based on the shape of its performance curve, it appears unlikely that many additional instances would be solved with \texttt{EST} algorithm in a reasonable time. In contrast, \texttt{ILP} is even slower by a few more orders of magnitude. However, its performance curve remains flat at the end, suggesting that if given additional time, the solver could potentially solve more instances. Lastly, \texttt{QCP}, which can only recognize yes-instances, solves only a small number of instances (since most PrefLib instances are not $2$-Euclidean).

																																																\paragraph*{Reduction rules}
																																																The reduction rules were successfully applied to 802 instances.
																																																Together both reduction rules were used 979 times. \Cref{rr:rr1plusplus} was applied 728 times and \Cref{rr:rr2} was applied 251 times.
																																																In total, 1729 candidates across all instances were removed, averaging more than 2 candidates per instances where any rule was applied.
																																																In most cases, the number of removed candidates was low -- only 1 candidate was removed in 487 instances and only 2 candidates in 155 instances. 
																																																However, in $70$ instances, at least $5$ candidates were removed; in $20$ instances $10$ or more candidates were removed, with the highest observed removal being $27$ candidates in $2$ instances.
																																																
																																																\paragraph*{Evaluating Solver Necessity}
																																																The number of solved instances and the running time alone do not fully capture the effectiveness of an approach. It is also crucial to consider whether a given approach can solve instances that other solvers cannot. In \Cref{fig:powersets}, for a given set $S$ of solvers, we illustrate how many instances are resolved by all solvers in $S$. From the initial columns, we can observe that majority of the instances (over 7000) are either trivial or solvable by at least three different approaches.
																																																However, each solver is justified with at least one instance that no other solver can handle. 
                                                                                                                                                                                                Notably, all of the instances except 8 can be solved with the combination of \texttt{3-8}, \texttt{Hull} and \texttt{QCP}.
                                                                                                                                                                                                For yes-instances, \texttt{QCP} recognizes all known cases except two solved only by \texttt{EST}.\footnote{In \Cref{fig:powersets} there are $5$ instances solved only by \texttt{EST}, but $3$ of them are no-instances.}
																																																
																																																\paragraph*{When to use which solver}
																																																In some approaches we try all possible subsets of either votes (\texttt{Hull}, partially \texttt{3-8}) or candidates (\texttt{ILP}, partially \texttt{EST}). 
                                                                                                                                                                                                This theoretically identifies when a particular algorithm should be the best. 
                                                                                                                                                                                                In practice it is less clear which solver will be the fastest for a given number of candidates and voters. 
                                                                                                                                                                                                The \texttt{EST} solver is limited by design to a maximum of $9$ candidates. 
                                                                                                                                                                                                By first reducing the instance size, we were able to solve $17$ additional instances. 
                                                                                                                                                                                                However, the \texttt{EST} algorithm has not solved any yes-instance with more than $8$ votes (after reductions). 
                                                                                                                                                                                                The \texttt{QCP} approach, on the other hand, was able to handle slightly larger instances, solving instances with up to $24$ candidates and $12$ votes. 
                                                                                                                                                                                                For small no-instances with at most $30$ candidates, \texttt{Hull} outperforms \texttt{3-8}, but as instance sizes increase, \texttt{3-8} dominates due to its lower complexity ($O(|V|^3|C|^2)$, see \Cref{lem:38patterrec}).
																																																
																																																\paragraph*{Unsolved instances}
																																																There are still $60$ unsolved instances in the PrefLib dataset. Of these, $16$ instances contain $3$ votes and at least $100$ candidates. The other $44$ unsolved instances are relatively small, they have at most $17$ votes and at most $23$ candidates. Minimal unsolved instances as a $(|C|, |V|)$ pair are of sizes $(12,3)$, $(6,6)$ and $(8,5)$.

																																																\FloatBarrier
																			
																																																\section{Conclusion}\label{sec:conclusion}
																																																In this work we discussed a practical approach to recognizing $2$-Euclidean preferences. We evaluated the performance of our approach on real-world datasets from PrefLib. We observed that most practical instances fail to be $2$-Euclidean due to a very simple and quickly recognizable forbidden configuration contained in the instance. In particular, it is the forbidden pattern based on the convex hull of the voters or the $3$-$8$ no-instance. Additionally, we designed an ILP that considers many more combinatorial properties of $2$-Euclidean elections and attempts to falsify them. Towards recognizing yes-instances, we improved the QCP formulation that arises from the definition of $2$-Euclidean elections. All our algorithms are also supported by reduction rules which remove trivial parts of the instance and reduce its size significantly.
																																																Using these techniques we were able to classify more than $82\%$ of previously unclassified instances of PrefLib.
																																																
																																																First open question that arises is how our approaches generalize to $d$-Euclidean elections for $d>2$. For example, \Cref{rr:rr1plus} generalizes well to $d>2$ by simply replacing the size of the tail blocks by $d+1$. On the other hand, it is not obvious how to generalize the convex hull approach to higher dimensions. How do controversial sets of voters of size greater than $2$ contribute to the higher-dimensional convex hull? Does this characterization help significantly to provide a no-certificate for $d>2$? How does the ILP approach generalize to higher dimensions?
																																															
																																																Finally, it remains to classify the remaining $60$ instances of PrefLib whether they are $2$-Euclidean.


\begin{acks}
This work was co-funded by the European Union under the project Robotics and advanced industrial production (reg. no. CZ.02.01.01/\\00/22\_008/0004590).
Michal Dvořák, Jan Pokorný and Martin Slávik further acknowledge the support of the Grant Agency of the Czech Technical University in Prague, grant No. SGS23/205/OHK3/3T/18.
\end{acks}

\balance

\bibliographystyle{ACM-Reference-Format}
																																																\bibliography{references}
																																																
																																																\clearpage
																																																\appendix
																																																\section*{Appendix}
																																																\appendixText

																																															\end{document}